\pdfoutput=1
\documentclass[11pt]{article}

\usepackage[T1]{fontenc}
\usepackage{lmodern}
\usepackage[margin=1in]{geometry}
\usepackage{microtype}

\usepackage{amsthm,amsmath,amsfonts,amssymb,mathtools}
\usepackage[authoryear,round]{natbib}
\bibpunct{(}{)}{;}{a}{ }{,}
\usepackage{graphicx}
\usepackage[font=small,labelfont=bf]{caption}
\usepackage{booktabs}
\usepackage{bm}
\usepackage{array}
\usepackage{enumitem}
\usepackage{placeins}
\usepackage[hidelinks]{hyperref}
\allowdisplaybreaks
\hypersetup{
  pdftitle={Adaptive Ridge-Regularized Hotelling Change-Point Tests for Functional Data},
  pdfauthor={Ping Zhao and Long Feng}
}
\numberwithin{equation}{section}

\theoremstyle{plain}
\newtheorem{theorem}{Theorem}[section]

\newtheorem{corollary}{Corollary}[section]
\newtheorem{lemma}{Lemma}[section]

\theoremstyle{definition}
\newtheorem{assumption}{Assumption}[section]

\newtheorem{algorithm}{Algorithm}[section]

\newcommand{\R}{\mathbb R}
\newcommand{\E}{\mathbb E}
\newcommand{\Pp}{\mathbb P}
\newcommand{\cH}{\mathcal H}
\newcommand{\Hs}{\mathcal H_{\varsigma}}
\newcommand{\cG}{\mathcal G}
\newcommand{\cK}{\mathcal K}
\newcommand{\cS}{\mathcal S}
\newcommand{\cL}{\mathcal L}

\newcommand{\tr}{\operatorname{tr}}
\newcommand{\Cov}{\operatorname{Cov}}
\newcommand{\Var}{\operatorname{Var}}
\newcommand{\diag}{\operatorname{diag}}

\newcommand{\op}{\mathrm{op}}
\newcommand{\HS}{\mathrm{HS}}
\newcommand{\argmax}{\operatorname*{arg\,max}}
\newcommand{\ind}{\mathbf 1}
\newcommand{\dto}{\mathrel{\Rightarrow}}
\newcommand{\pto}{\mathrel{\to_{\Pp}}}
\newcommand{\bxi}{\bm\xi}
\newcommand{\beps}{\bm\varepsilon}
\newcommand{\bdelta}{\bm\delta}
\newcommand{\bDelta}{\bm\Delta}
\newcommand{\bmu}{\bm\mu}
\newcommand{\bS}{\bm S}
\newcommand{\bY}{\bm Y}
\newcommand{\bOmega}{\bm\Omega}
\newcommand{\bGamma}{\bm\Gamma}
\newcommand{\bI}{\mathbf I}
\newcommand{\Id}{\operatorname{Id}}
\newcommand{\norm}[1]{\left\lVert #1\right\rVert}

\newcommand{\ip}[2]{\left\langle #1,#2\right\rangle}

\begin{document}

\title{Adaptive Ridge-Regularized Hotelling Change-Point Tests for Functional Data}
\author{%
\begin{tabular}{c}
Ping Zhao$^{1}$ \qquad Long Feng$^{2}$\\[0.6em]
{\small $^{1}$School of Mathematical Science, Tiangong University, Tianjin, China}\\
{\small $^{2}$School of Statistics and Data Science, Nankai University, Tianjin, China}
\end{tabular}%
}
\date{}

\maketitle

\begin{abstract}
We propose a unified ridge-regularized Hotelling framework for detecting and locating mean changes in functional time series. A growing basis expansion converts the functional observations into high-dimensional score vectors. Their long-run covariance is estimated by an edge-corrected difference-based procedure. Ridge regularization stabilizes inference under spectral decay. An explicit local-power formula shows that the power-maximizing ridge depends on the unknown spectral orientation of the change. We therefore combine a family of ridge CUSUM statistics by a Cauchy transform and calibrate the aggregate directly from their joint weighted-bridge limit. For multiple changes, we embed local maximum-ridge statistics in a wild binary segmentation procedure, followed by local refinement. Under mild conditions, we establish the validity, local power, consistency, and localization properties of the proposed tests. In the multiple-change setting, the procedure consistently recovers the number of changes and uniformly estimates their locations. The framework accommodates weak dependence and non-Gaussian functional errors. Simulations and two empirical applications demonstrate the favorable finite-sample performance of the proposed methods.
\end{abstract}

\begin{center}
\end{center}
\medskip

\section{Introduction}

Functional data analysis provides a statistical framework for observations whose natural units are curves, trajectories, images, or other function-valued objects. Such data arise when modern instruments record a process densely over time, wavelength, space, or another continuum, and treating the resulting measurements as unrelated coordinates can obscure smoothness, geometry, and scientifically meaningful modes of variation. The subject now has a mature methodological and theoretical foundation, with applications ranging from climatology and neuroscience to finance, biomechanics, and transportation; see the monographs of \citet{RamsaySilverman2005}, \citet{FerratyVieu2006}, and \citet{HorvathKokoszka2012}. Foundational contributions to this development include functional principal component recovery from sparse longitudinal measurements \citep{YaoMullerWang2005Sparse}, functional linear regression for longitudinal trajectories \citep{YaoMullerWang2005Regression}, and functional variance processes \citep{MullerStadtmullerYao2006}. These methods established score reconstruction, covariance estimation, and regression tools that remain central when functional observations are discretized, noisy, or incompletely observed.

Structural stability is a central issue for ordered functional observations. A change in the mean curve may indicate a climatic regime shift, a change in a physiological response, a modification of an industrial system, or an alteration in market dynamics. Early work focused on independent curves and developed tests and break-date estimators for a single mean change; representative contributions are \citet{BerkesEtAl2009} and \citet{AueEtAl2009}. The theory was subsequently extended to weakly dependent functional time series. \citet{ZhangEtAl2011} used self-normalization to avoid direct long-run covariance estimation, whereas \citet{AstonKirch2012} treated both at-most-one-change and epidemic alternatives under dependence. Projection-based procedures were refined by allowing the number of functional principal components to increase with sample size in \citet{FremdtEtAl2014}, and by constructing repeated or change-aligned principal components in \citet{Torgovitski2015}.

A complementary strand avoids an initial fixed-dimensional reduction. \citet{SharipovEtAl2016} developed a sequential block bootstrap for Hilbert-space CUSUM processes, and \citet{AueRiceSonmez2018} established fully functional detection and dating theory under a broad weak-dependence condition. More recent work has enlarged the class of alternatives, improved robustness, and clarified how the geometry of the function space affects detection. \citet{WegnerWendler2024} proposed a robust procedure based on Hilbert-space-valued U-statistics and a dependent wild bootstrap, and \citet{Bastian2025} studied the consequences of the norm used to aggregate a functional CUSUM. \citet{BonieceHorvathTrapani2025} used empirical energy distance to detect changes extending beyond the mean, while \citet{HorvathRiceVanderDoes2026} developed tests based on empirical characteristic functionals. Recent work has also addressed nonstationary functional time series observed with partial measurement error; see \citet{BaiHuWu2026}.

The literature has also moved beyond a single break.  Wild binary
segmentation (WBS), introduced by \citet{Fryzlewicz2014}, uses randomly drawn
subintervals to isolate individual changes before recursive segmentation.
For weakly dependent functional sequences, \citet{RiceZhang2022} established
consistency of binary segmentation, while \citet{HarrisLiTucker2022} and
\citet{ChenChiouHuang2023} proposed scalable and greedy segmentation
procedures.  \citet{BastianBasuDette2024} developed multiple-change
methodology motivated by biomechanical fatigue data, and earlier functional
applications include traffic segmentation \citep{ChiouEtAl2019}.  Related
change-point problems
concern structural monitoring in functional linear models
\citep{AueEtAl2014}, changes in spatiotemporal mean functions
\citep{GromenkoEtAl2017}, and changes in covariance operators or eigensystems
\citep{StoehrEtAl2021,DetteKutta2021,HorvathRiceZhao2022,JiaoFrostigOmbao2023}.
These developments collectively show that functional change-point analysis is
now a broad area, but they also reveal a persistent tension between dimension
reduction, spectral estimation, and the use of the full functional
information.

This paper takes a different route. We expand each observed curve in a deterministic basis and retain a growing number of numerical coefficients, thereby converting the functional sequence into a high-dimensional vector sequence. This construction preserves the computational advantages of multivariate CUSUM methods while avoiding the need to estimate a small collection of leading eigenfunctions before testing. It also permits a transparent treatment of discretization: the implemented procedure is built from numerical basis coefficients computed on the observation grid, and the discrepancy from the ideal functional projections is quantified separately.

The resulting problem is connected to the rapidly developing literature on changes in high-dimensional means. Coordinatewise maximum procedures are effective for sparse changes and have been studied by \citet{Jirak2015}; sparse projection and minimax detection theory were developed by \citet{WangSamworth2018}, \citet{EnikeevaHarchaoui2019}, and \citet{LiuGaoSamworth2021}. Finite-sample and resampling-based inference was studied by \citet{YuChen2021}, while \citet{WangEtAl2022} used self-normalization to accommodate nuisance dependence. Adaptive procedures combining dense and sparse evidence were proposed by \citet{ZhangWangShao2022} and \citet{WangFeng2023}; extensions to temporally dependent high-dimensional time series were considered by \citet{WangLiuFeng2025}. Multiple-change segmentation in high dimension has also been developed through sparsified binary segmentation and related methods; see, for example, \citet{ChoFryzlewicz2015} and \citet{Cho2016}.

Much of the existing theory for dense quadratic and adaptive high-dimensional mean-change procedures is developed under bounded-spectrum or trace-dispersion conditions that prevent a fixed number of directions from dominating the statistic; representative formulations appear in \citet{Jirak2015}, \citet{WangFeng2023}, and \citet{LiXu2026}. This regime is natural for generic vector data, but it does not describe a growing basis expansion of a square-integrable random function. Finite second moment in the function space makes the coefficient variances summable in every orthonormal basis and the covariance-operator eigenvalues summable, so both sequences necessarily tend to zero. Consequently, a few low-frequency directions can carry a nonvanishing proportion of the total variation, whereas many high-frequency directions have much smaller variances. Direct inversion is unstable, coordinatewise standardization can amplify estimation noise in the tail of the expansion, and an unregularized quadratic statistic can be dominated by a small number of directions. This spectral heterogeneity is not a secondary technical detail; it is the defining feature that separates the present problem from conventional high-dimensional change-point models.

Ridge-regularized Hotelling methods provide a natural response. For high-dimensional mean testing, ridge regularization was introduced by \citet{ChenEtAl2011} and developed into an adaptable Hotelling procedure by \citet{LiEtAl2020}. In change-point analysis, \citet{LiXu2026} constructed ridge-regularized CUSUM statistics for independent high-dimensional observations, and \citet{ZhaoZhouFeng2026} studied aggregation over a deterministic ridge grid. Those results operate in a proportional random-matrix regime with a diffuse spectral bulk. Their fixed-ridge null limit is a Gaussian process because many eigendirections contribute at a comparable scale. The functional regime considered here is fundamentally different. The long-run covariance operator is trace class and its spectrum decays, so the contribution of the leading directions does not vanish. The limiting process is therefore not Gaussian in general; it is a centered and standardized weighted sum of independent squared Brownian bridges. Moreover, under heterogeneous spectra no fixed ridge is guaranteed to be uniformly preferable over unknown mean-change directions. Small ridge values accentuate low-variance eigendirections, whereas large values move the statistic toward an unweighted $L^2$ aggregation; the local-power maximizing value therefore depends on the unknown spectral coordinates of the change. We handle this tuning problem by combining the marginal ridge $p$-values through a Cauchy transform and by using common bridge paths to retain cross-ridge dependence. Establishing the weighted-bridge limit, its joint version over a ridge grid, the alternative-specific power formula, and a feasible calibration based on estimated spectral weights requires an operator-level argument rather than a direct application of the existing random-matrix central limit theory.

Temporal dependence creates a second difficulty. The lag-zero covariance or the ordinary sample covariance used for independent observations does not normalize a functional CUSUM under serial dependence; the relevant object is the long-run covariance operator. Classical heteroskedasticity-and-autocorrelation-consistent estimators were developed by \citet{NeweyWest1987} and \citet{Andrews1991}, and functional long-run covariance estimation was studied by \citet{HormannKokoszka2010}, \citet{PanaretosTavakoli2013}, and \citet{RiceShang2017}. Conventional lag-window estimators, however, center the data by a global mean and may be substantially contaminated when that mean changes. Difference-based estimators were designed precisely to separate serial dependence from nonconstant mean structure; see \citet{Chan2022MAC} and \citet{Chan2022DB}. Their high-dimensional extension was developed by \citet{LiuSongFeng2026}. We adapt this principle to functional basis scores and use the actual number of available lagged products in every denominator. This edge correction removes a deterministic finite-sample attenuation that otherwise propagates through the estimated spectrum, the ridge inverse, and the weighted-bridge calibration. Because the compact kernel used in the implementation need not produce a positive semidefinite matrix in finite samples, we additionally analyze the positive spectral part before ridge inversion.

The paper makes several contributions.  First, it develops a basis-expansion
formulation of weakly dependent functional mean-change inference that retains
a growing number of coefficients and uses a trace-scaled ridge to accommodate
the intrinsic spectral decay.  Second, it establishes consistency rates for
an edge-corrected difference-based estimator of the functional long-run
covariance under an Aue--Rice--S\"onmez type Bernoulli-shift condition.  The
errors need not be Gaussian and no functional autoregressive, moving-average,
or linear-process representation is imposed; the required covariance and
fourth-order bounds are derived from the primitive approximation condition.
Third, it derives the non-Gaussian weighted Brownian-bridge-square null limit
and its joint finite-grid extension, and it provides direct spectral critical
values obtained from the explicit limit process without bootstrapping or
permuting the original curves.  Fourth, it establishes local-alternative
limits, derives an explicit power formula and an alternative-specific oracle
ridge, and uses this characterization to justify Cauchy aggregation over a
deterministic grid.  It also proves consistency and single-change localization
while accounting for numerical basis integration and for contamination of the
long-run covariance estimator under a mean shift.  Fifth, it develops a
multiple-change procedure
that reuses the same global difference-based covariance estimate in local
max-ridge CUSUMs within WBS.  Random intervals isolate individual changes,
and recursive selection followed by local refinement consistently recovers
the number and locations of well-separated changes under the same
finite-moment weak-dependence framework.  These results combine weak
approximation for functional partial sums, non-Gaussian control of a quadratic
difference-based covariance estimator, trace-class spectral perturbation
theory, and uniform ridge-resolvent arguments.

The remainder of the paper is organized as follows.
Section~\ref{sec:method} introduces the functional model, numerical scores, the
edge-corrected difference-based long-run covariance estimator, and the ridge
scan.  Section~\ref{sec:theory} presents the single-change null, calibration,
local-power, consistency, and localization results.
Section~\ref{sec:multiple} develops multiple-change estimation and its theory.
Section~\ref{sec:implementation} reports implementation details and the
simulation results.  Section~\ref{sec:applications} presents two empirical
applications, and Section~\ref{sec:conclusion} concludes.  Auxiliary results
and all proofs are collected in the Appendix.

\section{Model and methodology}
\label{sec:method}

Throughout, $\norm{\cdot}$ denotes the Euclidean norm for vectors and the
norm of the ambient Hilbert space for functions.  For matrices and operators,
$\norm{\cdot}_{\op}$, $\norm{\cdot}_{\HS}$, and $\norm{\cdot}_1$ denote the
operator, Hilbert--Schmidt, and trace norms, respectively.  We write
$\ip{\cdot}{\cdot}$ for an inner product, $\tr(A)$ for the trace of $A$,
$\Id$ for an identity operator, with a subscript when needed, and
$\ind(\cdot)$ for the indicator function.  The symbols $\dto$ and $\pto$ denote convergence in
distribution and in probability.  The notation $O_p(\cdot)$, $o_p(\cdot)$,
and $\asymp$ has its usual meaning.  Bold letters generally denote finite-dimensional
vectors and matrices.  The letter $C$ denotes a finite positive constant that
may change from line to line.

Let $\cH=L^2[0,1]$ with inner product
$\ip{f}{g}=\int_0^1f(t)g(t)\,dt$ and norm
$\norm{f}=\ip{f}{f}^{1/2}$.  We observe
$X_1,\ldots,X_n\in\cH$ from the single-change model
\begin{equation}
 X_i=\mu+\delta_n\ind(i>\tau_n)+e_i,
 \qquad i=1,\ldots,n,
 \label{eq:functional-change-model}
\end{equation}
where $\mu\in\cH$ is the baseline mean, $\delta_n\in\cH$ is a
possibly $n$-dependent jump, $\tau_n\in\{1,\ldots,n-1\}$ is the
unknown change location, and $\{e_i:i\in\mathbb Z\}$ is a centered
strictly stationary functional error process.  Its primitive construction and
regularity are specified in Assumption~\ref{ass:ARS-dependence}.  We test
\begin{equation*}
 H_0:\delta_n=0
 \qquad\text{against}\qquad
 H_1:\delta_n\ne0\text{ for an unknown }\tau_n.
\end{equation*}
Under $H_0$, $\tau_n$ is immaterial.  The dependence of $\delta_n$ on $n$
allows both local and fixed alternatives.

\subsection{Discrete functional observations and numerical Fourier scores}

Let $\{\phi_j:j\ge1\}$ be the real Fourier basis. The curves are observed without an additional measurement-error layer on the common regular grid
\begin{equation*}
 0=t_1<t_2<\cdots<t_N=1,
 \qquad t_l=\frac{l-1}{N-1}.
\end{equation*}
Let $\omega_{1,N}=\omega_{N,N}=\frac{1}{2(N-1)}$ and $\omega_{l,N}=(N-1)^{-1}$ for $2\le l\le N-1$.  Define the canonical Fourier isometry
\begin{equation*}
 \iota_p:\R^p\longrightarrow\cH,
 \qquad
 \iota_p x=\sum_{j=1}^p x_j\phi_j,
 \qquad
 \iota_p^*f=\{\ip{f}{\phi_j}\}_{j=1}^p.
\end{equation*}
The trapezoidal operator and its first-$p$ Fourier reconstruction are
\begin{equation*}
 Q_N(f)=\sum_{l=1}^N\omega_{l,N}f(t_l),
 \qquad
 \mathcal J_{p,N}f=\iota_p\{Q_N(f\phi_1),\ldots,Q_N(f\phi_p)\}^{\top}.
\end{equation*}
The implemented score is
\begin{equation*}
 \widehat\xi_{ij}=Q_N(X_i\phi_j),
 \qquad 1\le j\le p,
\end{equation*}
and $\widehat\bxi_i=(\widehat\xi_{i1},\ldots,\widehat\xi_{ip})^{\top}$. Define
\begin{equation*}
 \bmu_{p,N}=\{Q_N(\mu\phi_j)\}_{j=1}^p,
 \qquad
 \bdelta_{p,N}=\{Q_N(\delta_n\phi_j)\}_{j=1}^p,
 \qquad
 \beps_{i,p,N}=\{Q_N(e_i\phi_j)\}_{j=1}^p.
\end{equation*}
Then the numerical score vectors satisfy the exact model
\begin{equation*}
 \widehat\bxi_i
 =\bmu_{p,N}+\bdelta_{p,N}\ind(i>\tau_n)+\beps_{i,p,N}.
\end{equation*}
For a $p\times p$ matrix $A$, write $A^{\cH}=\iota_pA\iota_p^*$ for its extension by zero outside $\operatorname{span}(\phi_1,\ldots,\phi_p)$.  For every self-adjoint nonnegative $p\times p$ matrix $A$, every $c>0$, and every $x\in\R^p$,
\begin{equation}
 x^{\top}(A+c\bI_p)^{-1}x
 =\ip{\iota_px}{(A^{\cH}+c\Id_{\cH})^{-1}\iota_px}.
 \label{eq:matrix-operator-identification}
\end{equation}
All operator expressions below use this canonical embedding.  Whenever a finite score vector appears in a Hilbert-space inner product, it denotes its image under $\iota_p$.  In particular, convergence statements involving $\bdelta_{p,N}$ mean convergence of $\iota_p\bdelta_{p,N}$ in $\cH$.  No statistic below uses the unobserved integrals $\int_0^1X_i(t)\phi_j(t)\,dt$.

For $u\in\mathbb Z$, define the projected lag covariance matrix, its
canonical operator embedding, and their long-run sums by
\begin{equation*}
 \bGamma_{p,N}(u)=\Cov(\beps_{u,p,N},\beps_{0,p,N}),
 \qquad
 \Gamma_{p,N}^{\cH}(u)=\iota_p\bGamma_{p,N}(u)\iota_p^*,
\end{equation*}
\begin{equation*}
 \bOmega_{p,N}=\sum_{u=-\infty}^{\infty}\bGamma_{p,N}(u),
 \qquad
 \Omega_{p,N}^{\cH}=\sum_{u=-\infty}^{\infty}\Gamma_{p,N}^{\cH}(u)
 =\iota_p\bOmega_{p,N}\iota_p^*.
\end{equation*}
The nonzero eigenvalues, trace, Hilbert--Schmidt norm and trace norm of
$\Omega_{p,N}^{\cH}$ are exactly those of $\bOmega_{p,N}$; the same
convention applies to every canonically embedded finite matrix below.
Our lag convention places the later observation in the first argument, so
$\bGamma_{p,N}(u)=\E(\beps_{u,p,N}\beps_{0,p,N}^{\top})$ and
$\bGamma_{p,N}(-u)=\bGamma_{p,N}(u)^{\top}$.  This convention agrees
with the sample product in \eqref{eq:edge-autocov}.  The matrix $\bOmega_{p,N}$ governs normalized score partial sums and CUSUM
contrasts, while $\Omega_{p,N}^{\cH}$ governs their canonical embeddings in
$\cH$.  Neither object is, in general, equal to its lag-zero counterpart.

\subsection{Edge-corrected difference-based long-run covariance}

Fix a difference order $m$ and coefficients $c_0,\ldots,c_m$ satisfying
\begin{equation}
 \sum_{q=0}^{m}c_q=0,
 \qquad
 \sum_{q=0}^{m}c_q^2=1.
 \label{eq:difference-constraints}
\end{equation}
Let $h=h_n$ be the difference spacing and put $N_D=n-mh$. Reindex the available differences as
\begin{equation}
 \bY_{i_{\rm D}}=\sum_{q=0}^{m}c_q\widehat\bxi_{mh+i_{\rm D}-qh},
 \qquad i_{\rm D}=1,\ldots,N_D.
 \label{eq:DB-difference}
\end{equation}
For $0\le r<\ell=\ell_n$, define
\begin{equation}
 \widehat\bGamma_{\rm DB}(r)
 =\frac{1}{N_D-r}
 \sum_{i_{\rm D}=r+1}^{N_D}\bY_{i_{\rm D}}\bY_{i_{\rm D}-r}^{\top},
 \qquad
 \widehat\bGamma_{\rm DB}(-r)=\widehat\bGamma_{\rm DB}(r)^{\top}.
 \label{eq:edge-autocov}
\end{equation}
Let $K$ be an even compactly supported kernel with $K(0)=1$. The raw difference-based long-run covariance estimator is
\begin{equation}
 \widehat\bOmega_n^{\rm DB}
 =\sum_{|r|<\ell}K(r/\ell)\widehat\bGamma_{\rm DB}(r).
 \label{eq:DB-raw}
\end{equation}
For the quadratic compact kernel $K(x)=(1-|x|^2)_+$, the finite-sample matrix in \eqref{eq:DB-raw} need not be positive semidefinite. Write
\begin{equation*}
 \widehat\bOmega_n^{\rm DB}
 =\widehat{\bm O}_n\diag(\widehat\lambda_{1,n}^{\rm raw},\ldots,
 \widehat\lambda_{p,n}^{\rm raw})\widehat{\bm O}_n^{\top}
\end{equation*}
and define its positive spectral part
\begin{equation}
 \widehat\bOmega_n^{+}
 =\widehat{\bm O}_n\diag\{(\widehat\lambda_{1,n}^{\rm raw})_+,
 \ldots,(\widehat\lambda_{p,n}^{\rm raw})_+\}\widehat{\bm O}_n^{\top}.
 \label{eq:positive-part}
\end{equation}
Let $\widehat\lambda_{1,n}\ge\cdots\ge\widehat\lambda_{p,n}\ge0$
denote the ordered eigenvalues of $\widehat\bOmega_n^{+}$, and define
\begin{equation*}
 \widehat\Omega_{n,+}^{\cH}
 =\iota_p\widehat\bOmega_n^+\iota_p^*,
 \qquad
 \widehat{\mathfrak t}_n=\tr(\widehat\bOmega_n^{+})
 =\tr(\widehat\Omega_{n,+}^{\cH})
 =\sum_{j=1}^p\widehat\lambda_{j,n}.
\end{equation*}
The statistics and conditional bridge laws below are evaluated on
$\mathcal E_n=\{\widehat{\mathfrak t}_n>0\}$.  On $\mathcal E_n^c$ no matrix
inverse is formed.  We set $D_{n,a}(k)=T_{n,a}=0$ and every associated
marginal and combined $p$-value equal to one.  Once the local quantities are
introduced in Section~\ref{sec:multiple}, we also set
$D_{I,a}(k)=\cS_I(k)=\cS_I^{\max}=0$ for every admissible interval and
candidate, take the smallest candidate as a formal local maximizer, and let
Algorithm~\ref{alg:DB-RWBS} return no changes.  The conditional reference law
may be defined arbitrarily.  These conventions make both the single- and
multiple-change procedures measurable for every sample and have no asymptotic
effect, because the trace-consistency conclusions of
Lemma~\ref{prop:weight-consistency} imply $\Pp(\mathcal E_n)\to1$ under the
conditions of the corresponding theorems.

\subsection{Full-CUSUM ridge statistic}

For $1\le k\le n-1$, let
\begin{equation*}
 \widehat\bDelta_k
 =\frac{1}{n-k}\sum_{i=k+1}^n\widehat\bxi_i
 -\frac{1}{k}\sum_{i=1}^k\widehat\bxi_i,
 \qquad
 N_k=\frac{k(n-k)}{n}.
\end{equation*}
For $0\le t\le1$, define the tied-down score partial-sum process
\begin{equation*}
 \widehat\bS_n(t)
 =\frac{1}{\sqrt n}
 \left\{\sum_{i=1}^{\lfloor nt\rfloor}\widehat\bxi_i
 -\frac{\lfloor nt\rfloor}{n}\sum_{i=1}^n\widehat\bxi_i\right\}.
\end{equation*}
At $t_k=k/n$,
\begin{equation*}
 \sqrt{N_k}\widehat\bDelta_k
 =-\frac{\widehat\bS_n(t_k)}{\sqrt{t_k(1-t_k)}}.
\end{equation*}
The scan always uses all $n$ observations; differences are used only to estimate the long-run covariance.

Let $\cG=\{a_1,\ldots,a_{L_{\cG}}\}\subset(0,\infty)$ be a fixed deterministic ridge grid.  On $\mathcal E_n$, for $a\in\cG$, set
\begin{equation}
 \begin{aligned}
 \widehat\rho_{n,a}&=a\widehat{\mathfrak t}_n,\\
 \widehat{\bm R}_{n,a}
 &=(\widehat\bOmega_n^{+}+a\widehat{\mathfrak t}_n\bI_p)^{-1},\\
 \widehat{\mathcal R}_{n,a}
 &=(\widehat\Omega_{n,+}^{\cH}
     +a\widehat{\mathfrak t}_n \Id_{\cH})^{-1}.
 \end{aligned}
 \label{eq:ridge-inverse}
\end{equation}
The first inverse in \eqref{eq:ridge-inverse} is the implemented $p\times p$
matrix, whereas the second is the full Hilbert-space operator used in the
proofs.  Equation~\eqref{eq:matrix-operator-identification} gives, for every
$x\in\R^p$,
\begin{equation*}
 x^\top\widehat{\bm R}_{n,a}x
 =\ip{\iota_px}{\widehat{\mathcal R}_{n,a}\iota_px}.
\end{equation*}
Thus the two objects yield exactly the same statistic on the numerical Fourier
subspace, while $\widehat{\mathcal R}_{n,a}$ acts as
$(a\widehat{\mathfrak t}_n)^{-1}\Id_{\cH}$ on its orthogonal complement.

The unstandardized path is
\begin{equation*}
 V_{n,a}(k)
 =N_k\widehat\bDelta_k^{\top}\widehat{\bm R}_{n,a}\widehat\bDelta_k
 =\frac{\widehat\bS_n(t_k)^{\top}\widehat{\bm R}_{n,a}
 \widehat\bS_n(t_k)}{t_k(1-t_k)}.
\end{equation*}
Define the empirical ridge weights and spectral moments
\begin{equation}
 \widehat w_{j,n}(a)
 =\frac{\widehat\lambda_{j,n}}
 {\widehat\lambda_{j,n}+a\widehat{\mathfrak t}_n},
 \qquad
 \widehat A_{1,n}(a)=\sum_{j=1}^p\widehat w_{j,n}(a),
 \qquad
 \widehat A_{2,n}(a)=\sum_{j=1}^p\widehat w_{j,n}(a)^2.
 \label{eq:empirical-weights}
\end{equation}
The centered and standardized path and scan statistic are
\begin{equation}
 D_{n,a}(k)
 =\frac{V_{n,a}(k)-\widehat A_{1,n}(a)}
 {\{2\widehat A_{2,n}(a)\}^{1/2}},
 \qquad
 T_{n,a}=\max_{k\in\cK_{n,\epsilon}}D_{n,a}(k),
 \label{eq:D-statistic}
\end{equation}
where
\begin{equation*}
 \cK_{n,\epsilon}
 =\{\lceil n\epsilon\rceil,\ldots,\lfloor n(1-\epsilon)\rfloor\},
 \qquad 0<\epsilon<1/2.
\end{equation*}
The analytic pointwise centering and variance estimators are
\begin{equation*}
 \widehat A_{1,n}(a)
 \qquad\text{and}\qquad
 2\widehat A_{2,n}(a),
\end{equation*}
respectively. Theorem~\ref{thm:null-limit} shows that they consistently estimate the limiting null mean and variance of $V_{n,a}(k)$; they are not asserted to be exact finite-sample conditional moments.

\subsection{Direct calibration from the limiting bridge process}

Let $B_1,B_2,\ldots$ be independent standard Brownian bridges and put
\begin{equation*}
 U_j(t)=\frac{B_j(t)}{\sqrt{t(1-t)}}.
\end{equation*}
Conditional on the estimated weights, define
\begin{equation}
 \widehat{\mathcal D}_{n,a}^{*}(t)
 =\frac{\sum_{j=1}^p\widehat w_{j,n}(a)\{U_j(t)^2-1\}}
 {\{2\widehat A_{2,n}(a)\}^{1/2}},
 \qquad
 \widehat{\mathcal T}_{n,a}^{*}
 =\max_{k\in\cK_{n,\epsilon}}
 \widehat{\mathcal D}_{n,a}^{*}(t_k).
 \label{eq:conditional-limit}
\end{equation}
The conditional $(1-\alpha)$ quantile of $\widehat{\mathcal T}_{n,a}^{*}$ is the fixed-ridge critical value. This calculation generates only standard Gaussian variables from an explicit limiting law. It does not resample the observed curves and is therefore neither a bootstrap nor a permutation procedure.

Different ridge values apply different spectral filters, and the power-optimal
filter depends on the unknown change direction.  Selecting one ridge from the
same CUSUM path would therefore require an additional selection-and-calibration
step.  We instead retain the entire finite grid and combine its marginal
evidence.

The same bridge paths are shared across all $a\in\cG$, thereby preserving
cross-ridge dependence.  Let
\begin{equation*}
 \widehat{\bm w}_n(a)=\{\widehat w_{j,n}(a):j\ge1\},
 \qquad
 \widehat{\mathcal W}_n=\{\widehat{\bm w}_n(a_v):1\le v\le L_{\cG}\},
\end{equation*}
where $\widehat w_{j,n}(a)=0$ for $j>p$.  The collection
$\widehat{\mathcal W}_n$ contains all estimated ridge weights.  Writing
$\bm B=(B_1,B_2,\ldots)$, throughout the calibration arguments
$\Pp_{\bm B}$ and $\E_{\bm B}$ denote probability and expectation with
respect to all simulated bridge variables, conditionally on the observed data
and hence on $\widehat{\mathcal W}_n$.
On $\mathcal E_n$, define
\begin{equation*}
 \widehat F_{n,a_v}(x)
 =\Pp_{\bm B}\{\widehat{\mathcal T}_{n,a_v}^{*}\le x
 \mid\widehat{\mathcal W}_n\},
 \qquad
 P_{n,v}=1-\widehat F_{n,a_v}(T_{n,a_v}).
\end{equation*}
The conditional reference distribution has unbounded upper support, so
$P_{n,v}>0$.  It may nevertheless equal one when the observed statistic is at
its lower support boundary.  We therefore use the extended Cauchy transform
\begin{equation*}
 \mathfrak c_{\rm C}(z)=
 \begin{cases}
  \cot(\pi z),&0<z<1,\\
  -\infty,&z=1,
 \end{cases}
\end{equation*}
and set
\begin{equation*}
 \mathfrak C_n=\sum_{v=1}^{L_{\cG}}\varpi_v\mathfrak c_{\rm C}(P_{n,v}),
 \qquad \varpi_v>0,\qquad \sum_{v=1}^{L_{\cG}}\varpi_v=1.
\end{equation*}
If one summand equals $-\infty$, then $\mathfrak C_n=-\infty$.  The analytic Cauchy-tail
value is
\begin{equation*}
 P_{n,\mathrm{AC}}=\frac12-\frac1\pi\arctan(\mathfrak C_n),
\end{equation*}
with $\arctan(-\infty)=-\pi/2$.

To define fixed-level joint calibration, generate a fresh common bridge
collection, compute
$(\widehat{\mathcal T}_{n,a_1}^{\circ},\ldots,
\widehat{\mathcal T}_{n,a_{L_{\cG}}}^{\circ})$, and set
\begin{equation*}
 P_{n,v}^{\circ}
 =1-\widehat F_{n,a_v}(\widehat{\mathcal T}_{n,a_v}^{\circ}),
 \qquad
 \mathfrak C_n^{\circ}=\sum_{v=1}^{L_{\cG}}\varpi_v\mathfrak c_{\rm C}(P_{n,v}^{\circ}).
\end{equation*}
Lemma~\ref{lem:finite-reference-regularity} shows that, on $\mathcal E_n$,
$P_{n,v}^{\circ}\in(0,1)$ almost surely and the conditional law of
$\mathfrak C_n^{\circ}$ is continuous.  Its conditional $(1-\alpha)$ quantile is the
joint bridge critical value for $\mathfrak C_n$.  The analytic Cauchy tail may also be
reported, but its generic guarantee is a small-tail statement rather than
exactness at every fixed conventional level; see \citet{LiuXie2020}.

For a Monte Carlo sample of $N_{\rm sim}$ common-bridge replicates, indexed
by $r_{\rm sim}$, each upper-tail value is kept away from both singular
endpoints by the mid-rank construction
\begin{equation}
 \widetilde P_{n,v,N_{\rm sim}}(x)
 =\frac{\frac12+\sum_{r_{\rm sim}=1}^{N_{\rm sim}}
 \ind\{\widehat{\mathcal T}_{n,a_v,r_{\rm sim}}^{*}\ge x\}}{N_{\rm sim}+1}
 \in(0,1).
 \label{eq:midrank-bridge-pvalue}
\end{equation}
The ordinary empirical conditional cdf associated with the same batch is
\begin{equation*}
 \widehat F_{n,a_v,N_{\rm sim}}(x)
 =\frac1{N_{\rm sim}}\sum_{r_{\rm sim}=1}^{N_{\rm sim}}
 \ind\{\widehat{\mathcal T}_{n,a_v,r_{\rm sim}}^{*}\le x\}.
\end{equation*}
For a fully simulated two-batch implementation, use a first independent batch
of size $N_{{\rm sim},1}$ to construct all functions
$\widetilde P_{n,v,N_{{\rm sim},1}}(\cdot)$ and the observed statistic
\begin{equation*}
 \mathfrak C_{n,N_{{\rm sim},1}}=\sum_{v=1}^{L_{\cG}}\varpi_v
 \mathfrak c_{\rm C}\{\widetilde P_{n,v,N_{{\rm sim},1}}(T_{n,a_v})\}.
\end{equation*}
For one fresh common-bridge draw, independent of the first batch, put
\begin{equation*}
 \mathfrak C_{n,N_{{\rm sim},1}}^{\circ}
 =\sum_{v=1}^{L_{\cG}}\varpi_v\mathfrak c_{\rm C}
 \{\widetilde P_{n,v,N_{{\rm sim},1}}(\widehat{\mathcal T}_{n,a_v}^{\circ,2})\}.
\end{equation*}
Use a second independent common-bridge batch of size $N_{{\rm sim},2}$ to generate
conditionally independent copies
\begin{equation*}
 \mathfrak C_{n,r_{\rm sim}}^{\circ,(N_{{\rm sim},1})}
 =\sum_{v=1}^{L_{\cG}}\varpi_v\mathfrak c_{\rm C}
 \{\widetilde P_{n,v,N_{{\rm sim},1}}(\widehat{\mathcal T}_{n,a_v,r_{\rm sim}}^{\circ,2})\},
 \qquad 1\le r_{\rm sim}\le N_{{\rm sim},2},
\end{equation*}
Denote their generalized empirical $(1-\alpha)$ quantile by
\begin{equation}
 \widetilde c_{n,C,1-\alpha}^{(N_{{\rm sim},1},N_{{\rm sim},2})}
 =\inf\left\{x:\frac1{N_{{\rm sim},2}}\sum_{r_{\rm sim}=1}^{N_{{\rm sim},2}}
 \ind\{\mathfrak C_{n,r_{\rm sim}}^{\circ,(N_{{\rm sim},1})}\le x\}
 \ge1-\alpha\right\}.
 \label{eq:two-batch-critical-value}
\end{equation}
Theorem~\ref{thm:critical-values} and Lemma~\ref{lem:Monte-Carlo-reference}
prove consistency of this two-batch calibration when $N_{{\rm sim},1},N_{{\rm sim},2}\to\infty$.

\section{Main theoretical results}
\label{sec:theory}

We begin by stating the following assumptions.

\begin{assumption}[ARS-type weak dependence of the original errors]
\label{ass:ARS-dependence}
Fix $\varsigma>5/2$ and let $\Hs=H^{\varsigma}[0,1]$, equipped with its usual
separable Hilbert norm.  There exist constants $\nu>4$ and $\beta_{\rm dep}\ge2$ for
which the following conditions hold.
\begin{enumerate}[label=(\alph*),leftmargin=2.2em,itemsep=3pt]
\item There are a measurable space $\mathsf S$, independent and
identically distributed innovations $\{\eta_i:i\in\mathbb Z\}$ taking
values in $\mathsf S$, and a measurable map
$\mathfrak G:\mathsf S^\infty\to\Hs$ such that
\begin{equation*}
 e_i=\mathfrak G(\eta_i,\eta_{i-1},\ldots),
 \qquad \E e_i=0,
 \qquad \E\norm{e_0}_{\Hs}^{\nu}<\infty.
\end{equation*}
\item For every integer $M\ge1$, define
\begin{equation}
 e_{i,M}=\mathfrak G(\eta_i,\ldots,\eta_{i-M+1},
 \eta^{*}_{i,M,i-M},\eta^{*}_{i,M,i-M-1},\ldots),
 \label{eq:ARS-M-approximation}
\end{equation}
where the starred variables are independent copies of $\eta_0$, mutually
independent over all their indices and independent of the original
innovation sequence.  Then $\{e_{i,M}:i\in\mathbb Z\}$ is $M$-dependent,
every $e_{i,M}$ has the same marginal distribution as $e_i$, and
\begin{equation}
 \sum_{M=1}^{\infty}(1+M)^{\beta_{\rm dep}}\mathfrak d_\nu(M)<\infty,
 \qquad
 \mathfrak d_\nu(M)=
 \left\{\E\norm{e_0-e_{0,M}}_{\Hs}^{\nu}\right\}^{1/\nu}.
 \label{eq:quantitative-ARS-condition}
\end{equation}
\end{enumerate}
\end{assumption}

The structural part of Assumption~\ref{ass:ARS-dependence} is the
Bernoulli-shift and independent-copy $M$-approximation used in
Assumption~2.1 of \citet{AueRiceSonmez2018}; it is the
$L^\nu$--$M$-approximability framework introduced by
\citet{HormannKokoszka2010}.  The polynomial weight in
\eqref{eq:quantitative-ARS-condition} is a primitive quantitative version
of the same approximation condition.  The requirement $\nu>4$ supplies the
fourth moments needed by the quadratic DB estimator, while $\beta_{\rm dep}\ge2$
is sufficient for the largest-gap summation of the fourth-order
contractions.  The kernel order below is required only to satisfy
$\beta_K\le\beta_{\rm dep}$.  No Gaussianity, autoregressive or moving-average
representation, functional linear filter, cumulant condition, projected-score
condition, or estimator-consistency condition is imposed.  The choice
$\Hs=H^{\varsigma}[0,1]$ is used only to obtain the primitive path smoothness needed
for numerical Fourier integration.  Since $\varsigma>5/2$, the one-dimensional
Sobolev embedding gives
$\norm{f}_{C^2}\le C_{\varsigma}\norm{f}_{\Hs}$; see
\citet[Chapter~5]{AdamsFournier2003}.

\begin{assumption}[Grid and Fourier approximation]
\label{ass:grid}
For the single-change model, the functions $\mu$ and $\delta_n$ belong to $C^2[0,1]$, and
$\sup_n\{\norm{\mu}_{C^2}+\norm{\delta_n}_{C^2}\}<\infty$.
For the multiple-change model, this smoothness requirement is replaced by the
uniform condition in Assumption~\ref{ass:multiple-configuration}.  For derivative orders $r_{\rm der}=0,1,2$,
\begin{equation*}
 \max_{1\le j\le p}\norm{\phi_j^{(r_{\rm der})}}_\infty\le C_\phi p^{r_{\rm der}}.
\end{equation*}
The projection dimension $p=p_n$ and grid size $N=N_n$ satisfy
\begin{equation*}
 p\longrightarrow\infty,
 \qquad N\longrightarrow\infty,
 \qquad
 \chi_{p,N}:=\frac{p^{5/2}}{(N-1)^2}\longrightarrow0.
\end{equation*}
\end{assumption}

\begin{assumption}[Difference sequence, kernel and bandwidth]
\label{ass:DB-tuning}
The difference order $m$ is fixed, and \eqref{eq:difference-constraints} holds with
$C_c=\sum_{q=0}^m|c_q|<\infty$. The kernel $K$ is even, bounded, supported on $[-1,1]$, satisfies $K(0)=1$, and for constants $0<\beta_K\le\beta_{\rm dep}$ and $C_K<\infty$,
\begin{equation*}
 |1-K(x)|\le C_K|x|^{\beta_K},
 \qquad |x|\le1.
\end{equation*}
The positive integer sequences $\ell=\ell_n$ and $h=h_n$ satisfy
\begin{equation*}
 \ell\longrightarrow\infty,
 \qquad 2\ell\le h\le C_h\ell,
 \qquad n_{\rm loss}^{\rm DB}:=\ell+mh=o(n),
 \qquad N_D=n-mh.
\end{equation*}
Consequently $N_D-\ell=n-n_{\rm loss}^{\rm DB}\to\infty$, so every denominator
$N_D-r$, $0\le r<\ell$, is positive for all sufficiently large $n$.
\end{assumption}

\begin{assumption}[Nondegeneracy and growth]
\label{ass:spectrum-growth}
Let $\Omega$ be the long-run covariance operator of $\{e_i\}$. Lemma~\ref{prop:lrc-origin} derives from Assumption~\ref{ass:ARS-dependence} that $\Omega$ is nonnegative and trace class. We assume only the nondegeneracy condition
\begin{equation*}
 \mathfrak t_\Omega=\tr(\Omega)>0.
\end{equation*}
Let $\Pi_p$ be the exact orthogonal projection onto $\operatorname{span}(\phi_1,\ldots,\phi_p)$ and define
\begin{equation*}
 \mathfrak r_n=\left(\frac{\ell}{N_D}\right)^{1/2}+\ell^{-\beta_K},
 \qquad
 \zeta_{p,N}=\norm{\Omega-\Pi_p\Omega\Pi_p}_{1}+\chi_{p,N},
 \qquad
 \mathfrak e_n=\sqrt p\,\mathfrak r_n+\zeta_{p,N}.
\end{equation*}
The only additional growth restriction is
\begin{equation}
 \sqrt p\,\mathfrak r_n\longrightarrow0.
 \label{eq:main-growth}
\end{equation}
 The ridge grid $\cG=\{a_1,\ldots,a_{L_{\cG}}\}$ is fixed and obeys
\begin{equation*}
 0<a_-\le\min_v a_v\le\max_v a_v\le a_+<\infty.
\end{equation*}
\end{assumption}

Because $\Omega$ is trace class, $\Pi_p\to\Id_{\cH}$ strongly, and
$\chi_{p,N}\to0$ by Assumption~\ref{ass:grid}, Lemma~\ref{prop:quadrature}
implies $\zeta_{p,N}\to0$.  Hence \eqref{eq:main-growth} also implies
$\mathfrak e_n\to0$.  The term $\sqrt p\,\mathfrak r_n$ is needed only to convert the Hilbert--Schmidt error of an indefinite raw estimator into trace-norm control after positive-part projection. For the implementation $\beta_K=2$, $h=2\ell$ and $\ell\asymp n^{1/4}$, one has $\mathfrak r_n\asymp n^{-3/8}$ and \eqref{eq:main-growth} permits $p=o(n^{3/4})$, apart from the deterministic Fourier-approximation term.

All technical consequences of the primitive assumptions are stated in
Appendix~\ref{app:auxiliary-results}.  In particular, the Appendix derives
trace-class existence of the long-run covariance operator, the functional
invariance principle, numerical-projection error bounds, the rate of the
edge-corrected DB estimator, and consistency of the estimated ridge weights.
The present section contains only the principal distributional, calibration,
power, and localization results.

For $u\in\mathbb Z$, let
\begin{equation*}
 \mathcal C_u=\E(e_u\otimes e_0),
 \qquad
 \Omega=\sum_{u\in\mathbb Z}\mathcal C_u,
 \qquad
 \mathfrak t_\Omega=\tr(\Omega)>0.
\end{equation*}
Let $\lambda_1\ge\lambda_2\ge\cdots>0$ be the positive eigenvalues of
$\Omega$, repeated according to multiplicity; if $\Omega$ has finite rank, all
series below terminate at that rank.  For $a\in\cG$, define
\begin{equation}
 w_j(a)=\frac{\lambda_j}{\lambda_j+a\mathfrak t_\Omega},
 \qquad
 A_1(a)=\sum_{j=1}^{\infty}w_j(a),
 \qquad
 A_2(a)=\sum_{j=1}^{\infty}w_j(a)^2.
 \label{eq:population-ridge-weights}
\end{equation}
The series in \eqref{eq:population-ridge-weights} are finite, and
$A_2(a)>0$.

\subsection{Null distribution and direct spectral calibration}

Let $B_1,B_2,\ldots$ be independent standard Brownian bridges and define
\begin{equation}
 \mathcal D_a(t)
 =\frac{\sum_{j=1}^{\infty}w_j(a)
 \{B_j^2(t)/[t(1-t)]-1\}}
 {\{2A_2(a)\}^{1/2}},
 \qquad \epsilon\le t\le1-\epsilon.
 \label{eq:weighted-bridge-process}
\end{equation}
Lemma~\ref{lem:bridge-series-regularity} shows that the series in
\eqref{eq:weighted-bridge-process} converges uniformly almost surely on the
trimmed interval and defines a continuous process.

\begin{theorem}[Joint weighted-bridge null limit]
\label{thm:null-limit}
Suppose $H_0$ and Assumptions~\ref{ass:ARS-dependence}--
\ref{ass:spectrum-growth} hold. Then
\begin{equation}
 \left\{D_{n,a}(\lfloor nt\rfloor):
 a\in\cG,\ \epsilon\le t\le1-\epsilon\right\}
 \dto
 \left\{\mathcal D_a(t):
 a\in\cG,\ \epsilon\le t\le1-\epsilon\right\}
 \label{eq:joint-null-limit}
\end{equation}
in $\ell^\infty([\epsilon,1-\epsilon])^{L_{\cG}}$. Consequently,
\begin{equation}
 (T_{n,a_1},\ldots,T_{n,a_{L_{\cG}}})
 \dto
 (\mathcal T_{a_1},\ldots,\mathcal T_{a_{L_{\cG}}}),
 \qquad
 \mathcal T_a=\sup_{\epsilon\le t\le1-\epsilon}\mathcal D_a(t).
 \label{eq:joint-scan-limit}
\end{equation}
For $a,a'\in\cG$,
\begin{align}
 \Cov\{\mathcal D_a(t),\mathcal D_{a'}(t')\}
 &=\frac{\{(t\wedge t')-tt'\}^2}
 {t(1-t)t'(1-t')}
 \frac{\sum_{j\ge1}w_j(a)w_j(a')}
 {\{A_2(a)A_2(a')\}^{1/2}}.
 \label{eq:cross-ridge-covariance}
\end{align}
\end{theorem}

At a fixed $t$, the skewness of $\mathcal D_a(t)$ equals
\begin{equation*}
 2\sqrt2\,
 \frac{\sum_jw_j(a)^3}{\{\sum_jw_j(a)^2\}^{3/2}}.
\end{equation*}
It generally does not vanish when a finite number of low-frequency directions
carry non-negligible ridge weight.  A Gaussian-process limit is recovered only
under an additional diffuse-spectrum condition, which is not imposed here.

Let $\widehat F_{n,a}$ be the conditional distribution function of
$\widehat{\mathcal T}_{n,a}^{*}$ in \eqref{eq:conditional-limit}, and let
$F_a$ be the distribution function of $\mathcal T_a$.  Let
$\widehat F_n^{\rm joint}$ and $F^{\rm joint}$ be the corresponding joint
laws generated with common bridge paths.  For probability laws on a Polish
space, write $d_{\mathrm{BL}}$ for the bounded--Lipschitz metric.  Throughout
the calibration results, $0<\alpha<1$.  Finally, let $G_C$ be the
distribution function of
\begin{equation*}
 \mathfrak C_\infty=\sum_{v=1}^{L_{\cG}}\varpi_v\,
 \mathfrak c_{\rm C}\{1-F_{a_v}(\mathcal T_{a_v})\}.
\end{equation*}
Let
\begin{equation*}
 c_{C,1-\alpha}=G_C^{-1}(1-\alpha)
 =\inf\{x:G_C(x)\ge1-\alpha\}.
\end{equation*}
For the chosen finite ridge grid and level $\alpha$, write
$\mathsf{CR}_{\alpha}$ for the following one-dimensional calibration
regularity condition:
\begin{equation*}
 \begin{gathered}
 G_C\text{ is continuous at }c_{C,1-\alpha},\\
 G_C(c_{C,1-\alpha}-\eta)<1-\alpha<
 G_C(c_{C,1-\alpha}+\eta)
 \quad\text{for every }\eta>0.
 \end{gathered}
\end{equation*}
This condition concerns the joint Cauchy limit, not its continuous marginal
$p$-values, and is therefore stated explicitly rather than inferred from
marginal continuity.

\begin{theorem}[Direct spectral calibration]
\label{thm:critical-values}
Under the assumptions of Theorem~\ref{thm:null-limit},
\begin{equation*}
 d_{\mathrm{BL}}(\widehat F_n^{\rm joint},F^{\rm joint})\pto0.
\end{equation*}
Moreover,
\begin{equation}
 \max_{a\in\cG}\sup_x
 |\widehat F_{n,a}(x)-F_a(x)|\pto0.
 \label{eq:conditional-cdf-consistency}
\end{equation}
Every $F_a$ is continuous and strictly increasing on its support.  Hence, for
$c_{a,1-\alpha}=F_a^{-1}(1-\alpha)$ and its conditional plug-in version
$\widehat c_{n,a,1-\alpha}$,
\begin{equation*}
 \widehat c_{n,a,1-\alpha}\pto c_{a,1-\alpha},
 \qquad
 \Pp\{T_{n,a}>\widehat c_{n,a,1-\alpha}\}\longrightarrow\alpha.
\end{equation*}
If $\mathsf{CR}_{\alpha}$ holds, then the conditional joint-bridge critical
value $\widehat c_{n,C,1-\alpha}$ satisfies
\begin{equation*}
 \widehat c_{n,C,1-\alpha}\pto c_{C,1-\alpha},
 \qquad
 \Pp\{\mathfrak C_n>\widehat c_{n,C,1-\alpha}\}\longrightarrow\alpha.
\end{equation*}
\end{theorem}

Lemma~\ref{lem:finite-reference-regularity} establishes the finite conditional
reference-law regularity used above.  Lemma~\ref{lem:Monte-Carlo-reference}
shows that direct simulation of the explicit bridge law consistently
approximates both the fixed-ridge and, under $\mathsf{CR}_{\alpha}$, the
two-batch joint critical values.  The condition $\mathsf{CR}_{\alpha}$ is
needed only to identify a unique nonrandomized limiting critical value; the
joint weak convergence itself does not use it.  It holds automatically for a
singleton ridge grid and, more generally, when all normalized ridge-weight
sequences coincide, in which case $\mathfrak C_\infty$ is standard Cauchy.
The analytic Cauchy-tail value may also be reported, but its generic guarantee
is a small-tail approximation rather than exactness at every fixed conventional
level.

\subsection{Local alternatives, consistency and localization}

Let $\vartheta_n=\tau_n/n\to\vartheta\in(\epsilon,1-\epsilon)$ and put
\begin{equation*}
 d_n=\mathcal J_{p,N}\delta_n=\iota_p\bdelta_{p,N}.
\end{equation*}
For $t,\vartheta\in[0,1]$, define
\begin{equation*}
 g(t,\vartheta)
 =(t\wedge\vartheta)-t\vartheta
 =\begin{cases}
 t(1-\vartheta),&t\le\vartheta,\\
 \vartheta(1-t),&t>\vartheta.
 \end{cases}
\end{equation*}
For $t,\vartheta\in(0,1)$, define the common change-point shape
\begin{equation}
 \mathfrak h_{\rm cp}(t,\vartheta)
 =\frac{g(t,\vartheta)^2}{t(1-t)}.
 \label{eq:common-change-shape}
\end{equation}
Let $\{\psi_j\}$ be eigenfunctions corresponding to the positive eigenvalues
$\lambda_j$ of $\Omega$, and write
$d=d_0+\sum_{j\ge1}d_j\psi_j$, where $d_0\in\ker(\Omega)$.  Put
\begin{equation*}
 \mathfrak r_{\delta,n}
 =\frac{\ell mh}{N_D}\norm{d_n}^2
 +\frac{\ell\sqrt{mh}}{N_D}\norm{d_n}.
\end{equation*}
For $U_j(t)=B_j(t)/\sqrt{t(1-t)}$, define
\begin{align}
 \mathcal D_{a,d}(t)
 &=\{2A_2(a)\}^{-1/2}
 \Bigg[
 \sum_{j\ge1}w_j(a)\{U_j(t)^2-1\}
 \notag\\
 &\quad-\frac{2g(t,\vartheta)}{\sqrt{t(1-t)}}
 \sum_{j\ge1}
 \frac{\sqrt{\lambda_j}d_j}{\lambda_j+a\mathfrak t_\Omega}U_j(t)
 \notag\\
 &\quad+\frac{g(t,\vartheta)^2}{t(1-t)}
 \left\{
 \sum_{j\ge1}\frac{d_j^2}{\lambda_j+a\mathfrak t_\Omega}
 +\frac{\norm{d_0}^2}{a\mathfrak t_\Omega}
 \right\}
 \Bigg].
 \label{eq:noncentral-weighted-bridge}
\end{align}

\begin{theorem}[Local alternative]
\label{thm:local-alternative}
Suppose model~\eqref{eq:functional-change-model} holds and
Assumptions~\ref{ass:ARS-dependence}--
\ref{ass:spectrum-growth} are satisfied,
$\sqrt n\,d_n\to d$ in $\cH$, and
\begin{equation*}
 \sqrt p\{\mathfrak r_n+\mathfrak r_{\delta,n}\}\longrightarrow0.
\end{equation*}
Then, jointly over $a\in\cG$,
\begin{equation}
 D_{n,a}(\lfloor nt\rfloor)\dto\mathcal D_{a,d}(t)
 \label{eq:local-process-limit}
\end{equation}
in $\ell^\infty([\epsilon,1-\epsilon])^{L_{\cG}}$.
\end{theorem}

For $a>0$, define
\begin{equation*}
 \mathcal R_a=(\Omega+a\mathfrak t_\Omega \Id_{\cH})^{-1}.
\end{equation*}
For $f\in\cH$, the corresponding ridge signal energy is
\begin{equation*}
 \mathfrak q_a(f)=\ip{f}{\mathcal R_a f}.
\end{equation*}
For $x\in\R^p$, $\mathfrak q_a(x)$ means $\mathfrak q_a(\iota_px)$.

\begin{theorem}[Consistency]
\label{thm:consistency}
Suppose model \eqref{eq:functional-change-model} has one change with
$\vartheta_n\to\vartheta\in(\epsilon,1-\epsilon)$ and
Assumptions~\ref{ass:ARS-dependence}--\ref{ass:spectrum-growth} hold.  If
\begin{equation}
 \sqrt p\{\mathfrak r_n+\mathfrak r_{\delta,n}\}\longrightarrow0
 \label{eq:alt-estimation-condition}
\end{equation}
and $n\mathfrak q_a(d_n)\to\infty$ for at least one $a\in\cG$, then
\begin{equation*}
 T_{n,a'}\pto\infty
 \qquad\text{for every }a'\in\cG.
\end{equation*}
\end{theorem}

Consequently, every fixed-ridge bridge-calibrated test is consistent, all
marginal bridge-limit $p$-values converge to zero, and both the analytic and
the joint-limit calibrated Cauchy procedures are consistent.  The same conclusion holds for direct Monte Carlo bridge calibration when
both Monte Carlo batch sizes diverge.

The consistency statement concerns fixed alternatives.  Under the contiguous
alternatives of Theorem~\ref{thm:local-alternative}, the limiting rejection
probability can be characterized more sharply.  For $a\in\cG$, define
\begin{equation*}
 \mathcal T_{a,d}=\sup_{\epsilon\le t\le1-\epsilon}\mathcal D_{a,d}(t),
 \qquad
 \beta_a(d,\vartheta;\alpha)
 =\Pp\{\mathcal T_{a,d}>c_{a,1-\alpha}\}.
\end{equation*}
For $a_v\in\cG$, put, using the extended Cauchy transform defined in
Section~2,
\begin{equation*}
 P_{a_v,d}=1-F_{a_v}(\mathcal T_{a_v,d}),
 \qquad
 \mathfrak C_d=\sum_{v=1}^{L_{\cG}}\varpi_v
 \mathfrak c_{\rm C}(P_{a_v,d}),
\end{equation*}
and let
\begin{equation*}
 \beta_C(d,\vartheta;\alpha)
 =\Pp\{\mathfrak C_d>c_{C,1-\alpha}\}.
\end{equation*}

\begin{theorem}[Local asymptotic power]
\label{thm:local-power}
Suppose the conditions of Theorem~\ref{thm:local-alternative} hold.  Then the
distribution of $\mathcal T_{a,d}$ is continuous for every $a\in\cG$, and
\begin{equation}
 \Pp\{T_{n,a}>\widehat c_{n,a,1-\alpha}\}
 \longrightarrow \beta_a(d,\vartheta;\alpha).
 \label{eq:fixed-ridge-local-power}
\end{equation}
If $\mathsf{CR}_{\alpha}$ holds, then
\begin{align}
 \beta_C(d,\vartheta;\alpha)
 &\le\liminf_{n\to\infty}
 \Pp\{\mathfrak C_n>\widehat c_{n,C,1-\alpha}\}\notag\\
 &\le\limsup_{n\to\infty}
 \Pp\{\mathfrak C_n>\widehat c_{n,C,1-\alpha}\}
 \le\Pp\{\mathfrak C_d\ge c_{C,1-\alpha}\}.
 \label{eq:Cauchy-local-power-bounds}
\end{align}
The same bounds hold for the two-batch Monte Carlo implementation when both
batch sizes diverge.  In particular, if
$\Pp\{\mathfrak C_d=c_{C,1-\alpha}\}=0$, then both implementations satisfy
\begin{equation}
 \Pp\{\mathfrak C_n>\widehat c_{n,C,1-\alpha}\}
 \longrightarrow \beta_C(d,\vartheta;\alpha),
 \label{eq:Cauchy-local-power}
\end{equation}
and the corresponding two-batch rejection probability has the same limit.
\end{theorem}

The formulas in Theorem~\ref{thm:local-power} can be made explicit at the true
change fraction.  Let $Z_1,Z_2,\ldots$ be independent standard normal variables
and set
\begin{equation*}
 \eta_j=\frac{\{\vartheta(1-\vartheta)\}^{1/2}d_j}{\sqrt{\lambda_j}},
 \qquad j\ge1.
\end{equation*}
The sequence $\{\eta_j\}$ need not be square summable.  Nevertheless, for
every fixed $a>0$,
\begin{align*}
 \sum_{j\ge1}w_j(a)\{1+\eta_j^2\}
 &=A_1(a)+\vartheta(1-\vartheta)
 \sum_{j\ge1}\frac{d_j^2}{\lambda_j+a\mathfrak t_\Omega}\\
 &\le A_1(a)+
 \frac{\vartheta(1-\vartheta)\norm{d}^2}{a\mathfrak t_\Omega}<\infty.
\end{align*}
Hence the ridge-weighted noncentral chi-square series below converges almost
surely and in $L^1$.  Completing the square in
\eqref{eq:noncentral-weighted-bridge} gives
\begin{equation}
 \mathcal D_{a,d}(\vartheta)
 \stackrel{d}{=}
 \frac{
 \sum_{j\ge1}w_j(a)\{(Z_j-\eta_j)^2-1\}
 +\dfrac{\vartheta(1-\vartheta)\norm{d_0}^2}
 {a\mathfrak t_\Omega}}
 {\{2A_2(a)\}^{1/2}}.
 \label{eq:true-change-noncentral-form}
\end{equation}
Thus the pointwise statistic is a null-centered and null-standardized weighted
sum of independent noncentral chi-square variables.  Its mean at a general scan
fraction and its variance at the true fraction are
\begin{align}
 \E\{\mathcal D_{a,d}(t)\}
 &=\frac{\mathfrak h_{\rm cp}(t,\vartheta)\mathfrak q_a(d)}
 {\{2A_2(a)\}^{1/2}},
 \label{eq:local-power-mean-drift}\\
 \Var\{\mathcal D_{a,d}(\vartheta)\}
 &=1+\frac{2\vartheta(1-\vartheta)}{A_2(a)}
 \sum_{j\ge1}\frac{\lambda_jd_j^2}
 {(\lambda_j+a\mathfrak t_\Omega)^2}.
 \label{eq:local-power-pointwise-variance}
\end{align}
The function $\mathfrak h_{\rm cp}(t,\vartheta)$ is uniquely maximized at
$t=\vartheta$, and the resulting standardized ridge signal is
\begin{equation}
 \Lambda_a(d,\vartheta)
 =\frac{\vartheta(1-\vartheta)\mathfrak q_a(d)}
 {\{2A_2(a)\}^{1/2}}.
 \label{eq:ridge-standardized-signal}
\end{equation}
Because $\mathcal T_{a,d}\ge\mathcal D_{a,d}(\vartheta)$, the exact scan power
in \eqref{eq:fixed-ridge-local-power} obeys the following explicit oracle lower
bound, which is numerically evaluable once $(d,\Omega,\vartheta)$ are specified:
\begin{align}
 \beta_a(d,\vartheta;\alpha)
 &\ge
 \Pp\Bigg\{
 \sum_{j\ge1}w_j(a)\chi_{1,j}^2(\eta_j^2)
 >A_1(a)+\{2A_2(a)\}^{1/2}c_{a,1-\alpha}
 \notag\\
 &\hspace{45mm}
 -\frac{\vartheta(1-\vartheta)\norm{d_0}^2}
 {a\mathfrak t_\Omega}
 \Bigg\},
 \label{eq:pointwise-power-lower-bound}
\end{align}
where the $\chi_{1,j}^2(\eta_j^2)$ are independent noncentral chi-square
variables with one degree of freedom.

The exact oracle ridge and the ridge that maximizes the mean drift at the true
change are, respectively,
\begin{equation}
 \begin{aligned}
 a_{\rm opt}(d,\vartheta;\alpha)
 &=\min\left\{a\in\cG:
 \beta_a(d,\vartheta;\alpha)
 =\max_{b\in\cG}\beta_b(d,\vartheta;\alpha)\right\},\\
 a_{\rm drift}(d)
 &=\min\left\{a\in\cG:
 \frac{\mathfrak q_a(d)}{\{A_2(a)\}^{1/2}}
 =\max_{b\in\cG}\frac{\mathfrak q_b(d)}{\{A_2(b)\}^{1/2}}
 \right\}.
 \end{aligned}
 \label{eq:oracle-ridge-definitions}
\end{equation}
The first criterion uses the complete noncentral bridge process and the
ridge-specific critical value, whereas the second isolates the deterministic
signal-to-null-noise ratio in \eqref{eq:ridge-standardized-signal}.  Both depend
on the unknown spectral coordinates $\{d_j\}$ of the change.  To illustrate
the resulting tradeoff, if $\Omega$ has rank $r<\infty$ and $d_0=0$, then
\begin{equation}
 \lim_{a\downarrow0}
 \frac{\mathfrak q_a(d)}{\{A_2(a)\}^{1/2}}
 =\frac{\sum_{j=1}^{r}d_j^2/\lambda_j}{\sqrt r},
 \qquad
 \lim_{a\to\infty}
 \frac{\mathfrak q_a(d)}{\{A_2(a)\}^{1/2}}
 =\frac{\norm{d}^2}{\norm{\Omega}_{\rm HS}}.
 \label{eq:ridge-signal-limits}
\end{equation}
If $d_0\ne0$, the first limit is infinite.  These are population-level ridge
regime comparisons only: all preceding asymptotic results use a fixed finite
grid bounded away from zero and infinity and do not cover sequences
$a_n\downarrow0$ or $a_n\to\infty$.  A small ridge approximates
inverse-covariance weighting and can be particularly effective for changes in
low-variance or null-space directions, whereas a large ridge approaches an
$L^2$-type detector.  Under a heterogeneous spectrum, different unknown change
directions can therefore favor different ridge regimes, so no direction-free
ridge value is guaranteed to maximize power over the admissible alternatives.
Degenerate spectra can produce ties; for example, in the rank-one case with
$d_0=0$, the drift ratio in \eqref{eq:ridge-signal-limits} is independent
of $a$.

This is the reason for the Cauchy aggregation.  The alternative direction is
not known under $H_0$, so neither $a_{\rm opt}$ nor $a_{\rm drift}$ is available
when the test is constructed.  The Cauchy statistic retains the full ridge
grid, allows a small marginal $p$-value from a well-aligned ridge to contribute
strongly to the aggregate, and uses the common-bridge joint calibration to
account for cross-ridge dependence.  Equations~\eqref{eq:Cauchy-local-power-bounds}--
\eqref{eq:Cauchy-local-power} provide its local asymptotic power bounds and,
when the alternative limit has no atom at the calibrated boundary, its exact
local asymptotic power.  The combined test is not asserted
to equal the infeasible oracle for every local alternative; rather, it removes
the need to select an alternative-specific ridge while preserving valid
fixed-level inference.

\begin{theorem}[Location consistency]
\label{thm:consistency-location}
Under the conditions of Theorem~\ref{thm:consistency}, let
\begin{equation*}
 \widehat\tau_{n,a}
 =\argmax_{k\in\cK_{n,\epsilon}}V_{n,a}(k),
\end{equation*}
with the smallest maximizer used to break ties.  Then
\begin{equation}
 \max_{a\in\cG}\max_{k\in\cK_{n,\epsilon}}
 \left|
 \frac{V_{n,a}(k)}{n\mathfrak q_a(d_n)}
 -\frac{g(k/n,\vartheta_n)^2}{(k/n)(1-k/n)}
 \right|\pto0
 \label{eq:uniform-drift-main}
\end{equation}
and
\begin{equation}
 \max_{a\in\cG}
 \left|\frac{\widehat\tau_{n,a}}{n}-\vartheta\right|\pto0.
 \label{eq:location-consistency}
\end{equation}
\end{theorem}

Because the ridge grid is finite, the uniform conclusion in
\eqref{eq:location-consistency} also permits a data-dependent ridge choice.
For reporting one location from the ridge family, define
\begin{equation*}
 \widehat v_n
 =\min\left\{v:P_{n,v}=\min_{1\le u\le L_{\cG}}P_{n,u}\right\},
 \qquad
 \widehat a_n=a_{\widehat v_n},
 \qquad
 \widehat\tau_n^{\rm RHT}=\widehat\tau_{n,\widehat a_n}.
\end{equation*}
Then $|\widehat\tau_n^{\rm RHT}/n-\vartheta|\pto0$.  In computation,
$P_{n,v}$ is replaced by its marginal mid-rank Monte Carlo version in
\eqref{eq:midrank-bridge-pvalue}.  This adaptive location estimator is used
in the simulation study.

\section{Multiple-change estimation by wild binary segmentation}
\label{sec:multiple}

\subsection{Piecewise-constant mean model and local ridge paths}

We now extend the preceding construction from one unknown mean change to an
unknown finite collection of changes.  Let $J_{\rm cp}\ge1$ be fixed and
consider
\begin{equation}
 X_i=\mu_0+\sum_{b=1}^{J_{\rm cp}}\delta_{b,n}
 \ind(i>\tau_{b,n})+e_i,
 \qquad i=1,\ldots,n,
 \label{eq:multiple-change-model}
\end{equation}
where
\begin{equation*}
 0=\tau_{0,n}<\tau_{1,n}<\cdots<\tau_{J_{\rm cp},n}
 <\tau_{J_{\rm cp}+1,n}=n,
 \qquad
 \delta_{b,n}\ne0,\quad 1\le b\le J_{\rm cp}.
\end{equation*}
The nonzero-jump convention makes every listed location identifiable.  The
error process is the same centered stationary process as in
\eqref{eq:functional-change-model}; only the deterministic mean is allowed to
change.  Put
\begin{equation*}
 \Delta_n=\min_{0\le b\le J_{\rm cp}}
 (\tau_{b+1,n}-\tau_{b,n}),
 \qquad
 d_{b,n}=\mathcal J_{p,N}\delta_{b,n},
 \qquad
 \delta_{\min,n}=\min_{1\le b\le J_{\rm cp}}\norm{\delta_{b,n}}.
\end{equation*}
Here $\delta_{\min,n}$ is the smallest jump in the original functional model,
whereas $d_{b,n}$ is the deterministic jump retained by the implemented
numerical basis expansion.  For $B<\infty$, define
\begin{equation}
 \mathfrak a_{p,N}(B)
 =\sup_{\substack{f\in C^2[0,1]\\ \norm f_{C^2}\le B}}
 \norm{\mathcal J_{p,N}f-f}.
 \label{eq:uniform-functional-resolution}
\end{equation}
Lemma~\ref{prop:quadrature} gives $\mathfrak a_{p,N}(B)\to0$.

\begin{assumption}[Multiple-change configuration]
\label{ass:multiple-configuration}
The number $J_{\rm cp}$ is fixed, and there is a constant $c_{\Delta}>0$ such
that $\Delta_n\ge c_{\Delta}n$.  The functions
$\mu_0,\delta_{1,n},\ldots,\delta_{J_{\rm cp},n}$ belong to $C^2[0,1]$ and
satisfy
\begin{equation*}
 \sup_n\norm{\mu_0}_{C^2}<\infty,
 \qquad
 B_\delta:=\sup_n\max_{1\le b\le J_{\rm cp}}
 \norm{\delta_{b,n}}_{C^2}<\infty,
 \qquad
 \inf_n\delta_{\min,n}>0.
\end{equation*}
\end{assumption}

The fixed-$J_{\rm cp}$ and proportional-spacing formulation isolates the new
issue of recursive localization with one globally estimated long-run
covariance.  Since $n_{\rm loss}^{\rm DB}=mh+\ell=o(n)$ under
Assumption~\ref{ass:DB-tuning}, it also gives $mh+\ell=o(\Delta_n)$.

For an interval $I=(s,e]$ with $n_I=e-s\ge2$, define
\begin{equation*}
 \cK_I=\{s+1,\ldots,e-1\}.
\end{equation*}
A change at an endpoint is not an interior change of $I$.  For
$k\in\cK_I$, set
\begin{equation*}
 \widehat\bDelta_I(k)
 =\frac1{e-k}\sum_{i=k+1}^{e}\widehat\bxi_i
 -\frac1{k-s}\sum_{i=s+1}^{k}\widehat\bxi_i,
 \qquad
 N_I(k)=\frac{(k-s)(e-k)}{n_I},
\end{equation*}
and define
\begin{align*}
 V_{I,a}(k)
 &=N_I(k)\widehat\bDelta_I(k)^{\top}
 \widehat{\bm R}_{n,a}\widehat\bDelta_I(k),\\
 D_{I,a}(k)
 &=\frac{V_{I,a}(k)-\widehat A_{1,n}(a)}
 {\{2\widehat A_{2,n}(a)\}^{1/2}},\\
 \cS_I(k)&=\max_{a\in\cG}D_{I,a}(k),
 \qquad
 \cS_I^{\max}=\max_{k\in\cK_I}\cS_I(k),\\
 \widetilde\tau_I
 &=\min\argmax_{k\in\cK_I}\cS_I(k).
\end{align*}
Every interval uses the same full-sample matrices
$\widehat{\bm R}_{n,a}$ and the same spectral centering and scaling.  Hence the
long-run covariance is estimated only once.  The maximum over the finite ridge
grid is used for localization.  On a one-change interval, all ridge-specific
deterministic drifts have the same unique maximizer.

\subsection{Wild binary segmentation and local refinement}

Let $M_n^{\rm WBS}$ be a positive integer.  Independently of the observations,
draw endpoint pairs
\begin{equation*}
 (S_{r_{\rm W}}^{\rm WBS},E_{r_{\rm W}}^{\rm WBS}),
 \qquad 1\le r_{\rm W}\le M_n^{\rm WBS},
\end{equation*}
independently and uniformly with replacement from
\begin{equation*}
 \mathbb I_n^{\rm WBS}
 =\{(s,e):0\le s<e\le n,\ e-s\ge2\}.
\end{equation*}
Write $I_{r_{\rm W}}^{\rm WBS}=(S_{r_{\rm W}}^{\rm WBS},E_{r_{\rm W}}^{\rm WBS}]$ and
$\mathcal F_n^{\rm WBS}=\{I_{r_{\rm W}}^{\rm WBS}:1\le r_{\rm W}\le M_n^{\rm WBS}\}$, where
repeated draws retain their indices.  The random intervals are generated once
and reused throughout the recursion, as in the WBS construction of
\citet{Fryzlewicz2014}.  For a recursive segment $(s,e]$, put
\begin{equation*}
 \mathcal M_n^{\rm WBS}(s,e)
 =\{r_{\rm W}:s\le S_{r_{\rm W}}^{\rm WBS}<E_{r_{\rm W}}^{\rm WBS}\le e\}.
\end{equation*}

\begin{algorithm}[Difference-based ridge wild binary segmentation]
\label{alg:DB-RWBS}
The inputs are $\Lambda_n>0$, the score sequence, the global DB ridge
quantities, and $\mathcal F_n^{\rm WBS}$.  Apply the following operations.
\begin{enumerate}[label=\arabic*.,leftmargin=2.2em,itemsep=4pt]
\item If $\mathcal E_n^c$ occurs, return $\widehat J_{\rm cp}=0$ and the empty
location vector.  Otherwise call $\operatorname{WBS}(0,n)$.
\item For a call $\operatorname{WBS}(s,e)$, stop if $e-s<2$ or
$\mathcal M_n^{\rm WBS}(s,e)$ is empty.  Otherwise choose
$(r_{\rm W}^\star,k^\star)$ to maximize
\begin{equation*}
 \cS_{I_{r_{\rm W}}^{\rm WBS}}(k),
 \qquad
 r_{\rm W}\in\mathcal M_n^{\rm WBS}(s,e),\quad
 k\in\cK_{I_{r_{\rm W}}^{\rm WBS}}.
\end{equation*}
Break ties by the smallest random-interval index and then the smallest split point.
If the maximum does not exceed $\Lambda_n$, stop.  If it exceeds
$\Lambda_n$, record $k^\star$ and call
$\operatorname{WBS}(s,k^\star)$ and
$\operatorname{WBS}(k^\star,e)$.
\item Let $\widetilde J_{\rm cp}$ be the number of recorded points and order
them as
$\widetilde\tau_{1,n}<\cdots<\widetilde\tau_{\widetilde J_{\rm cp},n}$.
Set $\widetilde\tau_{0,n}=0$ and
$\widetilde\tau_{\widetilde J_{\rm cp}+1,n}=n$.  For
$1\le b\le\widetilde J_{\rm cp}$, define
\begin{align*}
 s_{b,n}^{\rm ref}
 &=\left\lfloor\frac{\widetilde\tau_{b-1,n}
 +\widetilde\tau_{b,n}}2\right\rfloor,\\
 e_{b,n}^{\rm ref}
 &=\left\lfloor\frac{\widetilde\tau_{b,n}
 +\widetilde\tau_{b+1,n}}2\right\rfloor,
 \qquad
 I_{b,n}^{\rm ref}=(s_{b,n}^{\rm ref},e_{b,n}^{\rm ref}],
\end{align*}
and put
\begin{equation*}
 \widehat\tau_{b,n}
 =\min\argmax_{k\in\cK_{I_{b,n}^{\rm ref}}}
 \cS_{I_{b,n}^{\rm ref}}(k).
\end{equation*}
If the finite-sample candidate set is empty, retain
$\widehat\tau_{b,n}=\widetilde\tau_{b,n}$.  Return
$\widehat J_{\rm cp}=\widetilde J_{\rm cp}$ and the labelled refined
locations.
\end{enumerate}
\end{algorithm}

The random localization mechanism distinguishes WBS from ordinary binary
segmentation.  With high probability, each true change is covered by a drawn
interval whose endpoints lie well inside the two adjacent constant-mean
segments.  Such an interval contains only that change and produces a local
ridge drift of order $\Delta_n\norm{d_{b,n}}^2$.  The midpoint refinement is
applied only after the recursive candidate set has been ordered.  It creates a
single-change interval around each preliminary estimate and does not require a
second long-run covariance estimate.

\subsection{Consistency and localization}

For model~\eqref{eq:multiple-change-model}, define
\begin{equation*}
 \mathfrak r_{\delta,{\rm cp},n}
 =\frac{\ell mh}{N_D}
 \sum_{b=1}^{J_{\rm cp}}\norm{d_{b,n}}^2
 +\frac{\ell\sqrt{mh}}{N_D}
 \left\{\sum_{b=1}^{J_{\rm cp}}\norm{d_{b,n}}^2\right\}^{1/2},
\end{equation*}
and
\begin{equation*}
 \mathfrak e_{{\rm cp},n}
 =\sqrt p\{\mathfrak r_n+\mathfrak r_{\delta,{\rm cp},n}\}
 +\zeta_{p,N}.
\end{equation*}
Lemma~\ref{prop:alternative-DB} shows that
$\mathfrak r_{\delta,{\rm cp},n}$ bounds the additional contribution of all
mean changes to the global difference-based estimator.  Put
\begin{equation*}
 \mathfrak v_n^{\rm WBS}=(n\log n)^{1/\nu}.
\end{equation*}
Lemma~\ref{lem:wbs-coverage} supplies a constant $c_{\rm W}>0$.  Define the
random-interval coverage bound
\begin{equation*}
 \pi_n^{\rm WBS}
 =J_{\rm cp}\left\{1-c_{\rm W}
 \left(\frac{\Delta_n}{n}\right)^2\right\}^{M_n^{\rm WBS}}.
\end{equation*}

\begin{theorem}[Exact recovery by WBS and refined localization]
\label{thm:multiple-consistency}
Suppose Assumptions~\ref{ass:ARS-dependence}--
\ref{ass:spectrum-growth} and~\ref{ass:multiple-configuration} hold.  Apply
Algorithm~\ref{alg:DB-RWBS} with a deterministic threshold $\Lambda_n$.  If
\begin{equation}
 \begin{gathered}
 \mathfrak e_{{\rm cp},n}\longrightarrow0,
 \qquad
 \pi_n^{\rm WBS}\longrightarrow0,\\
 \mathfrak v_n^{\rm WBS}\sqrt n=o(\Lambda_n),
 \qquad
 \Lambda_n=o(\Delta_n\delta_{\min,n}^2),
 \end{gathered}
 \label{eq:multiple-threshold-condition}
\end{equation}
then
\begin{equation}
 \Pp\{\widehat J_{\rm cp}=J_{\rm cp}\}\longrightarrow1,
 \qquad
 \max_{1\le b\le J_{\rm cp}}
 |\widetilde\tau_{b,n}-\tau_{b,n}|
 =O_{\Pp}\left(
 \frac{(\mathfrak v_n^{\rm WBS})^2}{\delta_{\min,n}^2}
 +\frac{\mathfrak v_n^{\rm WBS}\sqrt n}{\delta_{\min,n}}
 \right)
 =o_{\Pp}(\Delta_n).
 \label{eq:preliminary-exact-recovery}
\end{equation}
The refined locations satisfy
\begin{equation}
 \max_{1\le b\le J_{\rm cp}}
 |\widehat\tau_{b,n}-\tau_{b,n}|
 =O_{\Pp}\left(
 \frac{(\mathfrak v_n^{\rm WBS})^2}{\delta_{\min,n}^2}
 +\frac{\mathfrak v_n^{\rm WBS}\sqrt{\Delta_n}}
 {\delta_{\min,n}}
 \right)
 =o_{\Pp}(\Delta_n).
 \label{eq:multiple-localization-rate}
\end{equation}
Moreover,
\begin{equation*}
 \Pp\{\widehat\tau_{1,n}<\cdots<\widehat\tau_{J_{\rm cp},n}\}
 \longrightarrow1.
\end{equation*}
The probability is taken jointly over the functional observations and the
independent WBS interval draw.
\end{theorem}

The lower threshold condition controls three effects simultaneously: the
maximum local quadratic noise, the signal--noise cross term on random
intervals, and the residual effect of a previously localized change near a
recursive boundary.  The upper condition separates the threshold from the
smallest isolating-interval drift.  The lower bound on the original jumps is a
primitive identifiability condition for exact multiple-change recovery.  No
sub-Gaussian concentration condition is used.

\begin{corollary}[A convenient WBS specification]
\label{cor:multiple-fixed-jumps}
Suppose Assumptions~\ref{ass:ARS-dependence}--
\ref{ass:spectrum-growth} and~\ref{ass:multiple-configuration} hold, and
$\mathfrak e_{{\rm cp},n}\to0$.  Let
\begin{equation*}
 M_n^{\rm WBS}=\lceil C_M\log n\rceil,
 \qquad
 \Lambda_n=n^{\gamma_{\rm W}},
 \qquad
 C_M>0,\quad
 \gamma_{\rm W}\in(1/2+1/\nu,1).
\end{equation*}
Then \eqref{eq:preliminary-exact-recovery} and
\eqref{eq:multiple-localization-rate} hold.
\end{corollary}

For finite-sample threshold selection, the Gaussian construction used for the
single-change test can be applied conditionally on the realized WBS intervals.
Generate common standard Gaussian increments $Z_{ij}^{*}$ and put
$G_j^{*}(k)=\sum_{i=1}^{k}Z_{ij}^{*}$.  For $I=(s,e]$ and $k\in\cK_I$, define
\begin{equation*}
 U_{j,I}^{*}(k)
 =\left\{\frac{n_I}{(k-s)(e-k)}\right\}^{1/2}
 \left[G_j^{*}(k)-G_j^{*}(s)
 -\frac{k-s}{n_I}\{G_j^{*}(e)-G_j^{*}(s)\}\right]
\end{equation*}
and
\begin{equation*}
 D_{I,a}^{*}(k)
 =\frac{\sum_{j=1}^{p}\widehat w_{j,n}(a)
 \{U_{j,I}^{*}(k)^2-1\}}
 {\{2\widehat A_{2,n}(a)\}^{1/2}}.
\end{equation*}
A practical working threshold is a high conditional quantile of
\begin{equation*}
 \max_{1\le r_{\rm W}\le M_n^{\rm WBS}}
 \max_{k\in\cK_{I_{r_{\rm W}}^{\rm WBS}}}
 \max_{a\in\cG}D_{I_{r_{\rm W}}^{\rm WBS},a}^{*}(k).
\end{equation*}
The same Gaussian increments are shared across all overlapping intervals and
all ridge values, and the realized WBS intervals are held fixed across
reference draws.  This working calibration is not used in
Theorem~\ref{thm:multiple-consistency}, which is stated for a deterministic
diverging threshold.

\section{Simulation studies}
\label{sec:implementation}

\subsection{Simulation design and competing methods}

We compare the proposed single-change test with three established procedures. We refer to the proposed Cauchy combination of ridge-regularized Hotelling statistics as RHT. The competitors are the fully functional test of \citet{AueRiceSonmez2018}, denoted by ARS, the functional principal component test of \citet{BerkesEtAl2009}, denoted by BGHK, and the self-normalized functional principal component test of \citet{ZhangEtAl2011}, denoted by SN.

Let $\{\phi_j:j\ge1\}$ be the Fourier basis introduced in Section~\ref{sec:method}. For the power study, we take $(n,p)=(200,21)$ and generate
\begin{equation*}
 X_i(t)=\sum_{j=1}^{p}\varepsilon_{ij}\phi_j(t)
 +\delta(t)\ind(i>n/2),
 \qquad i=1,\ldots,n.
\end{equation*}
Thus the change fraction is $\vartheta=1/2$. We consider three innovation standard-deviation sequences:
\begin{enumerate}[label=(\roman*),leftmargin=2em]
\item finite rank: $\sigma_j=1$ for $j\le3$ and $\sigma_j=0$ for $j>3$;
\item fast decay: $\sigma_j=3^{-j}$;
\item slow decay: $\sigma_j=j^{-1}$.
\end{enumerate}
These designs are referred to as Settings 1--3, respectively.

Under independent errors,
\begin{equation*}
 \beps_i=(\varepsilon_{i1},\ldots,\varepsilon_{ip})^{\top}
 \stackrel{\mathrm{iid}}{\sim}
 N_p(\bm 0,\bm\Sigma_z),
 \qquad
 \bm\Sigma_z=\diag(\sigma_1^2,\ldots,\sigma_p^2).
\end{equation*}
Under the functional autoregressive design,
\begin{equation*}
 \beps_i=\bm\Psi\beps_{i-1}+\bm z_i,
 \qquad
 \bm z_i\stackrel{\mathrm{iid}}{\sim}
 N_p(\bm 0,\bm\Sigma_z).
\end{equation*}
For each replication, a matrix $\bm A$ is generated with entries $A_{jr}=\sigma_j\sigma_rG_{jr}$, where the $G_{jr}$ are independent standard normal variables, and
\begin{equation*}
 \bm\Psi=0.5\,\frac{\bm A}{\norm{\bm A}_{\op}}.
\end{equation*}
Hence $\norm{\bm\Psi}_{\op}=0.5$. We discard 100 burn-in observations. The long-run covariance matrix of the coefficient sequence is
\begin{equation*}
 \bOmega=
 \begin{cases}
 \bm\Sigma_z, & \text{independent errors},\\[2pt]
 (\bI_p-\bm\Psi)^{-1}\bm\Sigma_z
 (\bI_p-\bm\Psi)^{-\top}, & \text{FAR errors}.
 \end{cases}
\end{equation*}

The mean change is distributed over a random collection of Fourier directions. For $s_{\rm sig}\in\{5,20\}$, let $\cS_{s_{\rm sig}}$ be drawn uniformly from all subsets of $\{1,\ldots,p\}$ having cardinality $s_{\rm sig}$, independently in every replication, and set
\begin{equation*}
 \delta(t)=b_{\rm sig}\sum_{j\in\cS_{s_{\rm sig}}}\phi_j(t).
\end{equation*}
The innovation-variance ordering remains fixed and only the signal support is randomized. The amplitude is chosen to attain the prescribed signal-to-noise ratio
\begin{equation*}
 \mathrm{SNR}
 =\frac{\vartheta(1-\vartheta)\norm{\delta}^2}
 {\tr(\bOmega)},
 \qquad
 b_{\rm sig}=\left\{
 \frac{\mathrm{SNR}\tr(\bOmega)}
 {s_{\rm sig}\vartheta(1-\vartheta)}
 \right\}^{1/2}.
\end{equation*}
For FAR errors, this normalization is applied conditionally on the matrix $\bm\Psi$ generated in that replication. We use
\begin{equation*}
 \mathrm{SNR}\in\{0,0.003,0.005,0.006,0.0075,0.009,0.01\}.
\end{equation*}

RHT uses the ridge grid $\cG=\{0.1,0.2,0.4,0.8\}$. The grid spans moderately different spectral filters, and no single value is privileged because the local-power criterion in \eqref{eq:oracle-ridge-definitions} is alternative-specific. All single-change scans use $\epsilon=0.1$. The difference filter has order $m=3$ and normalized coefficients
\begin{equation*}
 (c_0,c_1,c_2,c_3)=(0.1942,0.2809,0.3832,-0.8582).
\end{equation*}
We use the kernel $K(x)=(1-x^2)_+$. The bandwidth and spacing are $(\ell,h)=(3,6)$ for independent errors and $(\ell,h)=(5,10)$ for FAR errors. The long-run covariance estimate is replaced by its positive spectral part before ridge inversion. RHT is calibrated with 1,999 common Gaussian-bridge paths. The observed and simulated Cauchy statistics are ranked symmetrically within the same batch, providing a finite-simulation implementation of the joint bridge calibration in Section~\ref{sec:method}.

ARS uses its fully functional maximum-CUSUM statistic and the long-run covariance bandwidth $\lfloor n^{1/4}\rfloor$. BGHK uses its functional principal component Cram\'er--von Mises statistic, whereas SN uses its self-normalized maximum-CUSUM statistic. The two principal component methods retain the smallest number of components explaining more than 85\% of the sample variation. SN retains at most 10 components. The reference distributions for ARS, BGHK, and SN are also approximated with 1,999 simulations.

Every power estimate is based on 100 paired replications, so all four methods are applied to the same data in each replication. No method is size-corrected under independent errors. Since BGHK is designed for independent observations and substantially over-rejects under FAR errors, its FAR power is evaluated with a setting-specific critical value obtained from 1,000 independent null replications. The other three procedures retain their original nominal 5\% calibrations.

\subsection{Empirical size}

For the size experiment, we set $\delta\equiv0$ and consider all combinations of $n\in\{200,400\}$ and $p\in\{21,41\}$. When $p=41$, the Fourier expansion and each innovation spectrum are extended through coordinate 41. All other tuning parameters remain unchanged. Each entry in Table~\ref{tab:four-method-size} is based on 1,000 replications and uses the original nominal calibration of the corresponding method. In particular, the FAR entries for BGHK are not size-corrected.

\begin{table}[!htbp]
\caption{Empirical size at the nominal 5\% level. Each entry is based on 1,000 replications and uses the original nominal calibration of the corresponding method.}
\label{tab:four-method-size}
\centering
\scriptsize
\setlength{\tabcolsep}{4.2pt}
\begin{tabular}{ccccrrrr}
\toprule
$n$ & $p$ & Errors & Setting & RHT & ARS & BGHK & SN\\
\midrule
200 & 21 & IID & 1 & 0.068 & 0.055 & 0.055 & 0.036\\
200 & 21 & IID & 2 & 0.063 & 0.050 & 0.058 & 0.050\\
200 & 21 & IID & 3 & 0.067 & 0.050 & 0.057 & 0.043\\
200 & 21 & FAR & 1 & 0.062 & 0.066 & 0.085 & 0.035\\
200 & 21 & FAR & 2 & 0.066 & 0.080 & 0.143 & 0.035\\
200 & 21 & FAR & 3 & 0.052 & 0.053 & 0.084 & 0.035\\
\midrule
200 & 41 & IID & 1 & 0.062 & 0.041 & 0.043 & 0.034\\
200 & 41 & IID & 2 & 0.058 & 0.040 & 0.048 & 0.049\\
200 & 41 & IID & 3 & 0.070 & 0.054 & 0.047 & 0.043\\
200 & 41 & FAR & 1 & 0.064 & 0.069 & 0.086 & 0.040\\
200 & 41 & FAR & 2 & 0.083 & 0.082 & 0.138 & 0.050\\
200 & 41 & FAR & 3 & 0.066 & 0.060 & 0.077 & 0.027\\
\midrule
400 & 21 & IID & 1 & 0.059 & 0.040 & 0.056 & 0.036\\
400 & 21 & IID & 2 & 0.052 & 0.051 & 0.048 & 0.033\\
400 & 21 & IID & 3 & 0.052 & 0.040 & 0.054 & 0.045\\
400 & 21 & FAR & 1 & 0.059 & 0.071 & 0.093 & 0.052\\
400 & 21 & FAR & 2 & 0.070 & 0.061 & 0.147 & 0.048\\
400 & 21 & FAR & 3 & 0.074 & 0.061 & 0.089 & 0.041\\
\midrule
400 & 41 & IID & 1 & 0.067 & 0.056 & 0.056 & 0.048\\
400 & 41 & IID & 2 & 0.059 & 0.038 & 0.048 & 0.040\\
400 & 41 & IID & 3 & 0.062 & 0.038 & 0.052 & 0.043\\
400 & 41 & FAR & 1 & 0.058 & 0.056 & 0.077 & 0.035\\
400 & 41 & FAR & 2 & 0.074 & 0.063 & 0.135 & 0.042\\
400 & 41 & FAR & 3 & 0.065 & 0.071 & 0.108 & 0.041\\
\bottomrule
\end{tabular}
\end{table}

Table~\ref{tab:four-method-size} shows that RHT and ARS maintain broadly satisfactory size across the covariance spectra and dimensions. RHT is mildly liberal in some designs, whereas SN is generally conservative. The original BGHK calibration is adequate under independence but becomes liberal under FAR dependence, especially under fast spectral decay. This pattern motivates the setting-specific size correction applied only to BGHK in the FAR power experiment.

\subsection{Empirical power}

Figures~\ref{fig:four-method-power-iid} and~\ref{fig:four-method-power-far} report the power curves over the full signal grid. The BGHK curves under FAR errors use the size-corrected critical values described above. All other curves use the original nominal 5\% calibrations.

Across both dependence structures, RHT gains power earlier and approaches high power more rapidly than the competing procedures. Its advantage is most pronounced under the finite-rank and fast-decay spectra, where relevant changes can load on directions that are poorly represented by a leading-principal-component truncation. The competitors become more effective under slow spectral decay, but RHT remains consistently strong for both signal-support sizes. The similarity of the IID and FAR patterns indicates that the difference-based long-run covariance estimator preserves the benefit of ridge adaptation under serial dependence.

\begin{figure}[!htbp]
\centering
\includegraphics[width=0.84\textwidth]{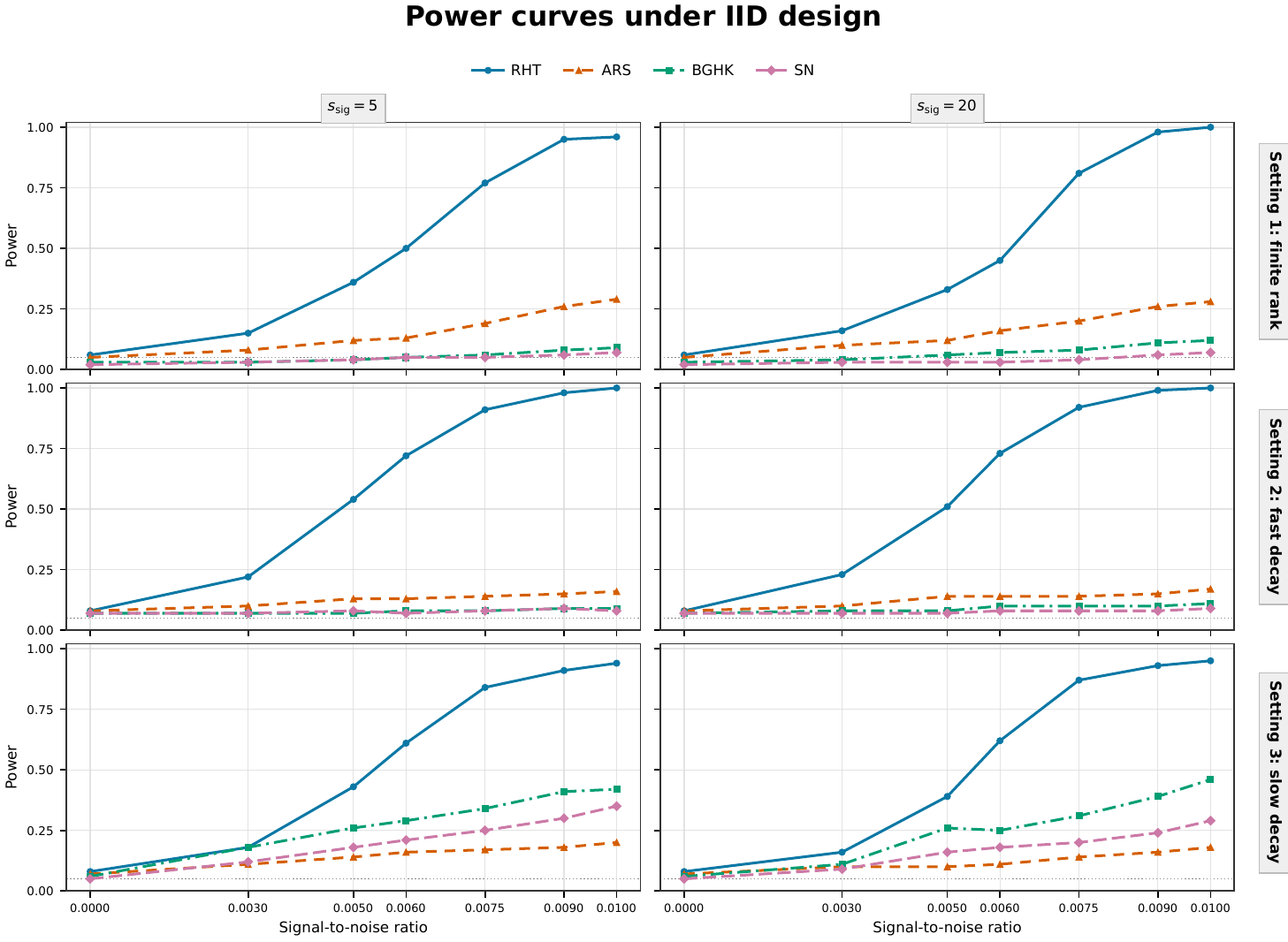}
\caption{Empirical power under IID errors. The rows correspond to the finite-rank, fast-decay, and slow-decay innovation spectra, and the columns correspond to $s_{\rm sig}=5$ and $s_{\rm sig}=20$ randomly selected signal directions. Each point is based on 100 paired replications.}
\label{fig:four-method-power-iid}
\end{figure}

\begin{figure}[!htbp]
\centering
\includegraphics[width=0.84\textwidth]{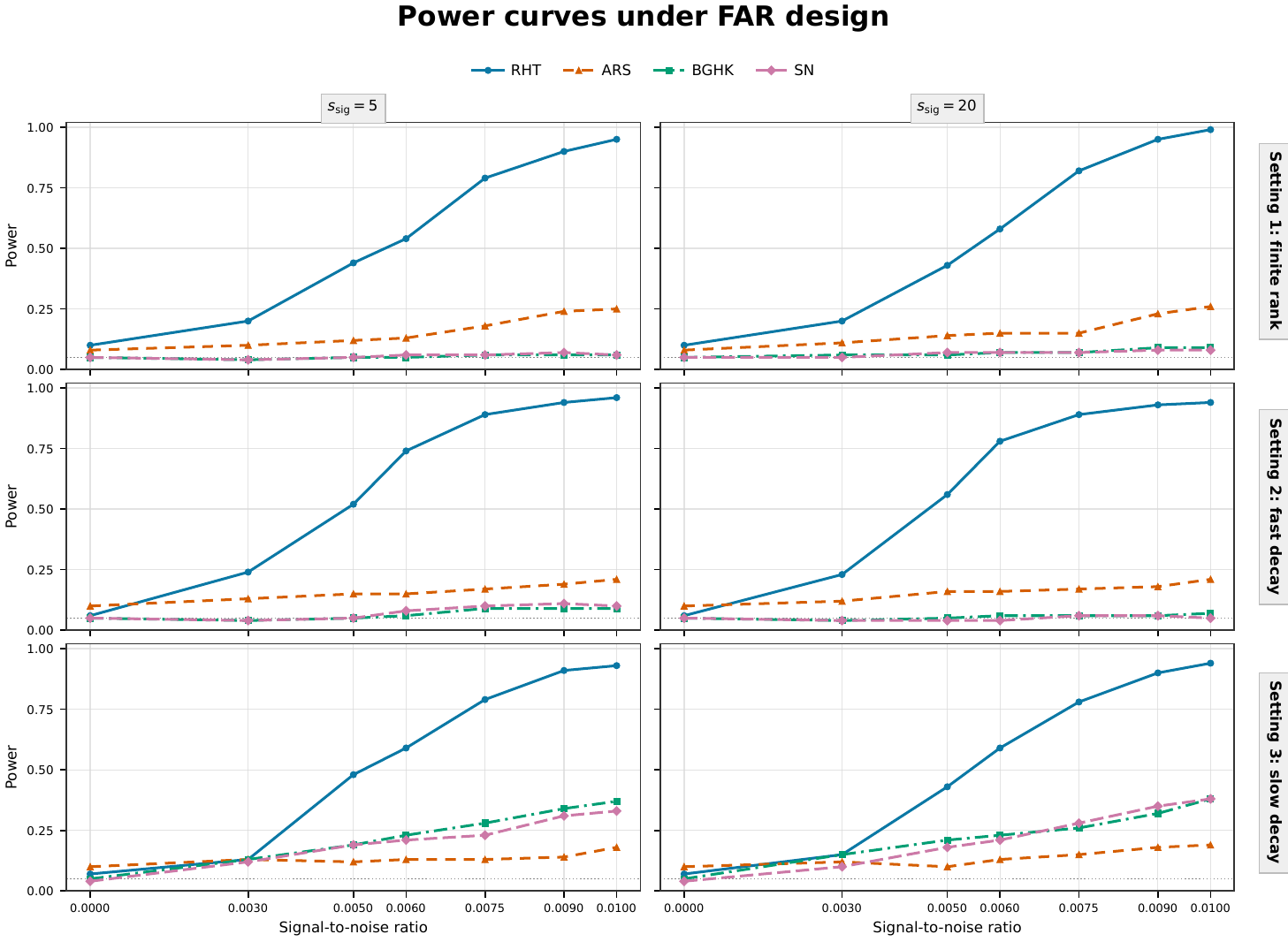}
\caption{Empirical power under FAR errors with $\|\bm\Psi\|_{\rm op}=0.5$. The panel arrangement is the same as in Figure~\ref{fig:four-method-power-iid}. BGHK uses setting-specific size-corrected critical values, whereas the other methods retain their original calibrations. Each point is based on 100 paired replications.}
\label{fig:four-method-power-far}
\end{figure}

\subsection{Single-change localization}

We next examine the finite-sample accuracy of the single-change location
estimators under the same independent and FAR designs.  We set $n=200$,
$p=21$, $\vartheta=1/2$, and $\mathrm{SNR}=0.1$.  Hence the true change
location is $\tau_n=100$.  The three covariance settings and
$s_{\rm sig}\in\{5,20\}$ are retained.  Each design cell is based on 100
paired replications.

For RHT, we use the adaptive estimator $\widehat\tau_n^{\rm RHT}$ defined
after Theorem~\ref{thm:consistency-location}.  The marginal bridge values are
computed by the same finite Monte Carlo batch used for testing.  The selected
location is the maximizer from the ridge scan having the smallest marginal
$p$-value.  ARS and BGHK use the locations returned by their corresponding
scan procedures, whereas SN uses the maximizer of its self-normalized FPCA
scan.  The location estimates are evaluated unconditionally under the
alternative.  The FAR size correction for BGHK is irrelevant because it
changes only the rejection threshold.

Let $N_{\rm rep}=100$.  For a method $\mathsf M$, write
$\widehat\tau_{\mathsf M}^{(r_{\rm rep})}$ for the estimate in replication
$r_{\rm rep}$.  We report
\begin{equation*}
 \mathrm{MAE}_{\mathsf M}
 =\frac{1}{N_{\rm rep}}\sum_{r_{\rm rep}=1}^{N_{\rm rep}}
 \left|\widehat\tau_{\mathsf M}^{(r_{\rm rep})}-\tau_n\right|,
 \qquad
 \mathrm{RMSE}_{\mathsf M}
 =\left\{\frac{1}{N_{\rm rep}}\sum_{r_{\rm rep}=1}^{N_{\rm rep}}
 \left(\widehat\tau_{\mathsf M}^{(r_{\rm rep})}-\tau_n\right)^2
 \right\}^{1/2}.
\end{equation*}

\begin{table}[!htbp]
\caption{Single-change localization accuracy at $\mathrm{SNR}=0.1$ with
$(n,p)=(200,21)$ and $\tau_n=100$. Entries are MAE (RMSE), measured in
observation indices. Each entry is based on 100 paired replications.}
\label{tab:four-method-location}
\centering
\scriptsize
\setlength{\tabcolsep}{3.8pt}
\begin{tabular}{lccrrrr}
\toprule
Errors & Setting & $s_{\rm sig}$ & RHT & ARS & BGHK & SN\\
\midrule
IID & 1 & 5  & 0.01 (0.10) & 0.28 (0.72) & 12.24 (19.60) & 16.37 (25.48)\\
IID & 1 & 20 & 0.00 (0.00) & 0.42 (0.84) & 7.93 (13.95)  & 11.78 (17.73)\\
IID & 2 & 5  & 0.00 (0.00) & 0.54 (1.38) & 3.88 (10.98)  & 5.51 (13.37)\\
IID & 2 & 20 & 0.00 (0.00) & 0.54 (1.30) & 2.32 (8.30)   & 3.52 (8.56)\\
IID & 3 & 5  & 0.00 (0.00) & 0.42 (1.10) & 0.03 (0.17)   & 0.53 (0.89)\\
IID & 3 & 20 & 0.02 (0.14) & 0.29 (0.78) & 0.01 (0.10)   & 0.44 (0.72)\\
\midrule
FAR & 1 & 5  & 0.02 (0.14) & 0.51 (1.50) & 9.67 (16.01)  & 14.25 (23.73)\\
FAR & 1 & 20 & 0.04 (0.20) & 0.27 (0.79) & 9.50 (16.01)  & 11.91 (18.81)\\
FAR & 2 & 5  & 0.01 (0.10) & 0.44 (1.33) & 13.52 (24.44) & 14.18 (25.44)\\
FAR & 2 & 20 & 0.03 (0.22) & 0.52 (1.52) & 10.19 (20.36) & 11.84 (22.84)\\
FAR & 3 & 5  & 0.09 (0.41) & 0.44 (1.14) & 0.22 (0.73)   & 0.54 (1.19)\\
FAR & 3 & 20 & 0.04 (0.20) & 0.34 (0.92) & 0.11 (0.52)   & 0.62 (1.16)\\
\bottomrule
\end{tabular}
\end{table}

Table~\ref{tab:four-method-location} shows that RHT provides the most accurate and stable localization across the designs. ARS is consistently competitive but less precise. BGHK and SN perform well under slow spectral decay, whereas their errors increase under the finite-rank and fast-decay spectra, especially under FAR dependence. This pattern is consistent with information loss from principal-component truncation when a change loads on low-variance directions. Thus the testing gains of RHT are accompanied by accurate change-point dating.

\subsection{Multiple-change localization}

We finally examine the recovery of two changes under the same independent and
FAR designs.  We retain $(n,p)=(200,21)$, the three innovation spectra, and
$s_{\rm sig}\in\{5,20\}$.  The true locations are
\begin{equation*}
 \mathfrak T_0=\{\tau_{1,n},\tau_{2,n}\}=\{70,130\},
\end{equation*}
and the mean follows the temporary-change pattern
\begin{equation*}
 \mu_i(t)=
 \begin{cases}
  0, & 1\le i\le70,\\
  \delta_{\rm mc}(t), & 71\le i\le130,\\
  0, & 131\le i\le200,
 \end{cases}
 \qquad
 \delta_{\rm mc}(t)
 =b_{\rm sig}\sum_{j\in\cS_{s_{\rm sig}}}\phi_j(t).
\end{equation*}
Equivalently, the two jumps in model~\eqref{eq:multiple-change-model} are
$\delta_{1,n}=\delta_{\rm mc}$ and
$\delta_{2,n}=-\delta_{\rm mc}$.  The support
\(\cS_{s_{\rm sig}}\) is redrawn independently in every replication,
while the ordering of the innovation variances remains fixed.  The amplitude
is chosen as
\begin{equation*}
 b_{\rm sig}
 =\left\{
 \frac{0.1\tr(\bOmega)}
 {s_{\rm sig}(0.35)(0.65)}
 \right\}^{1/2}.
\end{equation*}
Thus the positive jump at \(\tau_{1,n}\) and the negative jump at
\(\tau_{2,n}\) have the same trace-based signal-to-noise ratio of 0.1.

For RHT, Algorithm~\ref{alg:DB-RWBS} is implemented with
\(M_n^{\rm WBS}=500\) random intervals.  In the finite-sample scans, we use
\begin{equation*}
 \cK_{I,\epsilon_{\rm W}}
 =\left\{k\in\cK_I:
 k-s\ge\lceil\epsilon_{\rm W}n_I\rceil,\quad
 e-k\ge\lceil\epsilon_{\rm W}n_I\rceil\right\},
 \qquad \epsilon_{\rm W}=0.1,
\end{equation*}
in place of \(\cK_I\).  The difference-based tuning parameters remain
$(\ell,h)=(3,6)$ for independent errors and $(\ell,h)=(5,10)$ for FAR
errors.  The 5\% stopping threshold is computed from 1,999 shared Gaussian
reference paths conditional on the realized WBS intervals.  ARS and BGHK use
the binary-segmentation routines in \texttt{fChange}.  BGHK is applied with
its original stopping rule and without the FAR size correction used in the
single-change power experiment.  SN recursively recomputes its FPCA
self-normalized scan on each candidate segment.  The comparator reference
distributions also use 1,999 draws.  Each design cell is based on
$N_{\rm rep}=100$ paired replications.

For method \(\mathsf M\) in replication \(r_{\rm rep}\), let
\(\widehat{\mathfrak T}_{\mathsf M}^{(r_{\rm rep})}\) denote the estimated
change-point set.  We report the two directed Hausdorff components
\begin{align}
 \operatorname{HD}_{\rm T\to E}
 \left(\mathfrak T_0,
 \widehat{\mathfrak T}_{\mathsf M}^{(r_{\rm rep})}\right)
 &=\max_{\tau\in\mathfrak T_0}
 \min_{\widehat\tau\in
 \widehat{\mathfrak T}_{\mathsf M}^{(r_{\rm rep})}}
 |\tau-\widehat\tau|,
 \label{eq:multiple-HD-true-estimated}\\
 \operatorname{HD}_{\rm E\to T}
 \left(\widehat{\mathfrak T}_{\mathsf M}^{(r_{\rm rep})},
 \mathfrak T_0\right)
 &=\max_{\widehat\tau\in
 \widehat{\mathfrak T}_{\mathsf M}^{(r_{\rm rep})}}
 \min_{\tau\in\mathfrak T_0}
 |\widehat\tau-\tau|.
 \label{eq:multiple-HD-estimated-true}
\end{align}
The first measure reflects missed or poorly localized true changes.  The
second reflects estimates that are not close to a true change.  When a
procedure returns the empty set, both quantities are assigned the maximal
finite penalty \(n=200\).  The tables report their Monte Carlo means and
standard deviations.

\begin{table}[!htbp]
\caption{Multiple-change localization measured by the true-to-estimated
distance in \eqref{eq:multiple-HD-true-estimated}. Entries are mean
(standard deviation) over 100 paired replications.}
\label{tab:multiple-HD-true-estimated}
\centering
\scriptsize
\setlength{\tabcolsep}{4pt}
\begin{tabular}{lccrrrr}
\toprule
Errors & Setting & $s_{\rm sig}$ & RHT & ARS & BGHK & SN\\
\midrule
IID & 1 & 5  & 0.01 (0.10) & 16.87 (54.30) & 180.18 (57.13) & 190.11 (39.67)\\
IID & 1 & 20 & 0.00 (0.00) & 39.58 (78.13) & 162.39 (74.00) & 194.35 (28.26)\\
IID & 2 & 5  & 0.00 (0.00) & 154.91 (82.94) & 72.07 (94.73) & 195.99 (23.53)\\
IID & 2 & 20 & 0.00 (0.00) & 164.69 (75.77) & 31.99 (71.20) & 198.21 (17.90)\\
IID & 3 & 5  & 0.04 (0.24) & 88.82 (99.06) & 0.16 (0.44) & 191.62 (37.60)\\
IID & 3 & 20 & 0.02 (0.14) & 83.08 (97.97) & 0.10 (0.39) & 194.86 (29.85)\\
\midrule
FAR & 1 & 5  & 0.16 (0.49) & 41.42 (79.72) & 156.98 (77.96) & 190.96 (36.54)\\
FAR & 1 & 20 & 0.17 (0.47) & 47.50 (83.80) & 150.02 (83.33) & 195.86 (23.78)\\
FAR & 2 & 5  & 0.41 (1.50) & 127.12 (95.62) & 96.93 (99.74) & 193.53 (32.73)\\
FAR & 2 & 20 & 0.63 (3.34) & 126.55 (96.33) & 81.12 (97.75) & 197.05 (20.81)\\
FAR & 3 & 5  & 0.50 (2.50) & 83.04 (98.02) & 0.40 (1.08) & 190.48 (39.06)\\
FAR & 3 & 20 & 0.16 (0.47) & 87.07 (98.61) & 0.21 (0.67) & 195.78 (25.16)\\
\bottomrule
\end{tabular}
\end{table}

\begin{table}[!htbp]
\caption{Multiple-change localization measured by the estimated-to-true
distance in \eqref{eq:multiple-HD-estimated-true}. Entries are mean
(standard deviation) over 100 paired replications.}
\label{tab:multiple-HD-estimated-true}
\centering
\scriptsize
\setlength{\tabcolsep}{4pt}
\begin{tabular}{lccrrrr}
\toprule
Errors & Setting & $s_{\rm sig}$ & RHT & ARS & BGHK & SN\\
\midrule
IID & 1 & 5  & 0.01 (0.10) & 19.70 (54.27) & 179.21 (59.56) & 188.59 (45.42)\\
IID & 1 & 20 & 0.28 (2.80) & 41.80 (77.46) & 160.86 (76.50) & 192.55 (36.69)\\
IID & 2 & 5  & 0.00 (0.00) & 155.53 (81.89) & 73.88 (93.62) & 194.79 (29.80)\\
IID & 2 & 20 & 1.27 (8.31) & 165.40 (74.31) & 35.33 (70.30) & 198.21 (17.90)\\
IID & 3 & 5  & 0.63 (5.24) & 90.48 (97.82) & 3.46 (10.12) & 190.64 (41.01)\\
IID & 3 & 20 & 0.02 (0.14) & 84.83 (96.78) & 1.39 (6.77) & 194.26 (32.81)\\
\midrule
FAR & 1 & 5  & 1.27 (6.99) & 46.11 (78.20) & 156.81 (77.71) & 188.78 (44.68)\\
FAR & 1 & 20 & 2.79 (11.75) & 51.10 (82.40) & 150.07 (82.89) & 194.40 (32.01)\\
FAR & 2 & 5  & 12.21 (20.74) & 129.38 (92.96) & 108.12 (89.95) & 192.59 (36.53)\\
FAR & 2 & 20 & 9.29 (17.87) & 129.27 (93.23) & 91.26 (90.58) & 196.21 (26.68)\\
FAR & 3 & 5  & 6.64 (14.79) & 85.62 (96.25) & 8.18 (15.87) & 188.80 (44.61)\\
FAR & 3 & 20 & 3.51 (11.37) & 89.10 (97.12) & 8.96 (15.70) & 194.74 (30.08)\\
\bottomrule
\end{tabular}
\end{table}

Tables~\ref{tab:multiple-HD-true-estimated} and
\ref{tab:multiple-HD-estimated-true} show that RHT recovers the two changes
accurately and consistently across the covariance spectra and dependence
structures.  Its true-to-estimated distance remains uniformly small.  The
estimated-to-true distance is also well controlled, although FAR errors under
fast spectral decay occasionally produce additional estimates.  BGHK is
competitive under slow spectral decay but becomes substantially less stable
in the other spectral regimes.  ARS and SN frequently fail to recover both
changes under their stated stopping rules.  Since empty estimates receive the
maximal penalty, the very large competitor distances mainly reflect
nondetection rather than extreme displacement of nonempty estimates.  Overall,
the proposed WBS procedure provides the most stable balance between recovering
all true changes and avoiding spurious locations.

\FloatBarrier

\section{Empirical applications}
\label{sec:applications}

We illustrate the proposed procedures with two temperature data sets that have
markedly different within-curve resolutions.  In each application, all four
methods use the same deterministic Fourier representation.  The empirical
comparisons are descriptive: the estimated changes summarize structural
features of the observed series and are not interpreted as causal climate
attribution.

\subsection{Sydney annual minimum-temperature curves}
\label{subsec:sydney-application}

The first data set contains daily minimum air temperatures recorded at Sydney
(Observatory Hill) from the Australian Bureau of Meteorology Climate Data
Online archive, product \texttt{IDCJAC0011} \citep{BoMCDO2026}.  We use station
066062 through 31 August 2020 and its successor station 066214 thereafter.
Station 066062 is located at approximately $33.86^{\circ}$S,
$151.21^{\circ}$E and 39~m elevation, and commenced operation in 1858
\citep{BoMSydneyStation}.  The two stations overlap before the splice, with a
correlation of 0.998 between their observed daily minima.  The complete record
contains the 167 calendar years from 1859 through 2025.  Leap days are removed,
and 122 of the $167\times365$ observations are linearly interpolated.  Each
year is therefore treated as one functional observation with $N=365$ grid
points.  We retain $p=21$ real Fourier coefficients for every method.

The data are raw station observations rather than a homogenized climate
product.  Station moves, instrumentation changes, and local exposure can
produce nonclimatic discontinuities in long records.  The Australian Climate
Observations Reference Network--Surface Air Temperature data set is designed
to address such effects \citep{BoMACORNSAT}.  The analysis below should thus
be viewed as a statistical study of the observed station series.

The complete-sample tests and the sliding-window stability analysis are
summarized in Table~\ref{tab:sydney-tests}.  The whole-sample calculations use
1,999 reference draws.  For the window analysis, windows of lengths 90, 100,
110, and 120 years are advanced one year at a time, and the reported entries
are rejection frequencies at level 0.05.  All methods consistently detect a change point over the complete observation period. The moving-window analysis further shows that RHT attains higher rejection frequencies than the competing procedures across a broad range of window lengths. This pattern indicates that RHT has greater sensitivity to the underlying structural change and delivers more stable detection performance across different subsamples. The high BGHK frequencies should be interpreted cautiously because
the implementation uses its original independence calibration, while the
residual analysis below reveals dependence in a low-frequency direction.

\begin{table}[!htbp]
\caption{Whole-sample $p$-values and sliding-window rejection frequencies for
the Sydney annual minimum-temperature curves.  The complete sample contains
$n=167$ curves and uses $p=21$ Fourier coefficients.  Window entries are
percentages.}
\label{tab:sydney-tests}
\centering
\begin{tabular}{lrrrr}
\toprule
Summary & RHT & ARS & BGHK & SN\\
\midrule
Whole-sample $p$-value & 0.0005 & $<0.0005$ & $<0.0005$ & 0.0020\\
Window length 90  & 79.49 & 74.36 & 80.77 & 66.67\\
Window length 100 & 95.59 & 83.82 & 89.71 & 72.06\\
Window length 110 & 98.28 & 94.83 & 100.00 & 84.48\\
Window length 120 & 100.00 & 100.00 & 100.00 & 89.58\\
\bottomrule
\end{tabular}
\end{table}

The full-sample RHT scan places the single-change boundary between 1970 and
1971.  Algorithm~\ref{alg:DB-RWBS} retains one change between 1971 and 1972 and
selects no additional break.  The two analyses therefore give the same
substantive localization.  For interpretation, we use the WBS split and compare
1859--1971 with 1972--2025.  The corresponding average daily minimum
temperatures are $13.4893^{\circ}\mathrm C$ and $14.5799^{\circ}\mathrm C$.
Figure~\ref{fig:sydney-mean-curves} shows that the later mean curve lies above
the earlier one through almost the entire calendar.  Its 21-term Fourier
representation remains positive and indicates a broad level shift rather than
a change confined to one season.

\begin{figure}[!htbp]
\centering
\includegraphics[width=0.88\textwidth]{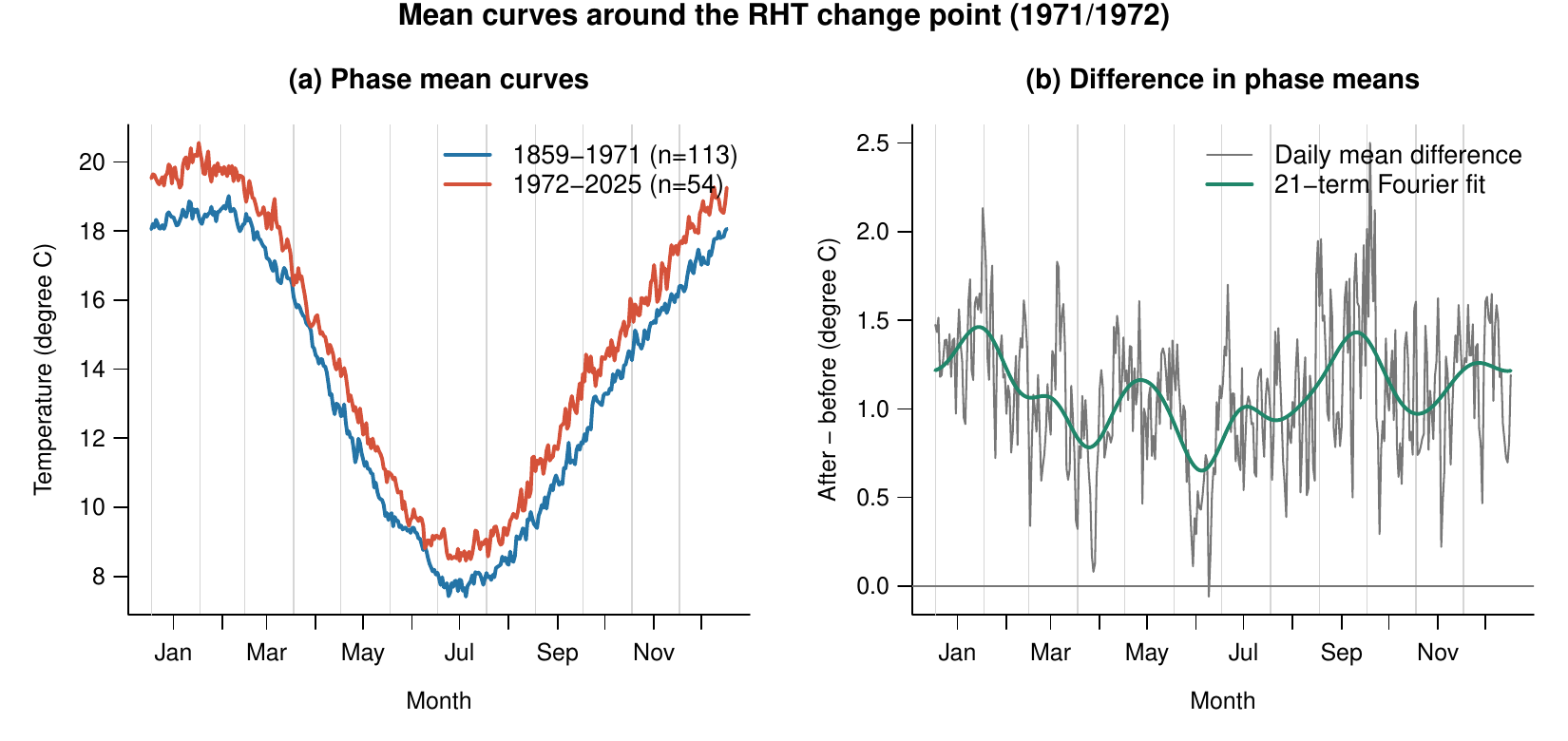}
\caption{Sydney annual minimum-temperature curves around the RHT change point.
The left panel shows the phase mean curves before and after the retained WBS
change.  The right panel shows their daywise difference and its 21-term Fourier
representation.}
\label{fig:sydney-mean-curves}
\end{figure}

The Fourier decomposition in Figure~\ref{fig:sydney-fourier} makes this
interpretation explicit.  The constant coefficient of the fitted contrast is
$1.0906$ and accounts for 96.85\% of its squared $L^2$ energy.  All
nonconstant harmonics together account for only 3.15\%.  The dominant
structural difference is therefore an approximately uniform upward displacement
of the annual curve, with comparatively small changes in seasonal shape.

\begin{figure}[!htbp]
\centering
\includegraphics[width=0.82\textwidth]{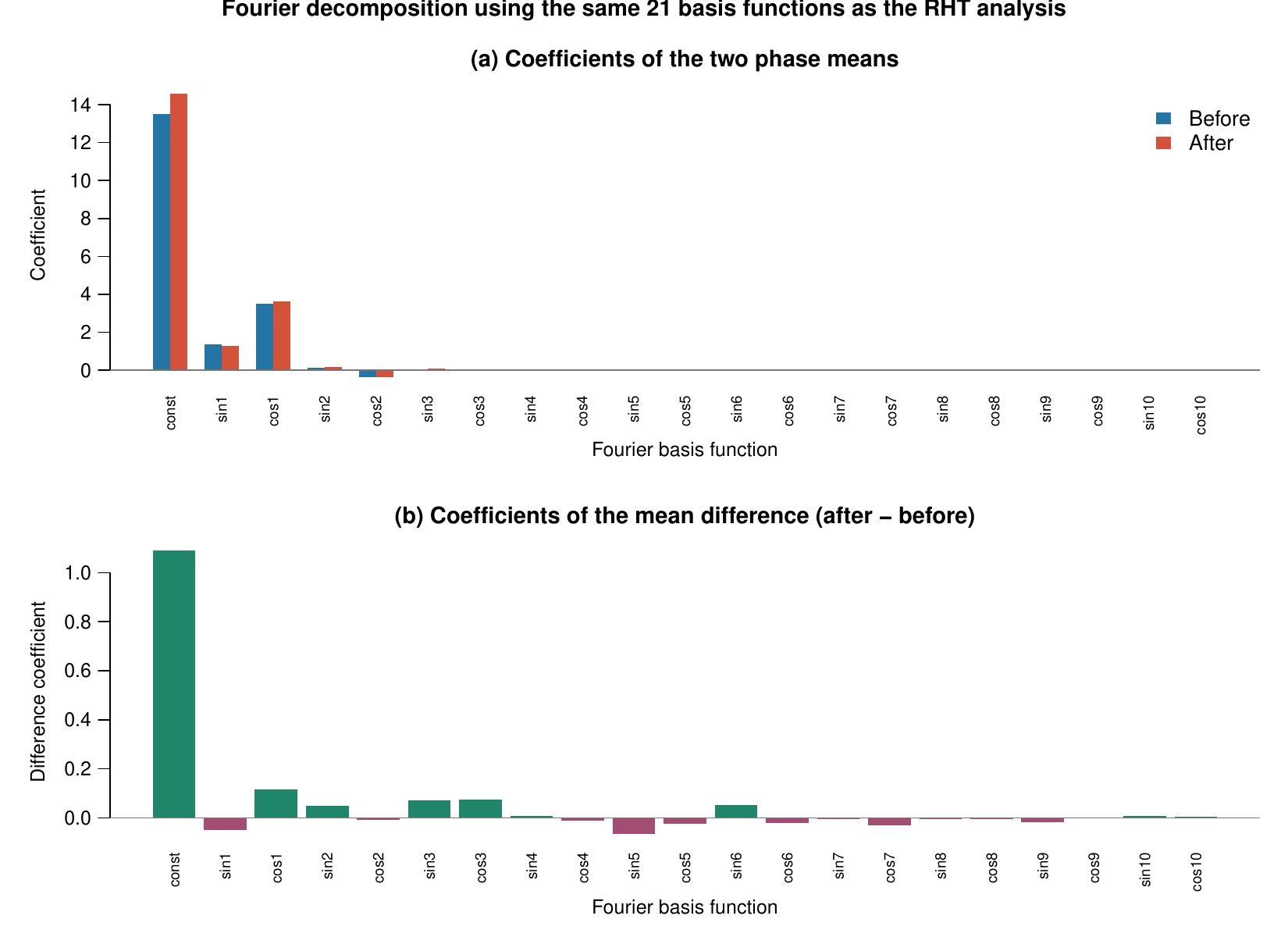}
\caption{Fourier decomposition of the Sydney phase means.  The upper panel
shows the coefficients of the two fitted phase means.  The lower panel shows
the signed coefficients of their difference.}
\label{fig:sydney-fourier}
\end{figure}

To assess temporal dependence after removing the estimated mean shift, we
subtract the phase-specific sample means, exclude lagged pairs crossing the
boundary, and use 4,999 within-phase permutations over lags 1--10.  Table
\ref{tab:sydney-dependence} shows little aggregate autocorrelation for the full
365-day curve.  The annual level direction, however, retains substantial
lag-one dependence.  This concentration of persistence in a low-frequency
component supports the use of a long-run covariance estimator rather than an
independence-only calibration.

\begin{table}[!htbp]
\caption{Residual serial-dependence diagnostics for the Sydney data after
removing the two RHT phase means.  The $p$-values use 4,999 within-phase
permutations over lags 1--10.}
\label{tab:sydney-dependence}
\centering
\begin{tabular}{lrrrr}
\toprule
Representation & Lag-1 ACF & Max. $|\mathrm{ACF}|$ & Lag & $p$-value\\
\midrule
Full 365-day curve & 0.0031 & 0.0124 & 8 & 0.9986\\
Annual level direction & 0.4009 & 0.4009 & 1 & 0.0002\\
\bottomrule
\end{tabular}
\end{table}

The estimated transition is consistent with a broad warming of the observed
station record.  Australia has warmed substantially since national records
began, with most warming occurring after 1950 and with increases evident
throughout the calendar \citep{BoMStateClimate2024}.  A step-change model
summarizes this gradual low-frequency evolution by the boundary that best
separates two mean regimes.  The estimated year should therefore not be read
as the onset of a single physical event.  The raw-station caveat also remains
important, although the Bureau reports no anomalous post-1910 urbanization
trend at Observatory Hill relative to nonurban sites in eastern New South
Wales \citep{BoMACORNSAT}.

\FloatBarrier

\subsection{NOAA ERSSTv6 Ni\~no 1+2 sea-surface-temperature curves}
\label{subsec:nino-application}

The second data set consists of absolute monthly Ni\~no 1+2 sea-surface
temperatures from NOAA's Extended Reconstructed Sea Surface Temperature,
version 6 \citep{NOAANino12,NOAAERSSTv6}.  We retain the 76 complete calendar
years from 1950 through 2025 and omit the incomplete 2026 curve.  Each year is
one functional observation on the common grid of $N=12$ months.  The basis
contains the constant and four sine--cosine pairs, so $p=9$.  These coefficients
reconstruct the monthly curves with relative residual root-sum-of-squares
0.065.  Their sample standard deviations range from 0.075 to 0.893, providing
a direct illustration of heterogeneous score scales.

All candidate boundaries satisfy $k/n\in[0.1,0.9]$.  RHT uses difference order
$m=3$, bandwidth $\ell=5$, spacing $h=10$, the kernel exponent $q=2$, and the
ridge grid $\cG=\{0.1,0.2,0.4,0.8\}$.  The raw long-run covariance estimate has
trace 0.8423.  Its positive spectral part has trace 0.8448; only one small
negative eigenvalue is truncated.  The joint calibration uses 1,999 common
weighted Brownian-bridge paths.  Localization follows the marginal ridge
$p$-value rule described after Theorem~\ref{thm:consistency-location}.

Table~\ref{tab:nino-tests} reports the complete-sample results.  RHT, ARS, and
SN reject stability at level 0.05, whereas BGHK does not.  RHT and ARS agree on
the boundary between 1981 and 1982.  The ridge-specific RHT scans are closely
aligned, and the smallest marginal $p$-value occurs at $a=0.8$.  The selected
path is displayed in the middle panel of Figure~\ref{fig:nino-three-panel}.
The figure also shows that the retained Fourier coordinates have markedly
different scales.

\begin{table}[!htbp]
\caption{Whole-sample single-change results for the 76 annual Ni\~no 1+2
curves.  Locations are reported as the last pre-change year and the first
post-change year.}
\label{tab:nino-tests}
\centering
\begin{tabular}{lcc}
\toprule
Method & $p$-value & Estimated boundary\\
\midrule
RHT  & 0.0365 & 1981/1982\\
ARS  & 0.0380 & 1981/1982\\
BGHK & 0.0925 & 1982/1983\\
SN   & 0.0095 & 1975/1976\\
\bottomrule
\end{tabular}
\end{table}

The fitted segment means differ primarily by a vertical displacement.  Their
calendar averages are $22.690^{\circ}\mathrm C$ for 1950--1981 and
$23.286^{\circ}\mathrm C$ for 1982--2025.  The month-specific increases are
similar across the year, so the estimated contrast is not concentrated in one
season.  As in the Sydney example, this change is a descriptive feature of the
observed regional index rather than a causal attribution statement.

\begin{figure}[!htbp]
\centering
\includegraphics[width=0.88\textwidth]{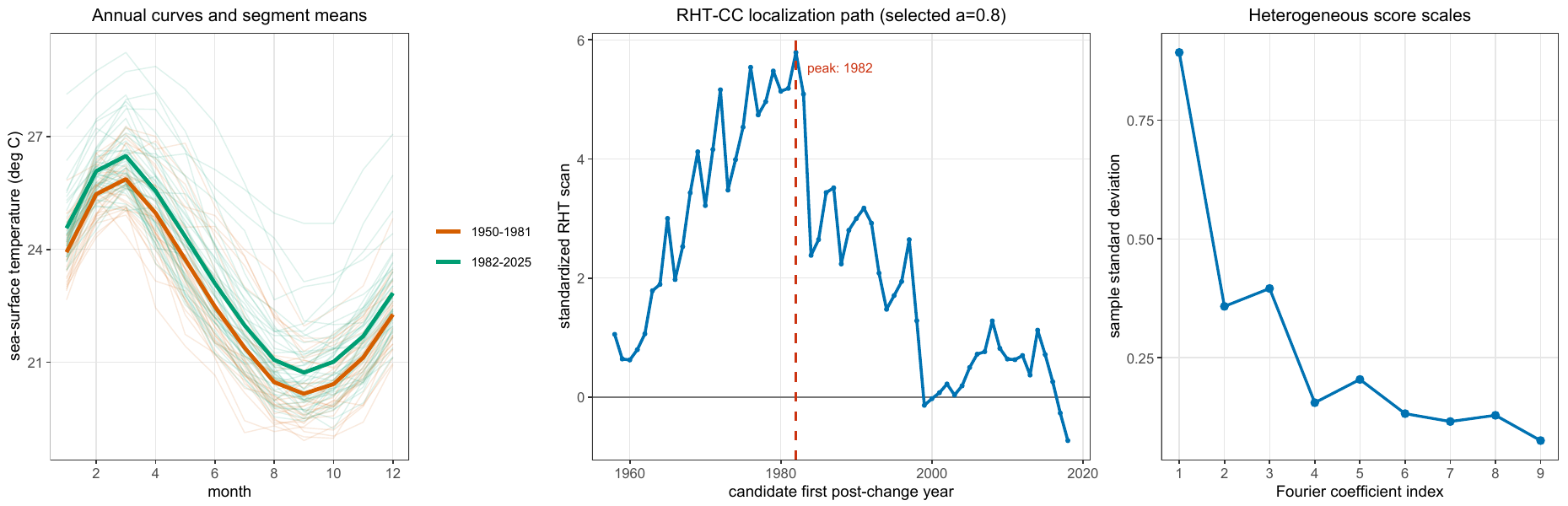}
\caption{NOAA ERSSTv6 Ni\~no 1+2 analysis.  The left panel shows the annual
monthly curves and the two fitted segment means.  The middle panel is the RHT
localization path selected by the smallest marginal ridge $p$-value.  The
right panel shows the sample standard deviations of the retained Fourier
scores.}
\label{fig:nino-three-panel}
\end{figure}

Residual dependence is assessed after subtracting the two fitted segment
means.  We exclude lagged pairs crossing the estimated boundary and use 4,999
within-segment reorderings over lags 1--10.  Table
\ref{tab:nino-dependence} shows no broad evidence of persistence after the
mean change is removed.  The diagnostics nevertheless reveal a borderline
lag-one cross-frequency association.  The evidence is therefore more
consistent with weak and localized dependence than with strict temporal
independence.

\begin{table}[!htbp]
\caption{Residual serial-dependence diagnostics for the Ni\~no 1+2 data after
removing the two fitted segment means.  The $p$-values use 4,999 within-segment
reorderings over lags 1--10.}
\label{tab:nino-dependence}
\centering
\scriptsize
\setlength{\tabcolsep}{4pt}
\begin{tabular}{p{0.39\textwidth}rrp{0.28\textwidth}}
\toprule
Diagnostic & Statistic & $p$-value & Maximizer or detail\\
\midrule
Maximum normalized score lag product
 & 0.4791 & 0.0564 & Lag 1; cross-frequency correlation $-0.4791$\\
Score $L^2$ portmanteau
 & 968.9564 & 0.0692 & Jointly over score pairs and lags\\
Maximum functional autocorrelation
 & 0.1951 & 0.2376 & Maximized at lag 2\\
Annual-mean portmanteau
 & 11.7048 & 0.3060 & Lag-one correlation $-0.0437$\\
\bottomrule
\end{tabular}
\end{table}

Taken together, the two applications illustrate complementary features of the
proposed framework.  The Sydney data contain a strong shift with residual
persistence concentrated in the annual level.  The Ni\~no 1+2 series produces
a more moderate change and only localized evidence of dependence.  In both
cases, ridge regularization accommodates heterogeneous Fourier scales, while
the difference-based long-run covariance estimate avoids treating the ordered
curves as unambiguously independent.

\FloatBarrier

\section{Conclusion}
\label{sec:conclusion}

This paper develops covariance-adaptive inference for one or several mean
changes in weakly dependent functional data.  A growing basis representation,
ridge-regularized Hotelling CUSUMs, and an edge-corrected difference-based
long-run covariance estimator jointly address the spectral decay and serial
dependence intrinsic to functional observations.  For a single change, the
method admits an explicit weighted squared Brownian-bridge calibration.  Its
local-power formula identifies an alternative-specific oracle ridge and gives
a direct theoretical rationale for the Cauchy aggregation over a ridge grid.
For multiple changes, a wild binary segmentation algorithm reuses the global covariance estimate, isolates individual changes through random
intervals, and attains consistent recovery after local refinement.  
The simulation results show satisfactory size, substantial power gains,
accurate single-change localization, and reliable multiple-change recovery
under both independent and functionally autoregressive errors.  The two
temperature applications identify interpretable shifts at different signal
strengths and show why heterogeneous score scales and residual temporal
dependence should be incorporated in functional change-point analysis.

Several extensions merit further investigation.  The multiple-change theory
in this paper treats a fixed number of proportionally separated changes; a
natural next step is to allow a diverging number of changes and shrinking
minimal spacing, which would require a multiscale strong approximation with
explicit tail control.  Inference for sparsely, irregularly, or partially
observed curves would require joint treatment of score reconstruction,
measurement error, and long-run covariance estimation
\citep{BaiHuWu2026}.  The same spectral regularization
principle could also be adapted to changes in covariance operators and
eigensystems
\citep{StoehrEtAl2021,DetteKutta2021,JiaoFrostigOmbao2023}.  Finally,
robustification against heavy tails and outliers, together with growing or
data-adaptive ridge grids, would connect the present framework to robust
functional change-point methods and ridge aggregation
\citep{WegnerWendler2024,ZhaoZhouFeng2026}.

\appendix

\section{Auxiliary results for the main theorems}
\label{app:auxiliary-results}

This section records the intermediate conclusions used in
Sections~\ref{sec:theory} and~\ref{sec:multiple}.  The statements are separated by purpose; their
proofs are given in Appendices~\ref{app:functional-proofs}--
\ref{app:multiple-proofs}.

\subsection{Weak dependence and numerical projection}

For centered random elements $x_1,x_2,x_3,x_4\in\cH$ with fourth moments,
define
\begin{align}
 \kappa_4(x_1,x_2,x_3,x_4)
 &=\E\{\ip{x_1}{x_3}\ip{x_2}{x_4}\}
 -\tr(\mathsf C_{13})\tr(\mathsf C_{24})
 \notag\\
 &\quad-\ip{\mathsf C_{12}}{\mathsf C_{34}}_{\HS}
 -\tr(\mathsf C_{14}\mathsf C_{23}),
 \label{eq:contracted-cumulant-definition}
\end{align}
where $\mathsf C_{ab}=\E(x_a\otimes x_b)$, and put
$\kappa_e(u_1,u_2,u_3)=\kappa_4(e_0,e_{u_1},e_{u_2},e_{u_3})$.

\begin{lemma}[Long-run covariance and fourth-order summability]
\label{prop:lrc-origin}
Under Assumption~\ref{ass:ARS-dependence},
\begin{equation}
 \sum_{u\in\mathbb Z}(1+|u|)^{\beta_{\rm dep}}\norm{\mathcal C_u}_{1}<\infty.
 \label{eq:weighted-cov-summability}
\end{equation}
Moreover,
\begin{equation}
 \sum_{u_1,u_2,u_3\in\mathbb Z}|\kappa_e(u_1,u_2,u_3)|<\infty.
 \label{eq:contracted-cumulant-summability}
\end{equation}
Consequently,
\begin{equation*}
 \Omega=\sum_{u\in\mathbb Z}\mathcal C_u
\end{equation*}
exists in trace norm and is a nonnegative trace-class operator.
\end{lemma}

The fourth-order summability in Lemma~\ref{prop:lrc-origin} is derived from
the same independent-copy approximation as the covariance summability; it is
not an additional cumulant assumption.

\begin{lemma}[Functional invariance principle]
\label{lem:ARS-main-FCLT}
Under Assumption~\ref{ass:ARS-dependence},
\begin{equation}
 \frac1{\sqrt n}\sum_{i=1}^{\lfloor nt\rfloor}e_i
 \dto W_\Omega(t)
 \label{eq:ARS-partial-sum-FCLT}
\end{equation}
in $D([0,1],\cH)$, where
$\E\{W_\Omega(t)\otimes W_\Omega(u)\}=(t\wedge u)\Omega$.  Hence
\begin{equation}
 \frac1{\sqrt n}\left\{
 \sum_{i=1}^{\lfloor nt\rfloor}e_i
 -\frac{\lfloor nt\rfloor}{n}\sum_{i=1}^ne_i
 \right\}
 \dto\mathbb B_\Omega(t)
 \label{eq:ARS-bridge-FCLT}
\end{equation}
in $D([0,1],\cH)$.
\end{lemma}

\begin{lemma}[Numerical Fourier approximation]
\label{prop:quadrature}
Under Assumption~\ref{ass:grid}, for every $f\in C^2[0,1]$,
\begin{equation}
 \norm{(\mathcal J_{p,N}-\Pi_p)f}
 \le C\chi_{p,N}\norm{f}_{C^2}.
 \label{eq:numerical-projection-operator}
\end{equation}
Moreover, for every fixed $B<\infty$,
\begin{equation}
 \mathfrak a_{p,N}(B)
 =\sup_{\substack{f\in C^2[0,1]\\ \norm f_{C^2}\le B}}
 \norm{\mathcal J_{p,N}f-f}\longrightarrow0.
 \label{eq:uniform-functional-resolution-limit}
\end{equation}
If Assumption~\ref{ass:ARS-dependence} also holds, then
\begin{equation}
 \norm{\Omega_{p,N}^{\cH}-\Omega}_{1}
 \le \norm{\Omega-\Pi_p\Omega\Pi_p}_{1}+C\chi_{p,N}.
 \label{eq:lrc-projection-error}
\end{equation}
\end{lemma}

\begin{lemma}[Stability of dependence under projection]
\label{lem:projected-dependence}
Under Assumptions~\ref{ass:ARS-dependence} and \ref{ass:grid}, let
$\mathcal J_n=\mathcal J_{p,N}$ and define
\begin{equation*}
 \varrho_p=
\norm{(\Id_{\cH}-\Pi_p)e_0}_{L^2(\cH)}
 +\sum_{M=1}^{\infty}
 \norm{(\Id_{\cH}-\Pi_p)(e_0-e_{0,M})}_{L^2(\cH)}.
\end{equation*}
Then $\varrho_p\to0$ and
\begin{equation}
 \sup_{0\le t\le1}
 \left\|\frac1{\sqrt n}\sum_{i=1}^{\lfloor nt\rfloor}(\mathcal J_ne_i-e_i)\right\|
 =O_{\Pp}(\varrho_p+\chi_{p,N})=o_{\Pp}(1).
 \label{eq:numerical-partial-sum-equivalence}
\end{equation}
Moreover,
\begin{equation}
 \sup_n\sum_{u\in\mathbb Z}(1+|u|)^{\beta_{\rm dep}}
 \norm{\bGamma_{p,N}(u)}_1<\infty,
 \label{eq:projected-weighted-cov-sum}
\end{equation}
$\bOmega_{p,N}$ is positive semidefinite, and, with
\begin{equation*}
 \kappa_{p,N}(u_1,u_2,u_3)=
 \kappa_4(\beps_{0,p,N},\beps_{u_1,p,N},
          \beps_{u_2,p,N},\beps_{u_3,p,N}),
\end{equation*}
\begin{equation}
 \sup_n\sum_{u_1,u_2,u_3\in\mathbb Z}|\kappa_{p,N}(u_1,u_2,u_3)|<\infty.
 \label{eq:projected-cumulant-summability}
\end{equation}
\end{lemma}

\subsection{Edge-corrected difference-based long-run covariance}

For $i_{\rm D}\in\mathbb Z$, define the stationary noise-only DB sequence
\begin{equation*}
 \beps_{i_{\rm D}}^{\rm DB}
 =\sum_{q=0}^{m}c_q\beps_{mh+i_{\rm D}-qh,p,N},
 \qquad
 \bGamma_{\rm DB}(r)
 =\E\{\beps_{i_{\rm D}+r}^{\rm DB}
       (\beps_{i_{\rm D}}^{\rm DB})^{\top}\}.
\end{equation*}
By stationarity, the covariance on the right does not depend on $i_{\rm D}$.  Define
\begin{equation*}
 \bOmega_{{\rm DB},\ell}
 =\sum_{|r|<\ell}K(r/\ell)\bGamma_{\rm DB}(r).
\end{equation*}
For comparison with the common-divisor implementation, define
\begin{equation*}
 \widehat\bGamma_{\rm DB}^{\rm com}(r)
 =\frac1n\sum_{i_{\rm D}=r+1}^{N_D}
 \bY_{i_{\rm D}}\bY_{i_{\rm D}-r}^{\top},
 \qquad
 \widehat\bGamma_{\rm DB}^{\rm com}(-r)
 =\widehat\bGamma_{\rm DB}^{\rm com}(r)^{\top},
\end{equation*}
and
\begin{equation*}
 \widehat\bOmega_n^{\rm com}
 =\sum_{|r|<\ell}K(r/\ell)
  \widehat\bGamma_{\rm DB}^{\rm com}(r).
\end{equation*}

\begin{lemma}[Endpoint correction]
\label{prop:edge-correction}
Under $H_0$ and Assumptions~\ref{ass:ARS-dependence}--
\ref{ass:DB-tuning}, for $0\le r<\ell$,
\begin{equation}
 \E\widehat\bGamma_{\rm DB}(r)=\bGamma_{\rm DB}(r).
 \label{eq:edge-unbiased}
\end{equation}
If every lag product is divided by $n$, then
\begin{equation}
 \E\widehat\bGamma_{\rm DB}^{\rm com}(r)
 =\frac{N_D-r}{n}\bGamma_{\rm DB}(r),
 \label{eq:paper-denominator-bias}
\end{equation}
and
\begin{equation}
 \norm{\E\widehat\bOmega_n^{\rm com}-\bOmega_{{\rm DB},\ell}}_1
 \le C\frac{mh+\ell}{n}.
 \label{eq:endpoint-bias-bound}
\end{equation}
\end{lemma}

Thus the divisor $N_D-r$ removes the finite-sample endpoint factor exactly;
it does not alter the stochastic order of the estimator.

\begin{lemma}[Population DB bias]
\label{prop:population-DB-bias}
Under Assumptions~\ref{ass:ARS-dependence}--\ref{ass:DB-tuning},
\begin{equation}
 \norm{\bOmega_{{\rm DB},\ell}-\bOmega_{p,N}}_1
 \le C\ell^{-\beta_K}.
 \label{eq:population-DB-rate}
\end{equation}
\end{lemma}

\begin{lemma}[Stochastic DB rate]
\label{thm:DB-rate}
Under $H_0$ and Assumptions~\ref{ass:ARS-dependence}--
\ref{ass:DB-tuning},
\begin{align}
 \E\norm{\widehat\bOmega_n^{\rm DB}-
 \E\widehat\bOmega_n^{\rm DB}}_{\HS}^2
 &\le C\frac{\ell}{N_D},
 \label{eq:DB-variance-rate}\\
 \norm{\widehat\bOmega_n^{\rm DB}-\bOmega_{p,N}}_{\HS}
 &=O_{\Pp}(\mathfrak r_n).
 \label{eq:DB-HS-rate}
\end{align}
\end{lemma}

\begin{lemma}[Positive-part projection]
\label{lem:positive-part-rate}
Under the conditions of Lemma~\ref{thm:DB-rate},
\begin{align}
 \norm{\widehat\bOmega_n^+-\bOmega_{p,N}}_{\HS}
 &=O_{\Pp}(\mathfrak r_n),
 \label{eq:positive-HS-rate}\\
 \norm{\widehat\bOmega_n^+-\bOmega_{p,N}}_{1}
 &=O_{\Pp}(\sqrt p\,\mathfrak r_n).
 \label{eq:positive-trace-rate}
\end{align}
\end{lemma}

\begin{lemma}[Contamination by one or several mean changes]
\label{prop:alternative-DB}
Under the single-change model~\eqref{eq:functional-change-model} and
Assumptions~\ref{ass:ARS-dependence}--\ref{ass:DB-tuning},
\begin{align}
 \norm{\widehat\bOmega_n^{\rm DB}-\bOmega_{p,N}}_{\HS}
 &=O_{\Pp}(\mathfrak r_n+\mathfrak r_{\delta,n}),
 \notag\\
 \norm{\widehat\bOmega_n^+-\bOmega_{p,N}}_1
 &=O_{\Pp}\left[\sqrt p\{\mathfrak r_n+\mathfrak r_{\delta,n}\}\right].
 \label{eq:alternative-positive-trace-rate}
\end{align}
Under the multiple-change model~\eqref{eq:multiple-change-model} and
Assumptions~\ref{ass:ARS-dependence}--\ref{ass:DB-tuning} and
\ref{ass:multiple-configuration}, the corresponding conclusions are
\begin{align}
 \norm{\widehat\bOmega_n^{\rm DB}-\bOmega_{p,N}}_{\HS}
 &=O_{\Pp}(\mathfrak r_n+\mathfrak r_{\delta,{\rm cp},n}),
 \label{eq:multiple-DB-HS-rate}\\
 \norm{\widehat\bOmega_n^+-\bOmega_{p,N}}_1
 &=O_{\Pp}\left[
 \sqrt p\{\mathfrak r_n+\mathfrak r_{\delta,{\rm cp},n}\}\right].
 \label{eq:multiple-positive-trace-rate}
\end{align}
For a fixed difference order, $h\asymp\ell$, and bounded jumps, the
contamination terms vanish whenever $\ell^2/n\to0$.
\end{lemma}

\subsection{Spectral weights and reference laws}

\begin{lemma}[Trace and ridge-weight consistency]
\label{prop:weight-consistency}
Under $H_0$ and Assumptions~\ref{ass:ARS-dependence}--
\ref{ass:spectrum-growth}, uniformly over $a\in\cG$,
\begin{equation}
 |\widehat{\mathfrak t}_n-\mathfrak t_\Omega|
 +\sum_{j\ge1}|\widehat w_{j,n}(a)-w_j(a)|
 =O_{\Pp}\{\sqrt p\,\mathfrak r_n+\zeta_{p,N}\},
 \label{eq:weight-l1-rate}
\end{equation}
where $\widehat w_{j,n}(a)=0$ for $j>p$. Consequently,
\begin{equation}
 \max_{a\in\cG}
 \left\{|\widehat A_{1,n}(a)-A_1(a)|
       +|\widehat A_{2,n}(a)-A_2(a)|\right\}
 =O_{\Pp}\{\sqrt p\,\mathfrak r_n+\zeta_{p,N}\}.
 \label{eq:Ar-consistency}
\end{equation}
Under one change, replace $\mathfrak r_n$ by
$\mathfrak r_n+\mathfrak r_{\delta,n}$.  Under
model~\eqref{eq:multiple-change-model}, replace it by
$\mathfrak r_n+\mathfrak r_{\delta,{\rm cp},n}$, so the right-hand sides of
\eqref{eq:weight-l1-rate}--\eqref{eq:Ar-consistency} are
$O_{\Pp}(\mathfrak e_{{\rm cp},n})$.
\end{lemma}

\begin{lemma}[Weighted-bridge regularity]
\label{lem:bridge-series-regularity}
For every $a\in\cG$, the series in
\eqref{eq:weighted-bridge-process} converges uniformly almost surely and in
$L^1$ on $[\epsilon,1-\epsilon]$.  Its limit has continuous sample paths.
\end{lemma}

\begin{lemma}[Finite conditional reference-law regularity]
\label{lem:finite-reference-regularity}
On $\mathcal E_n$, $\widehat F_{n,a}$ is continuous and strictly increasing
on its support for every $a\in\cG$.  For a fresh reference draw,
\begin{equation*}
 1-\widehat F_{n,a}(\widehat{\mathcal T}_{n,a}^{\circ})
 \mid\widehat{\mathcal W}_n\sim\operatorname{Uniform}(0,1).
\end{equation*}
The conditional distribution of $\mathfrak C_n^{\circ}$ is continuous.
\end{lemma}

\begin{lemma}[Monte Carlo approximation of the bridge reference]
\label{lem:Monte-Carlo-reference}
Conditional on the data and on $\mathcal E_n$, for every $n$, every
$N_{\rm sim}\ge1$, and every $y>0$, the empirical conditional marginal cdfs based on
$N_{\rm sim}$ independent common-bridge paths satisfy
\begin{equation}
 \Pp_{\bm B}\left\{
 \max_{a\in\cG}\sup_x
 |\widehat F_{n,a,N_{\rm sim}}(x)-\widehat F_{n,a}(x)|>y
 \mid\widehat{\mathcal W}_n\right\}
 \le 2L_{\cG}\exp(-2N_{\rm sim}y^2).
 \label{eq:DKW-marginal-bridge}
\end{equation}
Thus the conditional uniform error is $O_{\Pp_{\bm B}}(N_{\rm sim}^{-1/2})$,
uniformly in $n$.  Conditional additionally on the first bridge batch, let
\begin{equation*}
 \widehat G_{n,N_{{\rm sim},1}}^{\circ}(x)
 =\Pp_{\bm B}\{\mathfrak C_{n,N_{{\rm sim},1}}^{\circ}\le x
 \mid\widehat{\mathcal W}_n,\text{ first batch}\}
\end{equation*}
and let $\widehat G_{n,N_{{\rm sim},1},N_{{\rm sim},2}}^{\circ}$ be the empirical cdf of the
$N_{{\rm sim},2}$ second-batch copies
$\mathfrak C_{n,r_{\rm sim}}^{\circ,(N_{{\rm sim},1})}$.  Then
\begin{equation}
 \Pp_{\bm B}\left\{
 \sup_x|\widehat G_{n,N_{{\rm sim},1},N_{{\rm sim},2}}^{\circ}(x)
          -\widehat G_{n,N_{{\rm sim},1}}^{\circ}(x)|>y
 \mid\widehat{\mathcal W}_n,\text{ first batch}\right\}
 \le 2\exp(-2N_{{\rm sim},2}y^2).
 \label{eq:DKW-joint-bridge}
\end{equation}

Under the assumptions of Theorem~\ref{thm:critical-values}, if jointly
$n\to\infty$, $N_{{\rm sim},1}=N_{{\rm sim},1}(n)\to\infty$, and
$N_{{\rm sim},2}=N_{{\rm sim},2}(n)\to\infty$, and if
$\mathsf{CR}_{\alpha}$ holds, then the two-batch construction in
\eqref{eq:midrank-bridge-pvalue}--\eqref{eq:two-batch-critical-value} obeys
\begin{equation*}
 \widetilde c_{n,C,1-\alpha}^{(N_{{\rm sim},1},N_{{\rm sim},2})}\pto c_{C,1-\alpha},
 \qquad
 \Pp\{\mathfrak C_{n,N_{{\rm sim},1}}>
 \widetilde c_{n,C,1-\alpha}^{(N_{{\rm sim},1},N_{{\rm sim},2})}\}\longrightarrow\alpha.
\end{equation*}
No universal $N_{{\rm sim},1}^{-1/2}$ rate is asserted after the unbounded Cauchy
transform.
\end{lemma}

\subsection{Auxiliary results for multiple changes}

For $I=(s,e]$ and $k\in\cK_I$, put
$n_{I,L}(k)=k-s$ and $n_{I,R}(k)=e-k$, and define
\begin{equation*}
 \mathsf Z_I(k)
 =\sqrt{\frac{n_{I,L}(k)n_{I,R}(k)}{n_I}}
 \left\{
 \frac1{n_{I,R}(k)}\sum_{i=k+1}^{e}\beps_{i,p,N}
 -\frac1{n_{I,L}(k)}\sum_{i=s+1}^{k}\beps_{i,p,N}
 \right\}.
\end{equation*}
Let
\begin{equation*}
 \mathfrak Z_n^{\rm WBS}
 =\max_{0\le s<k<e\le n}\norm{\mathsf Z_{(s,e]}(k)}.
\end{equation*}

For $1\le b\le J_{\rm cp}$, write
$\Delta_{b,n}^{-}=\tau_{b,n}-\tau_{b-1,n}$ and
$\Delta_{b,n}^{+}=\tau_{b+1,n}-\tau_{b,n}$, and define
\begin{align*}
 \mathcal L_{b,n}^{\rm WBS}
 &=\left\{\tau_{b-1,n}+\left\lceil\frac{\Delta_{b,n}^{-}}3\right\rceil,
 \ldots,
 \tau_{b-1,n}+\left\lfloor\frac{2\Delta_{b,n}^{-}}3\right\rfloor\right\},\\
 \mathcal R_{b,n}^{\rm WBS}
 &=\left\{\tau_{b,n}+\left\lceil\frac{\Delta_{b,n}^{+}}3\right\rceil,
 \ldots,
 \tau_{b,n}+\left\lfloor\frac{2\Delta_{b,n}^{+}}3\right\rfloor\right\}.
\end{align*}
Let $\mathcal C_n^{\rm WBS}$ be the event that, for every $b$, at least one
sampled interval has its left endpoint in $\mathcal L_{b,n}^{\rm WBS}$ and
its right endpoint in $\mathcal R_{b,n}^{\rm WBS}$.

\begin{lemma}[Random-interval coverage]
\label{lem:wbs-coverage}
Under Assumption~\ref{ass:multiple-configuration}, there is a universal
constant $c_{\rm W}>0$ such that, for all sufficiently large $n$,
\begin{equation}
 \Pp\{(\mathcal C_n^{\rm WBS})^c\}
 \le J_{\rm cp}\left\{1-c_{\rm W}
 \left(\frac{\Delta_n}{n}\right)^2\right\}^{M_n^{\rm WBS}}.
 \label{eq:wbs-coverage-probability}
\end{equation}
On $\mathcal C_n^{\rm WBS}$, each $\tau_{b,n}$ has an interval
$I_{b,n}^{\rm iso}=(s_{b,n}^{\rm iso},e_{b,n}^{\rm iso}]$ in
$\mathcal F_n^{\rm WBS}$ whose interior contains no other change and
\begin{equation}
 \begin{gathered}
 s_{b,n}^{\rm iso}-\tau_{b-1,n}\ge\frac{\Delta_n}{3}-1,
 \qquad
 \tau_{b,n}-s_{b,n}^{\rm iso}\ge\frac{\Delta_n}{3}-1,\\
 e_{b,n}^{\rm iso}-\tau_{b,n}\ge\frac{\Delta_n}{3}-1,
 \qquad
 \tau_{b+1,n}-e_{b,n}^{\rm iso}\ge\frac{\Delta_n}{3}-1.
 \end{gathered}
 \label{eq:wbs-isolation-margins}
\end{equation}
\end{lemma}

\begin{lemma}[Uniform local-score envelope]
\label{lem:local-score-envelope}
Under Assumptions~\ref{ass:ARS-dependence} and~\ref{ass:grid}, there is a
constant $C<\infty$ such that, for all $x>0$ and all sufficiently large $n$,
\begin{equation}
 \Pp\{\mathfrak Z_n^{\rm WBS}>x\}
 \le Cn(1+\log n)x^{-\nu}.
 \label{eq:local-score-tail}
\end{equation}
Consequently,
\begin{equation}
 \mathfrak Z_n^{\rm WBS}
 =O_{\Pp}\{(n\log n)^{1/\nu}\}
 =O_{\Pp}(\mathfrak v_n^{\rm WBS}).
 \label{eq:local-score-order}
\end{equation}
\end{lemma}

For a deterministic retained jump
$d\in\operatorname{span}(\phi_1,\ldots,\phi_p)$, define
\begin{equation*}
 \widehat{\mathfrak q}_{n,a}(d)
 =\ip{d}{\widehat{\mathcal R}_{n,a}d},
 \qquad
 \widehat{\mathfrak q}_{n}^{\max}(d)
 =\max_{a\in\cG}
 \frac{\widehat{\mathfrak q}_{n,a}(d)}
 {\{2\widehat A_{2,n}(a)\}^{1/2}}.
\end{equation*}
Let
\begin{equation*}
 \mu_{i,n}^{\rm ret}
 =\mathcal J_{p,N}\left\{\mu_0+
 \sum_{b=1}^{J_{\rm cp}}\delta_{b,n}\ind(i>\tau_{b,n})\right\}
\end{equation*}
and, for $I=(s,e]$ and $k\in\cK_I$, put
\begin{equation*}
 \mathsf M_I(k)
 =\sqrt{N_I(k)}\left\{
 \frac1{e-k}\sum_{i=k+1}^{e}\mu_{i,n}^{\rm ret}
 -\frac1{k-s}\sum_{i=s+1}^{k}\mu_{i,n}^{\rm ret}\right\}.
\end{equation*}
Define
\begin{equation*}
 \cS_I^{\rm sig}(k)
 =\max_{a\in\cG}
 \frac{\ip{\mathsf M_I(k)}
 {\widehat{\mathcal R}_{n,a}\mathsf M_I(k)}}
 {\{2\widehat A_{2,n}(a)\}^{1/2}},
 \qquad
 \cS_I^{\rm sig,max}=\max_{k\in\cK_I}\cS_I^{\rm sig}(k).
\end{equation*}
Write
\begin{equation*}
 \mathcal T_I^{\rm cp}=\{\tau_{b,n}:s<\tau_{b,n}<e\},
 \qquad
 d_I^{\max}=\max_{\tau_{b,n}\in\mathcal T_I^{\rm cp}}\norm{d_{b,n}},
\end{equation*}
with $d_I^{\max}=0$ when $\mathcal T_I^{\rm cp}$ is empty.

\begin{lemma}[Uniform local ridge approximation]
\label{lem:local-ridge-decomposition}
Suppose model~\eqref{eq:multiple-change-model} and
Assumptions~\ref{ass:ARS-dependence}--\ref{ass:spectrum-growth} and
\ref{ass:multiple-configuration} hold, and
$\mathfrak e_{{\rm cp},n}\to0$.  There is an event
$\mathcal E_n^{\rm loc}\subseteq\mathcal E_n$ with
$\Pp(\mathcal E_n^{\rm loc})\to1$ on which the following bounds hold
simultaneously for every interval $I=(s,e]$, every $k\in\cK_I$, and every
$v\in\operatorname{span}(\phi_1,\ldots,\phi_p)$:
\begin{equation}
 c_{\rm R}\norm v^2
 \le
 \max_{a\in\cG}
 \frac{\ip v{\widehat{\mathcal R}_{n,a}v}}
 {\{2\widehat A_{2,n}(a)\}^{1/2}}
 \le C_{\rm R}\norm v^2,
 \label{eq:local-signal-norm-equivalence}
\end{equation}
and
\begin{align}
 |\cS_I(k)-\cS_I^{\rm sig}(k)|
 &\le C\left\{1+\norm{\mathsf Z_I(k)}^2
 +\norm{\mathsf M_I(k)}\norm{\mathsf Z_I(k)}\right\}
 \label{eq:wbs-uniform-score-exact}\\
 &\le C\left\{1+\norm{\mathsf Z_I(k)}^2
 +\sqrt{n_I}\,d_I^{\max}\norm{\mathsf Z_I(k)}\right\}.
 \label{eq:wbs-uniform-score-approximation}
\end{align}
If $I$ contains exactly one change $\tau_I$ with retained jump $d_I$, put
$x_{I,k}=(k-s)/n_I$ and $\theta_I=(\tau_I-s)/n_I$.  Then, uniformly over
$a\in\cG$ and $k\in\cK_I$,
\begin{equation}
 D_{I,a}(k)
 =n_I\mathfrak h_{\rm cp}(x_{I,k},\theta_I)
 \frac{\widehat{\mathfrak q}_{n,a}(d_I)}
 {\{2\widehat A_{2,n}(a)\}^{1/2}}
 +\mathrm{Rem}_{I,a}(k),
 \label{eq:local-ridge-decomposition}
\end{equation}
where
\begin{equation}
 |\mathrm{Rem}_{I,a}(k)|
 \le C\left\{1+\norm{\mathsf Z_I(k)}^2
 +\sqrt{n_I}\norm{d_I}\norm{\mathsf Z_I(k)}\right\}.
 \label{eq:local-ridge-remainder}
\end{equation}
Moreover,
\begin{equation}
 c_{\rm sig}\norm{d_I}^2
 \le\widehat{\mathfrak q}_{n}^{\max}(d_I)
 \le C_{\rm sig}\norm{d_I}^2,
 \label{eq:local-ridge-signal-equivalence}
\end{equation}
and, for every $\varepsilon_{\rm gap}\in(0,1/2)$,
\begin{equation}
 \mathfrak h_{\rm cp}(\theta,\theta)
 -\mathfrak h_{\rm cp}(x,\theta)
 \ge\varepsilon_{\rm gap}|x-\theta|,
 \quad
 \theta\in[\varepsilon_{\rm gap},1-\varepsilon_{\rm gap}],\quad x\in(0,1).
 \label{eq:local-shape-linear-gap}
\end{equation}
\end{lemma}

\begin{lemma}[Deterministic geometry of the WBS signal]
\label{lem:wbs-deterministic-geometry}
On $\mathcal E_n^{\rm loc}$, uniformly over all intervals $I=(s,e]$, the map
$k\mapsto\cS_I^{\rm sig}(k)$ is convex between successive points of
$\mathcal T_I^{\rm cp}\cup\{s,e\}$.  If
$\cS_I^{\rm sig,max}>0$, a maximizer belongs to
$\mathcal T_I^{\rm cp}$, and
\begin{equation}
 \cS_I^{\rm sig,max}\le Cn_I(d_I^{\max})^2.
 \label{eq:wbs-deterministic-upper}
\end{equation}
For every fixed $c_0>0$, there are constants $c_1,c_2>0$ such that
\begin{equation*}
 \cS_I^{\rm sig,max}\ge c_0\Delta_n(d_I^{\max})^2
\end{equation*}
implies
\begin{equation}
 \cS_I^{\rm sig,max}-\cS_I^{\rm sig}(k)
 \ge c_1(d_I^{\max})^2
 \left\{\operatorname{dist}(k,\mathcal T_I^{\rm cp})
 \wedge c_2\Delta_n\right\},
 \qquad k\in\cK_I.
 \label{eq:wbs-deterministic-gap}
\end{equation}
Finally, if $k-s\le r_{\rm bd}$ or $e-k\le r_{\rm bd}$, then
\begin{equation}
 \cS_I^{\rm sig}(k)+\norm{\mathsf M_I(k)}^2\le Cr_{\rm bd}.
 \label{eq:wbs-boundary-score}
\end{equation}
\end{lemma}

\section{Preliminary operator and probability results}

For a compact operator $A$, write $\norm{A}_1$, $\norm{A}_{\HS}$ and
$\norm{A}_{\op}$ for its trace, Hilbert--Schmidt and operator norms.  For
$x,y\in\cH$, $(x\otimes y)z=\ip{y}{z}x$.  We repeatedly use
\begin{equation*}
 \norm{x\otimes y}_{\HS}=\norm{x}\norm{y},
 \qquad
 \norm{ABC}_1\le \norm{A}_{\op}\norm{B}_1\norm{C}_{\op},
 \qquad
 |\tr(AB)|\le\norm{A}_1\norm{B}_{\op}.
\end{equation*}
If $A$ has rank at most $p$, then
\begin{equation*}
 \norm{A}_1\le \sqrt p\,\norm{A}_{\HS}.
\end{equation*}
For a self-adjoint finite-rank operator $A$, let $A_+$ denote the operator
obtained by replacing every negative eigenvalue by zero.  The map
$A\mapsto A_+$ is the metric projection, in Hilbert--Schmidt norm, onto the
closed convex cone of positive semidefinite operators.  Hence, for every
$B\succeq0$,
\begin{equation}
 \norm{A_+-B}_{\HS}\le \norm{A-B}_{\HS}.
 \label{eq:PSD-projection-audited}
\end{equation}
The finite-dimensional statement is proved in \citet{Higham1988}; the same
identity applies to a finite-rank operator after restriction to its range.

\begin{lemma}[Trace-norm approximation by strongly convergent projections]
\label{lem:trace-projection}
Let $A$ be a trace-class operator on a separable Hilbert space and let
$\mathsf P_p$ be orthogonal projections satisfying $\mathsf P_pf\to f$ for every $f$.
Then
\begin{equation}
 \norm{A-\mathsf P_pA\mathsf P_p}_1\longrightarrow0.
 \label{eq:trace-projection-convergence}
\end{equation}
\end{lemma}

\begin{proof}
For $\varepsilon_{\rm app}>0$, choose a finite-rank operator
\[
 A_{\rm fin}=\sum_{j_{\rm fin}=1}^{J_{\rm fin}}
 x_{j_{\rm fin}}\otimes y_{j_{\rm fin}},
 \qquad
 \norm{A-A_{\rm fin}}_1<\varepsilon_{\rm app},
\]
using the density of finite-rank operators in the trace class
\citep[Chapter~2]{Simon2005}.  Since $\norm{\mathsf P_p}_{\op}\le1$,
\begin{align*}
 \norm{A-\mathsf P_pA\mathsf P_p}_1
 &\le \norm{A-A_{\rm fin}}_1
      +\norm{A_{\rm fin}-\mathsf P_pA_{\rm fin}\mathsf P_p}_1
      +\norm{\mathsf P_p(A_{\rm fin}-A)\mathsf P_p}_1\\
 &\le2\varepsilon_{\rm app}
      +\norm{A_{\rm fin}-\mathsf P_pA_{\rm fin}\mathsf P_p}_1.
\end{align*}
For $1\le j_{\rm fin}\le J_{\rm fin}$,
\begin{align*}
 &x_{j_{\rm fin}}\otimes y_{j_{\rm fin}}
 -\mathsf P_px_{j_{\rm fin}}\otimes \mathsf P_py_{j_{\rm fin}}\\
 &\quad=(\Id-\mathsf P_p)x_{j_{\rm fin}}\otimes y_{j_{\rm fin}}
 +\mathsf P_px_{j_{\rm fin}}\otimes(\Id-\mathsf P_p)y_{j_{\rm fin}},\\
 &\norm{x_{j_{\rm fin}}\otimes y_{j_{\rm fin}}
 -\mathsf P_px_{j_{\rm fin}}\otimes \mathsf P_py_{j_{\rm fin}}}_1\\
 &\quad\le
 \norm{(\Id-\mathsf P_p)x_{j_{\rm fin}}}\norm{y_{j_{\rm fin}}}
 +\norm{\mathsf P_px_{j_{\rm fin}}}\norm{(\Id-\mathsf P_p)y_{j_{\rm fin}}}.
\end{align*}
The strong convergence $\mathsf P_pf\to f$ therefore gives
\[
 \norm{A_{\rm fin}-\mathsf P_pA_{\rm fin}\mathsf P_p}_1\longrightarrow0,
 \qquad
 \limsup_{p\to\infty}\norm{A-\mathsf P_pA\mathsf P_p}_1
 \le2\varepsilon_{\rm app}.
\]
Letting $\varepsilon_{\rm app}\downarrow0$ proves
\eqref{eq:trace-projection-convergence}.
\end{proof}

\begin{lemma}[Centered rank-one covariance decomposition]
\label{lem:rank-one-cumulant}
Let $x_1,x_2,x_3,x_4$ be centered random elements of a real separable
Hilbert space with finite fourth moments.  Write
$\mathsf C_{ab}=\E(x_a\otimes x_b)$ for $1\le a,b\le4$, and let
$\kappa_4(x_1,x_2,x_3,x_4)$ be defined by
\eqref{eq:contracted-cumulant-definition}.  Then
\begin{align}
 &\E\ip{x_1\otimes x_2-\mathsf C_{12}}
 {x_3\otimes x_4-\mathsf C_{34}}_{\HS}
 \notag\\
 &\qquad=\tr(\mathsf C_{13})\tr(\mathsf C_{24})
       +\tr(\mathsf C_{14}\mathsf C_{23})+\kappa_4(x_1,x_2,x_3,x_4).
 \label{eq:rank-one-cumulant-identity}
\end{align}
Consequently,
\begin{align}
 &\left|\E\ip{x_1\otimes x_2-\mathsf C_{12}}
 {x_3\otimes x_4-\mathsf C_{34}}_{\HS}\right|
 \notag\\
 &\qquad\le \norm{\mathsf C_{13}}_1\norm{\mathsf C_{24}}_1
       +\norm{\mathsf C_{14}}_1\norm{\mathsf C_{23}}_1
       +|\kappa_4(x_1,x_2,x_3,x_4)|.
 \label{eq:rank-one-cumulant-bound}
\end{align}
\end{lemma}

\begin{proof}
The rank-one identity
\[
 \ip{x_1\otimes x_2}{x_3\otimes x_4}_{\HS}
 =\ip{x_1}{x_3}\ip{x_2}{x_4}
\]
implies
\begin{align*}
 &\E\ip{x_1\otimes x_2-\mathsf C_{12}}
 {x_3\otimes x_4-\mathsf C_{34}}_{\HS}\\
 &\quad=
 \E\{\ip{x_1}{x_3}\ip{x_2}{x_4}\}
 -\ip{\mathsf C_{12}}{\mathsf C_{34}}_{\HS}.
\end{align*}
By \eqref{eq:contracted-cumulant-definition},
\begin{align*}
 \E\{\ip{x_1}{x_3}\ip{x_2}{x_4}\}
 &=\tr(\mathsf C_{13})\tr(\mathsf C_{24})
   +\ip{\mathsf C_{12}}{\mathsf C_{34}}_{\HS}\\
 &\quad+\tr(\mathsf C_{14}\mathsf C_{23})
   +\kappa_4(x_1,x_2,x_3,x_4).
\end{align*}
Cancellation of $\ip{\mathsf C_{12}}{\mathsf C_{34}}_{\HS}$ gives
\eqref{eq:rank-one-cumulant-identity}.  Moreover,
\[
 |\tr(AB)|
 \le\norm A_1\norm B_{\op}
 \le\norm A_1\norm B_1,
\]
which yields \eqref{eq:rank-one-cumulant-bound}.
\end{proof}

\begin{lemma}[Replacement bound for the contracted cumulant]
\label{lem:cumulant-replacement}
Let $x_j,x_j'$, $1\le j\le4$, be centered Hilbert-space random elements
with finite fourth moments and
\begin{equation*}
 \max_{j\le4}\{\norm{x_j}_{L^4(\cH)},
                    \norm{x_j'}_{L^4(\cH)}\}\le C_4^{\rm mom}.
\end{equation*}
Then
\begin{equation}
 \left|\kappa_4(x_1,x_2,x_3,x_4)
       -\kappa_4(x_1',x_2',x_3',x_4')\right|
 \le C (C_4^{\rm mom})^3\sum_{j=1}^4\norm{x_j-x_j'}_{L^4(\cH)}.
 \label{eq:cumulant-replacement-bound}
\end{equation}
If the four variables can be divided into two nonempty independent
subcollections, then their contracted fourth cumulant is zero.
\end{lemma}

\begin{proof}
For
\[
 \bm x^{(j_{\rm rep})}
 =(x_1',\ldots,x_{j_{\rm rep}}',x_{j_{\rm rep}+1},\ldots,x_4),
 \qquad 0\le j_{\rm rep}\le4,
\]
the telescoping identity is
\[
 \kappa_4(\bm x^{(0)})-\kappa_4(\bm x^{(4)})
 =\sum_{j_{\rm rep}=1}^{4}
 \{\kappa_4(\bm x^{(j_{\rm rep}-1)})
 -\kappa_4(\bm x^{(j_{\rm rep})})\}.
\]
It therefore suffices to replace one argument.  For the raw fourth moment,
\begin{align*}
 &\left|
 \E\{\ip{x_1}{x_3}\ip{x_2}{x_4}\}
 -\E\{\ip{x_1'}{x_3}\ip{x_2}{x_4}\}
 \right|\\
 &\quad\le
 \norm{x_1-x_1'}_{L^4}
 \norm{x_2}_{L^4}\norm{x_3}_{L^4}\norm{x_4}_{L^4}
 \le (C_4^{\rm mom})^3\norm{x_1-x_1'}_{L^4}.
\end{align*}
For covariance operators,
\begin{align*}
 \norm{\mathsf C_{xy}}_{\HS}
 &\le\norm{x}_{L^2}\norm{y}_{L^2},\\
 |\tr(\mathsf C_{xy})|
 &\le\norm{x}_{L^2}\norm{y}_{L^2},\\
 \norm{\mathsf C_{xy}-\mathsf C_{x'y}}_{\HS}
 &\le\norm{x-x'}_{L^2}\norm{y}_{L^2}.
\end{align*}
Each covariance pairing in
\eqref{eq:contracted-cumulant-definition} consequently changes by at most
\[
 C (C_4^{\rm mom})^3\norm{x_1-x_1'}_{L^4}.
\]
Summation over the four telescoping replacements proves
\eqref{eq:cumulant-replacement-bound}.

For the independence assertion, expand with an orthonormal basis
$\{v_j\}$:
\[
 \E\{\ip{x_1}{x_3}\ip{x_2}{x_4}\}
 =
 \sum_{j,k}
 \E\!\left[
 \ip{x_1}{v_j}\ip{x_3}{v_j}
 \ip{x_2}{v_k}\ip{x_4}{v_k}
 \right].
\]
Fubini is justified by
\[
 \sum_j|\ip{x}{v_j}\ip{y}{v_j}|
 \le\norm{x}\norm{y}.
\]
For the three two-block partitions, independence gives respectively
\begin{align*}
 (x_1,x_2)\perp(x_3,x_4)
 &\Longrightarrow
 \E\{\ip{x_1}{x_3}\ip{x_2}{x_4}\}
 =\ip{\mathsf C_{12}}{\mathsf C_{34}}_{\HS},\\
 (x_1,x_3)\perp(x_2,x_4)
 &\Longrightarrow
 \E\{\ip{x_1}{x_3}\ip{x_2}{x_4}\}
 =\tr(\mathsf C_{13})\tr(\mathsf C_{24}),\\
 (x_1,x_4)\perp(x_2,x_3)
 &\Longrightarrow
 \E\{\ip{x_1}{x_3}\ip{x_2}{x_4}\}
 =\tr(\mathsf C_{14}\mathsf C_{23}).
\end{align*}
If one centered variable is independent of the other three, the raw moment
and every pairing involving that variable vanish.  Hence
$\kappa_4=0$ for every nontrivial independent partition; see
\citet[Section~2.3]{Brillinger2001}.
\end{proof}

\begin{lemma}[Moment maximal inequality for Bernoulli shifts]
\label{lem:Bernoulli-maximal}
Let $\nu_0\ge2$ and let
$\mathsf Y_i=\mathfrak G_{\mathsf Y}(\eta_i,\eta_{i-1},\ldots)$ be centered
random elements of a real separable Hilbert space $\mathbb H$.  Define
$\mathsf Y_{i,M}$ by the independent-copy construction in
\eqref{eq:ARS-M-approximation}, and put
\begin{equation*}
 \mathfrak d_{\mathsf Y,\nu_0}(M)
 =\norm{\mathsf Y_0-\mathsf Y_{0,M}}_{L^{\nu_0}(\mathbb H)}.
\end{equation*}
If $\norm{\mathsf Y_0}_{L^{\nu_0}(\mathbb H)}<\infty$ and
$\sum_{M\ge1}\mathfrak d_{\mathsf Y,\nu_0}(M)<\infty$, then
\begin{equation}
 \left\|\max_{1\le k\le n}
 \left\|\sum_{i=1}^{k}\mathsf Y_i\right\|_{\mathbb H}\right\|_{L^{\nu_0}}
 \le C_{\nu_0}\sqrt n\left\{\norm{\mathsf Y_0}_{L^{\nu_0}(\mathbb H)}
       +\sum_{M=1}^{\infty}\mathfrak d_{\mathsf Y,\nu_0}(M)\right\}.
 \label{eq:Bernoulli-maximal}
\end{equation}
The constant $C_{\nu_0}$ depends only on $\nu_0$.
\end{lemma}

\begin{proof}
Let
\[
 \mathcal F_i=\sigma(\eta_i,\eta_{i-1},\ldots),
 \qquad
 \mathcal P_i Z=\E(Z\mid\mathcal F_i)-\E(Z\mid\mathcal F_{i-1}).
\]
The tail field of the independent innovations is trivial.  Martingale
convergence therefore gives
\[
 \mathsf Y_i=\sum_{\jmath=0}^{\infty}
 \mathcal P_{i-\jmath}\mathsf Y_i
 \quad\text{in }L^{\nu_0}(\mathbb H).
\]
For $\jmath\ge1$, $\mathsf Y_{i,\jmath}$ is independent of
$\mathcal F_{i-\jmath}$ and has mean zero.  Conditional expectation is a
contraction in $L^{\nu_0}$, so
\begin{align*}
 \mathcal P_{i-\jmath}\mathsf Y_i
 &=\mathcal P_{i-\jmath}(\mathsf Y_i-\mathsf Y_{i,\jmath}),\\
 \norm{\mathcal P_{i-\jmath}\mathsf Y_i}_{L^{\nu_0}}
 &\le2\mathfrak d_{\mathsf Y,\nu_0}(\jmath).
\end{align*}
For $\jmath=0$,
$\norm{\mathcal P_i\mathsf Y_i}_{L^{\nu_0}}
 \le2\norm{\mathsf Y_0}_{L^{\nu_0}}$.

For fixed $\jmath$, put
$\mathsf D_{i,\jmath}=\mathcal P_i\mathsf Y_{i+\jmath}$.  This is a
stationary Hilbert-valued martingale-difference sequence.  By the maximal
inequality of \citet{Pinelis1994},
\begin{align*}
 &\left\|\max_{1\le k\le n}
 \left\|\sum_{i=1}^{k}\mathsf D_{i,\jmath}\right\|\right\|_{L^{\nu_0}}\\
 &\quad\le C_{\nu_0}\left[
 \left\|\left\{\sum_{i=1}^{n}
 \E(\norm{\mathsf D_{i,\jmath}}^2\mid\mathcal F_{i-1})\right\}^{1/2}
 \right\|_{L^{\nu_0}}
 +\left\{\sum_{i=1}^{n}
 \norm{\mathsf D_{i,\jmath}}_{L^{\nu_0}}^{\nu_0}\right\}^{1/\nu_0}
 \right].
\end{align*}
Minkowski's inequality in $L^{\nu_0/2}$ gives
\begin{equation*}
 \left\|\left\{\sum_{i=1}^{n}
 \E(\norm{\mathsf D_{i,\jmath}}^2\mid\mathcal F_{i-1})\right\}^{1/2}
 \right\|_{L^{\nu_0}}
 \le\sqrt n\,\norm{\mathsf D_{0,\jmath}}_{L^{\nu_0}},
\end{equation*}
and the second term is at most
$n^{1/\nu_0}\norm{\mathsf D_{0,\jmath}}_{L^{\nu_0}}
\le\sqrt n\norm{\mathsf D_{0,\jmath}}_{L^{\nu_0}}$.  For every $k$,
\begin{equation*}
 \sum_{i=1}^{k}\mathcal P_{i-\jmath}\mathsf Y_i
 =\sum_{i_0=1-\jmath}^{k-\jmath}\mathsf D_{i_0,\jmath}.
\end{equation*}
Stationarity of the martingale-difference sequence implies that the process on
the right, indexed by $1\le k\le n$, has the same law as
$\{\sum_{i_0=1}^{k}\mathsf D_{i_0,\jmath}:1\le k\le n\}$.  Consequently,
\begin{equation*}
 \left\|\max_{1\le k\le n}
 \left\|\sum_{i=1}^{k}
 \mathcal P_{i-\jmath}\mathsf Y_i\right\|\right\|_{L^{\nu_0}}
 \le
 \begin{cases}
 C_{\nu_0}\sqrt n\norm{\mathsf Y_0}_{L^{\nu_0}},&\jmath=0,\\
 C_{\nu_0}\sqrt n\mathfrak d_{\mathsf Y,\nu_0}(\jmath),&\jmath\ge1.
 \end{cases}
\end{equation*}

For $R_{\rm proj}<\infty$, Minkowski's inequality over $\jmath$ now gives
\begin{align*}
 &\left\|\max_{1\le k\le n}
 \left\|\sum_{i=1}^{k}\sum_{\jmath=0}^{R_{\rm proj}}
 \mathcal P_{i-\jmath}\mathsf Y_i\right\|\right\|_{L^{\nu_0}}\\
 &\quad\le C_{\nu_0}\sqrt n\left\{
 \norm{\mathsf Y_0}_{L^{\nu_0}}
 +\sum_{\jmath=1}^{R_{\rm proj}}
 \mathfrak d_{\mathsf Y,\nu_0}(\jmath)\right\}.
\end{align*}
The same bound applied to a projection tail
$R_1<\jmath\le R_2$ is at most
\[
 C_{\nu_0}\sqrt n\sum_{\jmath>R_1}
 \mathfrak d_{\mathsf Y,\nu_0}(\jmath)
 \longrightarrow0
 \qquad (R_1\to\infty).
\]
Thus the truncated partial-sum processes are Cauchy in
$L^{\nu_0}\{\ell^\infty(\{1,\ldots,n\},\mathbb H)\}$, and their
coordinatewise limit is the partial-sum process of $\mathsf Y_i$.  Letting
$R_{\rm proj}\to\infty$ proves \eqref{eq:Bernoulli-maximal}.
\end{proof}

\begin{lemma}[Gaussian coupling from ARS approximability]
\label{lem:ARS-FCLT}
Let $\mathsf Y_i=\mathfrak G_{\mathsf Y}(\eta_i,\eta_{i-1},\ldots)$ be centered random elements of
$\cH=L^2[0,1]$, and define $\mathsf Y_{i,M}$ by the independent-copy construction
in \eqref{eq:ARS-M-approximation}.  Suppose that, for some $\nu_0>2$,
\begin{align}
 &\E\norm{\mathsf Y_0}_{\cH}^{\nu_0}<\infty,
 \notag\\
 &\sum_{M=1}^{\infty}
 \left\{\E\norm{\mathsf Y_0-\mathsf Y_{0,M}}_{\cH}^{\nu_0}\right\}^{1/\nu_0}<\infty.
 \label{eq:Jirak-condition}
\end{align}
Let
\begin{equation*}
 \Omega_{\mathsf Y}=\sum_{u\in\mathbb Z}\E(\mathsf Y_u\otimes \mathsf Y_0),
\end{equation*}
where the series exists in trace norm.  On an extension of the probability
space there are centered Gaussian processes $\mathbb G_n(t)$,
$0\le t\le1$, such that
\begin{equation*}
 \E\{\mathbb G_n(t)\otimes\mathbb G_n(t')\}
 =(t\wedge t')\Omega_{\mathsf Y}
\end{equation*}
and
\begin{equation*}
 \sup_{0\le t\le1}
 \left\|\frac1{\sqrt n}\sum_{i=1}^{\lfloor nt\rfloor}\mathsf Y_i
       -\mathbb G_n(t)\right\|_{\cH}^{2}=o_{\Pp}(1).
\end{equation*}
Consequently,
\begin{equation}
 \frac1{\sqrt n}\sum_{i=1}^{\lfloor nt\rfloor}\mathsf Y_i
 \dto W_{\Omega_{\mathsf Y}}(t)
 \quad\text{in }D([0,1],\cH),
 \label{eq:general-ARS-FCLT}
\end{equation}
where $W_{\Omega_{\mathsf Y}}$ is Brownian motion with covariance operator
$\Omega_{\mathsf Y}$.
\end{lemma}

\begin{proof}
Assumption~2.1 of \citet{AueRiceSonmez2018} requires a Bernoulli-shift
representation and the unweighted independent-copy summability in
\eqref{eq:Jirak-condition}.  Theorem~A.1 of
\citet{AueRiceSonmez2018Supplement}, which restates Theorem~1.2 of
\citet{Jirak2013}, therefore applies directly to $\{\mathsf Y_i\}$.  It
provides, on an extension of the probability space, centered Gaussian
processes $\mathbb G_n$ satisfying
\begin{align*}
 \E\{\mathbb G_n(t)\otimes\mathbb G_n(t')\}
 &=(t\wedge t')\Omega_{\mathsf Y},\\
 \sup_{0\le t\le1}
 \left\|
 n^{-1/2}\sum_{i=1}^{\lfloor nt\rfloor}\mathsf Y_i-\mathbb G_n(t)
 \right\|^2
 &=o_{\Pp}(1).
\end{align*}
The trace-norm existence of $\Omega_{\mathsf Y}$ assumed in the statement
ensures that the displayed covariance defines an $\cH$-valued Brownian
motion with continuous paths.  Since this covariance does not depend on
$n$, every $\mathbb G_n$ has the law of $W_{\Omega_{\mathsf Y}}$.
The uniform coupling error then implies
\eqref{eq:general-ARS-FCLT} in the Skorokhod space by Slutsky's theorem.
The cited construction uses $M$-dependent approximation, blocking and
Gaussian coupling; it does not impose Gaussianity or a linear
representation on $\{\mathsf Y_i\}$.
\end{proof}

\begin{lemma}[Resolvent and quadratic-form replacement]
\label{lem:resolvent-replacement}
Let $\mathcal V_n$ and $\mathcal V$ be self-adjoint nonnegative operators, and let
$\mathfrak s_n,\mathfrak s>0$.  For $a\in[a_-,a_+]$, define
\[
 \mathsf R_{n,a}^{\rm gen}=(\mathcal V_n+a\mathfrak s_n\Id)^{-1},
 \qquad
 \mathsf R_a^{\rm gen}=(\mathcal V+a\mathfrak s\Id)^{-1}.
\]
On the event $\mathfrak s_n\ge\mathfrak s/2$,
\begin{equation}
 \sup_{a\in[a_-,a_+]}\norm{\mathsf R_{n,a}^{\rm gen}-\mathsf R_a^{\rm gen}}_{\op}
 \le C\{\norm{\mathcal V_n-\mathcal V}_{\op}
          +|\mathfrak s_n-\mathfrak s|\}.
 \label{eq:uniform-resolvent-bound}
\end{equation}
If $Z_n^{\rm qf}$ is a random element of $\ell^\infty(\mathbb T_{\rm qf},\cH)$ satisfying
$\sup_{t\in\mathbb T_{\rm qf}}\norm{Z_n^{\rm qf}(t)}=O_{\Pp}(1)$, then
\begin{align}
 &\sup_{a\in[a_-,a_+]}
 \sup_{t\in\mathbb T_{\rm qf}}
 \left|\ip{Z_n^{\rm qf}(t)}
 {\mathsf R_{n,a}^{\rm gen}Z_n^{\rm qf}(t)}\right. \notag\\
 &\hspace{8em}\left.
 -\ip{Z_n^{\rm qf}(t)}
 {\mathsf R_a^{\rm gen}Z_n^{\rm qf}(t)}\right| \notag\\
 &\quad=O_{\Pp}\{\norm{\mathcal V_n-\mathcal V}_{\op}
             +|\mathfrak s_n-\mathfrak s|\}.
 \label{eq:quadratic-resolvent-bound}
\end{align}
\end{lemma}

\begin{proof}
The resolvent identity yields
\[
 \mathsf R_{n,a}^{\rm gen}-\mathsf R_a^{\rm gen}
 =\mathsf R_{n,a}^{\rm gen}\{\mathcal V-\mathcal V_n
       +a(\mathfrak s-\mathfrak s_n)\Id\}\mathsf R_a^{\rm gen}.
\]
On $\{\mathfrak s_n\ge\mathfrak s/2\}$,
\[
 \sup_{a\in[a_-,a_+]}\norm{\mathsf R_{n,a}^{\rm gen}}_{\op}
 \le\frac{2}{a_-\mathfrak s},
 \qquad
 \sup_{a\in[a_-,a_+]}\norm{\mathsf R_a^{\rm gen}}_{\op}
 \le\frac{1}{a_-\mathfrak s}.
\]
Consequently,
\begin{align*}
 \sup_{a\in[a_-,a_+]}\norm{\mathsf R_{n,a}^{\rm gen}-\mathsf R_a^{\rm gen}}_{\op}
 &\le
 \frac{2}{a_-^2\mathfrak s^2}
 \{\norm{\mathcal V_n-\mathcal V}_{\op}
    +a_+|\mathfrak s_n-\mathfrak s|\}\\
 &\le C\{\norm{\mathcal V_n-\mathcal V}_{\op}
           +|\mathfrak s_n-\mathfrak s|\},
\end{align*}
which proves \eqref{eq:uniform-resolvent-bound}.  Moreover,
\begin{align*}
 &\sup_{a,t}
 \left|\ip{Z_n^{\rm qf}(t)}
              {(\mathsf R_{n,a}^{\rm gen}-\mathsf R_a^{\rm gen})Z_n^{\rm qf}(t)}\right|\\
 &\quad\le
 \sup_a\norm{\mathsf R_{n,a}^{\rm gen}-\mathsf R_a^{\rm gen}}_{\op}
 \sup_t\norm{Z_n^{\rm qf}(t)}^2,
\end{align*}
and \eqref{eq:quadratic-resolvent-bound} follows.
\end{proof}

\section{Proofs for the functional model and numerical projection}
\label{app:functional-proofs}

\subsection{Proofs of Lemmas~\ref{prop:lrc-origin} and~\ref{lem:ARS-main-FCLT}}

For $u\ge1$, $e_{u,u}$ is independent of $e_0$, has the same
marginal law as $e_u$, and satisfies $\E e_{u,u}=0$.  Hence
\begin{equation*}
 \mathcal C_u
 =\E\{(e_u-e_{u,u})\otimes e_0\}.
\end{equation*}
For $X,Y\in L^2(\cH)$,
\begin{align}
 \norm{\E(X\otimes Y)}_1
 &=\sup_{\norm A_{\op}\le1}
   |\tr\{A\E(X\otimes Y)\}|
 \notag\\
 &=\sup_{\norm A_{\op}\le1}
   |\E\ip{Y}{AX}|
 \le\norm X_{L^2(\cH)}\norm Y_{L^2(\cH)}.
 \label{eq:rank-one-expectation-trace}
\end{align}
The embedding $\Hs\hookrightarrow\cH$ gives
$\norm{e_0-e_{0,u}}_{L^2(\cH)}\le C\mathfrak d_\nu(u)$, and therefore
\begin{equation*}
 \norm{\mathcal C_u}_1
 \le C\mathfrak d_\nu(u),\qquad u\ge1.
\end{equation*}
Since $\mathcal C_{-u}=\mathcal C_u^*$,
\begin{align*}
 \sum_{u\in\mathbb Z}(1+|u|)^{\beta_{\rm dep}}\norm{\mathcal C_u}_1
 &\le C+C\sum_{u\ge1}(1+u)^{\beta_{\rm dep}}\mathfrak d_\nu(u)
 <\infty.
\end{align*}
This proves \eqref{eq:weighted-cov-summability}.

For $H_{\rm B}\ge1$, define the Bartlett sum
\[
 \Omega_{H_{\rm B}}
 =\sum_{|u|<H_{\rm B}}\left(1-\frac{|u|}{H_{\rm B}}\right)\mathcal C_u
 =\frac1{H_{\rm B}}\E\left[
 \left(\sum_{i=1}^{H_{\rm B}}e_i\right)\otimes
 \left(\sum_{i=1}^{H_{\rm B}}e_i\right)\right]\succeq0.
\]
Moreover,
\begin{align*}
 \norm{\Omega_{H_{\rm B}}-\Omega}_1
 &\le
 \frac1{H_{\rm B}}\sum_{|u|<H_{\rm B}}|u|\norm{\mathcal C_u}_1
 +\sum_{|u|\ge H_{\rm B}}\norm{\mathcal C_u}_1
 \longrightarrow0.
\end{align*}
Thus $\Omega$ exists in trace norm, is trace class, and
$\Omega\succeq0$.

We next verify \eqref{eq:contracted-cumulant-summability}.  For labelled
times $t_1,\ldots,t_4$, let $\pi_{\rm ord}$ satisfy
\[
 t_{\pi_{\rm ord}(1)}\le t_{\pi_{\rm ord}(2)}
 \le t_{\pi_{\rm ord}(3)}\le t_{\pi_{\rm ord}(4)},
 \qquad
 \mathsf g_{j_{\rm gap}}^{\rm gap}
 =t_{\pi_{\rm ord}(j_{\rm gap}+1)}-t_{\pi_{\rm ord}(j_{\rm gap})},
 \quad1\le j_{\rm gap}\le3.
\]
Set
\[
 j_{\rm gap}^{\star}
 =\min\arg\max_{1\le j_{\rm gap}\le3}\mathsf g_{j_{\rm gap}}^{\rm gap},
 \qquad
 \mathsf g_{\star}^{\rm gap}=\mathsf g_{j_{\rm gap}^{\star}}^{\rm gap},
 \qquad
 t_{\star}=t_{\pi_{\rm ord}(j_{\rm gap}^{\star})}.
\]
Let $\{\eta_{i'}^\dagger:i'\le t_{\star}\}$ be an independent copy of
$\{\eta_{i'}:i'\le t_{\star}\}$ and put
\[
 \bar\eta_{i'}=\eta_{i'}\ind(i'>t_{\star})
 +\eta_{i'}^\dagger\ind(i'\le t_{\star}).
\]
For each original label $j_{\rm lab}\in\{1,\ldots,4\}$ define
\[
 \widetilde e_{t_{j_{\rm lab}}}=
 \begin{cases}
  e_{t_{j_{\rm lab}}},&t_{j_{\rm lab}}\le t_{\star},\\
  \mathfrak G(\bar\eta_{t_{j_{\rm lab}}},
              \bar\eta_{t_{j_{\rm lab}}-1},\ldots),
              &t_{j_{\rm lab}}>t_{\star}.
 \end{cases}
\]
The two labelled subcollections separated by the largest gap are
independent, while the joint law inside each subcollection is unchanged.
Consequently,
\[
 \kappa_4(\widetilde e_{t_1},\widetilde e_{t_2},
          \widetilde e_{t_3},\widetilde e_{t_4})=0.
\]
For $t_{j_{\rm lab}}>t_{\star}$,
\begin{equation}
 \norm{e_{t_{j_{\rm lab}}}-\widetilde e_{t_{j_{\rm lab}}}}_{L^4(\cH)}
 \le\mathfrak d_\nu(t_{j_{\rm lab}}-t_{\star}).
 \label{eq:largest-gap-ARS-coupling}
\end{equation}
Lemma~\ref{lem:cumulant-replacement} yields
\begin{equation}
 |\kappa_4(e_{t_1},e_{t_2},e_{t_3},e_{t_4})|
 \le C\sum_{j_{\rm ord}=j_{\rm gap}^{\star}+1}^{4}
 \mathfrak d_\nu\{t_{\pi_{\rm ord}(j_{\rm ord})}-t_{\star}\}.
 \label{eq:cumulant-largest-gap-ARS}
\end{equation}

For $\kappa_e(u_1,u_2,u_3)=\kappa_4(e_0,e_{u_1},e_{u_2},e_{u_3})$, the position of zero
and the assignments of the other labels contribute only a fixed factor.
Let $\mathsf g_{\max}^{\rm gap}$ be the selected largest gap, and let
$\mathsf g_{2}^{\rm rem}$ and $\mathsf g_{3}^{\rm rem}$ denote the remaining gaps. Then
$0\le \mathsf g_{2}^{\rm rem},\mathsf g_{3}^{\rm rem}\le \mathsf g_{\max}^{\rm gap}$ and
\begin{align*}
 &\sum_{\mathsf g_{\max}^{\rm gap}\ge1}
   \sum_{\mathsf g_{2}^{\rm rem}=0}^{\mathsf g_{\max}^{\rm gap}}
   \sum_{\mathsf g_{3}^{\rm rem}=0}^{\mathsf g_{\max}^{\rm gap}}
 \Big\{\mathfrak d_\nu(\mathsf g_{\max}^{\rm gap})
       +\mathfrak d_\nu(\mathsf g_{\max}^{\rm gap}+\mathsf g_{2}^{\rm rem})
       +\mathfrak d_\nu(\mathsf g_{\max}^{\rm gap}+\mathsf g_{2}^{\rm rem}+\mathsf g_{3}^{\rm rem})\Big\}\\
 &\quad\le
 C\sum_{M\ge1}(1+M)^2\mathfrak d_\nu(M).
\end{align*}
Indeed, the coefficient of $\mathfrak d_\nu(M)$ is bounded by
\[
 C\{1+M+(1+M)^2\}\le C(1+M)^2.
\]
The cases in which either remaining gap attains the maximum give the same bound.
Therefore
\[
 \sum_{u_1,u_2,u_3\in\mathbb Z}|\kappa_e(u_1,u_2,u_3)|
 \le C+C\sum_{M\ge1}(1+M)^2\mathfrak d_\nu(M)<\infty,
\]
which proves \eqref{eq:contracted-cumulant-summability}.

Finally,
\[
 \sum_{M\ge1}
 \{\E\norm{e_0-e_{0,M}}^{\nu}\}^{1/\nu}
 \le C\sum_{M\ge1}\mathfrak d_\nu(M)<\infty.
\]
Lemma~\ref{lem:ARS-FCLT}, with $\mathsf Y_i=e_i$ and $\Omega_{\mathsf Y}=\Omega$,
gives \eqref{eq:ARS-partial-sum-FCLT}.  For
\[
 S_n(t)=n^{-1/2}\sum_{i=1}^{\lfloor nt\rfloor}e_i,
\]
the tied-down map satisfies
\begin{align*}
 &\sup_{0\le t\le1}
 \left\|S_n(t)-\frac{\lfloor nt\rfloor}{n}S_n(1)
       -\{S_n(t)-tS_n(1)\}\right\|\\
 &\quad\le n^{-1}\norm{S_n(1)}
 =O_{\Pp}(n^{-1}).
\end{align*}
The continuous mapping theorem proves
\eqref{eq:ARS-bridge-FCLT}.

\subsection{Proofs of Lemmas~\ref{prop:quadrature} and~\ref{lem:projected-dependence}}

For $f\in C^2[0,1]$, the composite trapezoidal rule gives
\begin{equation*}
 \left|Q_N(f)-\int_0^1f(t)\,dt\right|
 \le\frac{\norm{f''}_\infty}{12(N-1)^2}.
\end{equation*}
For $f\in C^2[0,1]$,
\[
 (f\phi_j)''=f''\phi_j+2f'\phi_j'+f\phi_j'',
\]
so Assumption~\ref{ass:grid} gives
\[
 \max_{1\le j\le p}\norm{(f\phi_j)''}_\infty
 \le Cp^2\norm{f}_{C^2}.
\]
Summing the coordinatewise squared quadrature errors yields
\[
 \norm{(\mathcal J_{p,N}-\Pi_p)f}
 \le C\chi_{p,N}\norm{f}_{C^2},
\]
which proves \eqref{eq:numerical-projection-operator}.

For a fixed $B<\infty$, let
\[
 \mathcal B_B=\{f\in C^2[0,1]:\norm f_{C^2}\le B\}.
\]
This set is relatively compact in $L^2[0,1]$ by the Arzel\`a--Ascoli
theorem.  Since $\Pi_p\to\Id_{\cH}$ strongly and
$\sup_p\norm{\Id_{\cH}-\Pi_p}_{\op}\le1$, the convergence is uniform on the
$L^2$-closure of $\mathcal B_B$.  Consequently,
\begin{align*}
 \mathfrak a_{p,N}(B)
 &\le \sup_{f\in\mathcal B_B}\norm{(\Id_{\cH}-\Pi_p)f}
 +\sup_{f\in\mathcal B_B}
   \norm{(\mathcal J_{p,N}-\Pi_p)f}\\
 &\le \sup_{f\in\mathcal B_B}\norm{(\Id_{\cH}-\Pi_p)f}
 +CB\chi_{p,N}\longrightarrow0,
\end{align*}
which proves \eqref{eq:uniform-functional-resolution-limit}.

Let $\mathcal J_n=\mathcal J_{p,N}$.  For every fixed $M$,
\[
 \norm{(\Id_{\cH}-\Pi_p)(e_0-e_{0,M})}_{L^2(\cH)}
 \longrightarrow0,
\]
because $\Pi_p\to\Id_{\cH}$ strongly.  Furthermore,
\[
 \norm{(\Id_{\cH}-\Pi_p)(e_0-e_{0,M})}_{L^2(\cH)}
 \le\norm{e_0-e_{0,M}}_{L^2(\cH)}
 \le C\mathfrak d_\nu(M).
\]
Since $\sum_M\mathfrak d_\nu(M)<\infty$,
\begin{align*}
 \varrho_p
 &=\norm{(\Id_{\cH}-\Pi_p)e_0}_{L^2(\cH)}
 +\sum_{M\ge1}
 \norm{(\Id_{\cH}-\Pi_p)(e_0-e_{0,M})}_{L^2(\cH)}
 \longrightarrow0.
\end{align*}
The Sobolev embedding and
\eqref{eq:numerical-projection-operator} also give
\begin{equation*}
 \norm{(\mathcal J_n-\Pi_p)f}
 \le C\chi_{p,N}\norm f_{\Hs},
 \qquad f\in\Hs.
\end{equation*}

Apply Lemma~\ref{lem:Bernoulli-maximal} to
\[
 \mathsf Y_{i,p}^{(1)}=(\Pi_p-\Id_{\cH})e_i.
\]
Its maximal-inequality coefficient is
\begin{align*}
 &\norm{\mathsf Y_{0,p}^{(1)}}_{L^2}
 +\sum_{M\ge1}
 \norm{\mathsf Y_{0,p}^{(1)}-(\Pi_p-\Id_{\cH})e_{0,M}}_{L^2}\\
 &\quad=
 \norm{(\Id_{\cH}-\Pi_p)e_0}_{L^2}
 +\sum_{M\ge1}
 \norm{(\Id_{\cH}-\Pi_p)(e_0-e_{0,M})}_{L^2}
 =\varrho_p.
\end{align*}
Apply the same lemma to
\[
 \mathsf Y_{i,p,N}^{(2)}=(\mathcal J_n-\Pi_p)e_i.
\]
By \eqref{eq:numerical-projection-operator} and the Sobolev embedding,
\begin{align*}
 \norm{\mathsf Y_{0,p,N}^{(2)}}_{L^2}
 &\le C\chi_{p,N},\\
 \norm{\mathsf Y_{0,p,N}^{(2)}-(\mathcal J_n-\Pi_p)e_{0,M}}_{L^2}
 &\le C\chi_{p,N}\mathfrak d_\nu(M).
\end{align*}
Therefore
\begin{align*}
 &\left[
 \E\sup_{0\le t\le1}
 \left\|
 n^{-1/2}\sum_{i=1}^{\lfloor nt\rfloor}(\mathcal J_n-\Id_{\cH})e_i
 \right\|^2
 \right]^{1/2}\\
 &\quad\le C(\varrho_p+\chi_{p,N}),
\end{align*}
which proves \eqref{eq:numerical-partial-sum-equivalence}.

For all sufficiently large $n$,
\begin{align}
 \norm{\mathcal J_nf}
 &\le\norm{\Pi_pf}+\norm{(\mathcal J_n-\Pi_p)f}
 \le C\norm f_{\Hs}.
 \label{eq:numerical-map-uniform-bound}
\end{align}
For $u\ge1$, $\mathcal J_ne_{u,u}$ is independent of $\mathcal J_ne_0$, so
\[
 \Gamma_{p,N}^{\cH}(u)
 =\E\{\mathcal J_n(e_u-e_{u,u})\otimes \mathcal J_ne_0\}.
\]
Equations \eqref{eq:rank-one-expectation-trace} and
\eqref{eq:numerical-map-uniform-bound} imply
\begin{equation*}
 \norm{\Gamma_{p,N}^{\cH}(u)}_1
 =\norm{\bGamma_{p,N}(u)}_1\le C\mathfrak d_\nu(u).
\end{equation*}
Hence
\[
 \sup_n\sum_{u\in\mathbb Z}(1+|u|)^{\beta_{\rm dep}}
 \norm{\bGamma_{p,N}(u)}_1<\infty,
\]
which proves \eqref{eq:projected-weighted-cov-sum}; the Bartlett-sum
identity then gives $\bOmega_{p,N}\succeq0$.

For $u\ge1$,
\begin{align*}
 \Gamma_{p,N}^{\cH}(u)-\Pi_p\mathcal C_u\Pi_p
 &=\E\{(\mathcal J_n-\Pi_p)(e_u-e_{u,u})\otimes \mathcal J_ne_0\}\\
 &\quad+\E\{\Pi_p(e_u-e_{u,u})\otimes(\mathcal J_n-\Pi_p)e_0\},
\end{align*}
and therefore
\[
 \norm{\Gamma_{p,N}^{\cH}(u)-\Pi_p\mathcal C_u\Pi_p}_1
 \le C\chi_{p,N}\mathfrak d_\nu(u).
\]
At $u=0$,
\[
 \norm{\Gamma_{p,N}^{\cH}(0)-\Pi_p\mathcal C_0\Pi_p}_1
 \le C\chi_{p,N}.
\]
Summation over $u\in\mathbb Z$ gives
\[
 \norm{\Omega_{p,N}^{\cH}-\Pi_p\Omega\Pi_p}_1
 \le C\chi_{p,N}.
\]
Thus
\begin{align*}
 \norm{\Omega_{p,N}^{\cH}-\Omega}_1
 &\le C\chi_{p,N}+\norm{\Pi_p\Omega\Pi_p-\Omega}_1\\
 &\longrightarrow0
\end{align*}
by Lemma~\ref{lem:trace-projection}, proving
\eqref{eq:lrc-projection-error}.

Finally, $\{\mathcal J_ne_i\}$ is a Bernoulli shift with approximations
$\{\mathcal J_ne_{i,M}\}$, and
\[
 \norm{\mathcal J_n(e_0-e_{0,M})}_{L^\nu(\cH)}
 \le C\mathfrak d_\nu(M).
\]
The largest-gap replacement used in
\eqref{eq:largest-gap-ARS-coupling} therefore yields, uniformly in
$n,p,N$,
\[
 |\kappa_{p,N}(u_1,u_2,u_3)|
 \le C\sum_{\text{right of the largest gap}}\mathfrak d_\nu(\text{distance}).
\]
The same three-gap count as in
\eqref{eq:cumulant-largest-gap-ARS} gives
\[
 \sup_n\sum_{u_1,u_2,u_3\in\mathbb Z}|\kappa_{p,N}(u_1,u_2,u_3)|
 \le C+C\sum_{M\ge1}(1+M)^2\mathfrak d_\nu(M)<\infty,
\]
which proves \eqref{eq:projected-cumulant-summability}.

\section{Proofs for the edge-corrected difference-based estimator}

\subsection{Proof of Lemma~\ref{prop:edge-correction}}

Under $H_0$, for $i_{\rm D}=1,\ldots,N_D$,
\[
 \bY_{i_{\rm D}}
 =\sum_{q=0}^{m}c_q
 \{\bmu_{p,N}+\beps_{mh+i_{\rm D}-qh,p,N}\}
 =\beps_{i_{\rm D}}^{\rm DB},
\]
because $\sum_{q=0}^{m}c_q=0$.  Stationarity gives, for
$0\le r<\ell$,
\begin{align*}
 \E\widehat\bGamma_{\rm DB}(r)
 &=\frac1{N_D-r}\sum_{i_{\rm D}=r+1}^{N_D}
   \E\{\beps_{i_{\rm D}}^{\rm DB}
   (\beps_{i_{\rm D}-r}^{\rm DB})^\top\}\\
 &=\bGamma_{\rm DB}(r),
\end{align*}
which proves \eqref{eq:edge-unbiased}.  With the common divisor $n$,
\[
 \E\widehat\bGamma_{\rm DB}^{\rm com}(r)
 =\frac{N_D-r}{n}\bGamma_{\rm DB}(r),
\]
which is \eqref{eq:paper-denominator-bias}.

The filtered covariance satisfies
\[
 \bGamma_{\rm DB}(r)
 =\sum_{q_1,q_2=0}^{m}c_{q_1}c_{q_2}
 \bGamma_{p,N}\{r+(q_2-q_1)h\}.
\]
With $\mathfrak b_r=\norm{\bGamma_{\rm DB}(r)}_1$ and fixed $m$,
\begin{align}
 \sum_{r\in\mathbb Z}\mathfrak b_r
 &\le
 \sum_{q_1,q_2=0}^{m}|c_{q_1}c_{q_2}|
 \sum_{u\in\mathbb Z}\norm{\bGamma_{p,N}(u)}_1
 \le C.
 \label{eq:filtered-cov-summability-audited}
\end{align}
Therefore
\begin{align*}
 \norm{\E\widehat\bOmega_n^{\rm com}-\bOmega_{{\rm DB},\ell}}_1
 &\le
 \sum_{|r|<\ell}|K(r/\ell)|
 \left|1-\frac{N_D-|r|}{n}\right|\mathfrak b_r\\
 &\le
 \frac{mh+\ell}{n}\sum_{r\in\mathbb Z}\mathfrak b_r
 \le C\frac{mh+\ell}{n},
\end{align*}
which proves \eqref{eq:endpoint-bias-bound}.

\subsection{Proof of Lemma~\ref{prop:population-DB-bias}}

The identity
\begin{equation*}
 \bGamma_{\rm DB}(r)
 =\sum_{q_1,q_2=0}^{m}c_{q_1}c_{q_2}
 \bGamma_{p,N}\{r+(q_2-q_1)h\}
\end{equation*}
splits into diagonal and off-diagonal filter terms.  Since
$\sum_{q=0}^{m}c_q^2=1$,
\begin{align*}
 &\left\|
 \sum_{|r|<\ell}K(r/\ell)\bGamma_{p,N}(r)
 -\bOmega_{p,N}
 \right\|_1
 \notag\\
 &\quad\le
 C_K\ell^{-\beta_K}
 \sum_{|r|<\ell}|r|^{\beta_K}\norm{\bGamma_{p,N}(r)}_1
 +\sum_{|r|\ge\ell}\norm{\bGamma_{p,N}(r)}_1\\
 &\quad\le C\ell^{-\beta_K}.
\end{align*}
For $q_1\ne q_2$ and $|r|<\ell$,
\[
 |r+(q_2-q_1)h|
 \ge h-|r|\ge\ell.
\]
For each fixed $(q_1,q_2)$, the map $r\mapsto r+(q_2-q_1)h$ is injective.
Consequently,
\begin{align*}
 &\left\|
 \sum_{|r|<\ell}K(r/\ell)
 \sum_{q_1\ne q_2}c_{q_1}c_{q_2}
 \bGamma_{p,N}\{r+(q_2-q_1)h\}
 \right\|_1\\
 &\quad\le
 C\sum_{|u|\ge\ell}\norm{\bGamma_{p,N}(u)}_1
 \le C\ell^{-\beta_K}.
\end{align*}
Combining the two displays proves
\eqref{eq:population-DB-rate}.

\subsection{Proofs of Lemmas~\ref{thm:DB-rate} and~\ref{lem:positive-part-rate}}

For $1\le i_{{\rm D},1},i_{{\rm D},2}\le N_D$, define
\[
 \mathsf w_{i_{{\rm D},1}i_{{\rm D},2}}^{\rm DB}
 =\frac{K\{(i_{{\rm D},1}-i_{{\rm D},2})/\ell\}}
 {N_D-|i_{{\rm D},1}-i_{{\rm D},2}|}
 \ind(|i_{{\rm D},1}-i_{{\rm D},2}|<\ell).
\]
Since $K$ is even,
\begin{equation*}
 \widehat\bOmega_n^{\rm DB}
 =\sum_{i_{{\rm D},1},i_{{\rm D},2}=1}^{N_D}
 \mathsf w_{i_{{\rm D},1}i_{{\rm D},2}}^{\rm DB}
 \beps_{i_{{\rm D},1}}^{\rm DB}\otimes
 \beps_{i_{{\rm D},2}}^{\rm DB}.
\end{equation*}
Let
\[
 \mathsf W_{\rm DB}
 =\big(|\mathsf w_{i_{{\rm D},1}i_{{\rm D},2}}^{\rm DB}|\big)
 _{i_{{\rm D},1},i_{{\rm D},2}=1}^{N_D},
 \qquad
 \mathsf C_{\rm DB}
 =\big(\mathfrak b_{i_{{\rm D},1}-i_{{\rm D},2}}\big)
 _{i_{{\rm D},1},i_{{\rm D},2}=1}^{N_D}.
\]
By \eqref{eq:filtered-cov-summability-audited},
\begin{equation*}
 \norm{\mathsf C_{\rm DB}}_{\op}
 \le\max_{i_{{\rm D},1}}\sum_{i_{{\rm D},2}}
 \mathfrak b_{i_{{\rm D},1}-i_{{\rm D},2}}\le C.
\end{equation*}
Moreover,
\begin{align*}
 \norm{\mathsf W_{\rm DB}}_{\HS}^2
 &=\sum_{|r|<\ell}
 (N_D-|r|)\frac{K(r/\ell)^2}{(N_D-|r|)^2}\\
 &\le\frac{C\ell}{N_D-\ell}
 \le C\frac{\ell}{N_D}.
\end{align*}

For
\[
 \kappa_{\rm DB}(i_{{\rm D},1},i_{{\rm D},2},i_{{\rm D},3},i_{{\rm D},4})
 =\kappa_4(\beps_{i_{{\rm D},1}}^{\rm DB},
           \beps_{i_{{\rm D},2}}^{\rm DB},
           \beps_{i_{{\rm D},3}}^{\rm DB},
           \beps_{i_{{\rm D},4}}^{\rm DB}),
\]
finite expansion of the difference filter gives
\begin{align}
 \sup_{i_{{\rm D},1}\in\mathbb Z}
 \sum_{i_{{\rm D},2},i_{{\rm D},3},i_{{\rm D},4}\in\mathbb Z}
 |\kappa_{\rm DB}(i_{{\rm D},1},i_{{\rm D},2},
                   i_{{\rm D},3},i_{{\rm D},4})|
 &\le C_c^4
 \sum_{u_1,u_2,u_3\in\mathbb Z}|\kappa_{p,N}(u_1,u_2,u_3)|
 \le C.
 \label{eq:difference-cumulant-sum}
\end{align}
Also,
\[
 \sup_{n,p,N,h}\E\norm{\beps_0^{\rm DB}}^{4}
 \le C_c^4\sup_n\E\norm{\mathcal J_ne_0}^{4}<\infty.
\]

Applying Lemma~\ref{lem:rank-one-cumulant} yields
\begin{align*}
 &\left|
 \E\ip{\beps_{i_{{\rm D},1}}^{\rm DB}\otimes
          \beps_{i_{{\rm D},2}}^{\rm DB}-
 \E(\beps_{i_{{\rm D},1}}^{\rm DB}\otimes
    \beps_{i_{{\rm D},2}}^{\rm DB})}
 {\beps_{i_{{\rm D},3}}^{\rm DB}\otimes
  \beps_{i_{{\rm D},4}}^{\rm DB}-
 \E(\beps_{i_{{\rm D},3}}^{\rm DB}\otimes
    \beps_{i_{{\rm D},4}}^{\rm DB})}_{\HS}
 \right|\\
 &\quad\le
 \mathfrak b_{i_{{\rm D},1}-i_{{\rm D},3}}
 \mathfrak b_{i_{{\rm D},2}-i_{{\rm D},4}}
 +\mathfrak b_{i_{{\rm D},1}-i_{{\rm D},4}}
 \mathfrak b_{i_{{\rm D},2}-i_{{\rm D},3}}\\
 &\qquad
 +|\kappa_{\rm DB}(i_{{\rm D},1},i_{{\rm D},2},
                    i_{{\rm D},3},i_{{\rm D},4})|.
\end{align*}
Hence
\begin{align*}
 &\E\norm{\widehat\bOmega_n^{\rm DB}
 -\E\widehat\bOmega_n^{\rm DB}}_{\HS}^2\\
 &\quad\le
 \sum_{i_{{\rm D},1},i_{{\rm D},2},i_{{\rm D},3},i_{{\rm D},4}}
 |\mathsf w_{i_{{\rm D},1}i_{{\rm D},2}}^{\rm DB}|
 |\mathsf w_{i_{{\rm D},3}i_{{\rm D},4}}^{\rm DB}|\\
 &\qquad\times
 \Big\{\mathfrak b_{i_{{\rm D},1}-i_{{\rm D},3}}
          \mathfrak b_{i_{{\rm D},2}-i_{{\rm D},4}}
       +\mathfrak b_{i_{{\rm D},1}-i_{{\rm D},4}}
          \mathfrak b_{i_{{\rm D},2}-i_{{\rm D},3}}\\
 &\hspace{8em}
       +|\kappa_{\rm DB}(i_{{\rm D},1},i_{{\rm D},2},
                          i_{{\rm D},3},i_{{\rm D},4})|\Big\}.
\end{align*}
The first pairing satisfies
\begin{align*}
 &\sum_{i_{{\rm D},1},i_{{\rm D},2},i_{{\rm D},3},i_{{\rm D},4}}
 |\mathsf w_{i_{{\rm D},1}i_{{\rm D},2}}^{\rm DB}|
 |\mathsf w_{i_{{\rm D},3}i_{{\rm D},4}}^{\rm DB}|
 \mathfrak b_{i_{{\rm D},1}-i_{{\rm D},3}}
 \mathfrak b_{i_{{\rm D},2}-i_{{\rm D},4}}\\
 &\qquad=\tr(\mathsf W_{\rm DB}^{\top}\mathsf C_{\rm DB}
              \mathsf W_{\rm DB}\mathsf C_{\rm DB}^{\top})
 \le\norm{\mathsf C_{\rm DB}}_{\op}^2
     \norm{\mathsf W_{\rm DB}}_{\HS}^2
 \le C\frac{\ell}{N_D}.
\end{align*}
The second pairing has the same bound. Since
\[
 \max_{i_{{\rm D},1},i_{{\rm D},2}}
 |\mathsf w_{i_{{\rm D},1}i_{{\rm D},2}}^{\rm DB}|
 \le\frac{\norm K_\infty}{N_D-\ell},
\]
stationarity and \eqref{eq:difference-cumulant-sum} give
\begin{align*}
 &\sum_{i_{{\rm D},1},i_{{\rm D},2},i_{{\rm D},3},i_{{\rm D},4}}
 |\mathsf w_{i_{{\rm D},1}i_{{\rm D},2}}^{\rm DB}|
 |\mathsf w_{i_{{\rm D},3}i_{{\rm D},4}}^{\rm DB}|
 |\kappa_{\rm DB}(i_{{\rm D},1},i_{{\rm D},2},
                    i_{{\rm D},3},i_{{\rm D},4})|\\
 &\quad\le
 \frac{C}{(N_D-\ell)^2}
 \sum_{i_{{\rm D},1}=1}^{N_D}
 \sum_{i_{{\rm D},2},i_{{\rm D},3},i_{{\rm D},4}\in\mathbb Z}
 |\kappa_{\rm DB}(i_{{\rm D},1},i_{{\rm D},2},
                    i_{{\rm D},3},i_{{\rm D},4})|\\
 &\quad\le\frac{CN_D}{(N_D-\ell)^2}
 \le C\frac{\ell}{N_D}.
\end{align*}
Thus
\[
 \E\norm{\widehat\bOmega_n^{\rm DB}
 -\E\widehat\bOmega_n^{\rm DB}}_{\HS}^2
 \le C\frac{\ell}{N_D},
\]
which is \eqref{eq:DB-variance-rate}.

By Lemma~\ref{prop:edge-correction},
\[
 \E\widehat\bOmega_n^{\rm DB}=\bOmega_{{\rm DB},\ell}.
\]
Let
\[
 \mathfrak r_{n,\mathrm{st}}=\left(\frac{\ell}{N_D}\right)^{1/2},
 \qquad
 \mathfrak r_{n,\mathrm{bias}}=\ell^{-\beta_K}.
\]
Lemma~\ref{prop:population-DB-bias} gives
\[
 \norm{\bOmega_{{\rm DB},\ell}-\bOmega_{p,N}}_{\HS}\le C\mathfrak r_{n,\mathrm{bias}}.
\]
For $C_0\ge2C$,
\begin{align*}
 &\Pp\left\{
 \norm{\widehat\bOmega_n^{\rm DB}-\bOmega_{p,N}}_{\HS}
 >C_0(\mathfrak r_{n,\mathrm{st}}+\mathfrak r_{n,\mathrm{bias}})\right\}\\
 &\quad\le
 \frac{4\E\norm{\widehat\bOmega_n^{\rm DB}
 -\E\widehat\bOmega_n^{\rm DB}}_{\HS}^2}
 {C_0^2(\mathfrak r_{n,\mathrm{st}}+\mathfrak r_{n,\mathrm{bias}})^2}
 \le\frac{C}{C_0^2}.
\end{align*}
Therefore
\[
 \norm{\widehat\bOmega_n^{\rm DB}-\bOmega_{p,N}}_{\HS}
 =O_{\Pp}(\mathfrak r_{n,\mathrm{st}}+\mathfrak r_{n,\mathrm{bias}}),
\]
which proves \eqref{eq:DB-HS-rate}.  Since
$\bOmega_{p,N}\succeq0$,
\[
 \norm{\widehat\bOmega_n^+-\bOmega_{p,N}}_{\HS}
 \le
 \norm{\widehat\bOmega_n^{\rm DB}-\bOmega_{p,N}}_{\HS}
 =O_{\Pp}(\mathfrak r_n),
\]
by \eqref{eq:PSD-projection-audited}.  Both operators act on the same $p$-dimensional Fourier subspace, so their
difference has rank at most $p$ and
\[
 \norm{\widehat\bOmega_n^+-\bOmega_{p,N}}_1
 \le\sqrt p\,
 \norm{\widehat\bOmega_n^+-\bOmega_{p,N}}_{\HS}
 =O_{\Pp}(\sqrt p\,\mathfrak r_n).
\]
This proves \eqref{eq:positive-HS-rate} and
\eqref{eq:positive-trace-rate}.

\subsection{Proof of Lemma~\ref{prop:alternative-DB}}

The same algebra covers the two assertions, but their spacing assumptions are
kept separate.  In the single-change case, use the notational convention
$J_{\rm cp}=1$, $\tau_{1,n}=\tau_n$, $\delta_{1,n}=\delta_n$, and
$d_{1,n}=d_n$ only within this proof; no part of
Assumption~\ref{ass:multiple-configuration} is then invoked.  Let
\[
 \bm d_{b,n}=\iota_p^*d_{b,n}
 =\{Q_N(\delta_{b,n}\phi_j)\}_{j=1}^p
\]
and decompose, for $i_{\rm D}=1,\ldots,N_D$,
\begin{equation*}
 \bY_{i_{\rm D}}
 =\bm m_{i_{\rm D}}^{\rm DB}+\beps_{i_{\rm D}}^{\rm DB},
 \qquad
 \bm m_{i_{\rm D}}^{\rm DB}
 =\sum_{b=1}^{J_{\rm cp}}\sum_{q=0}^{m}c_q\bm d_{b,n}
 \ind(mh+i_{\rm D}-qh>\tau_{b,n}).
\end{equation*}
For the $b$th change, write $\bm m_{b,i_{\rm D}}^{\rm DB}$ for the
corresponding summand.  The zero-sum identity in
\eqref{eq:difference-constraints} gives
\begin{equation}
 \bm m_{b,i_{\rm D}}^{\rm DB}\ne0
 \quad\Longrightarrow\quad
 \tau_{b,n}-mh<i_{\rm D}\le\tau_{b,n}.
 \label{eq:multiple-DB-support}
\end{equation}
For one change, \eqref{eq:multiple-DB-support} alone shows that the support
has at most $mh$ indices.  For several changes,
Assumption~\ref{ass:multiple-configuration} and $mh=o(\Delta_n)$ imply that
the support sets are disjoint for all sufficiently large $n$.  Thus, in the
respective single- and multiple-change cases,
\begin{align}
 \max_{1\le i_{\rm D}\le N_D}\norm{\bm m_{i_{\rm D}}^{\rm DB}}
 &\le C_c\max_b\norm{d_{b,n}},
 \notag\\
 \sum_{i_{\rm D}=1}^{N_D}\norm{\bm m_{i_{\rm D}}^{\rm DB}}^2
 &=\sum_{b=1}^{J_{\rm cp}}\sum_{i_{\rm D}=1}^{N_D}
   \norm{\bm m_{b,i_{\rm D}}^{\rm DB}}^2
 \le Cmh\sum_{b=1}^{J_{\rm cp}}\norm{d_{b,n}}^2.
 \label{eq:multiple-DB-mean-energy}
\end{align}
For $J_{\rm cp}=1$, the equality in the second line is tautological and does
not require a spacing condition.

Expand each lag product in \eqref{eq:edge-autocov}.  Relative to the
noise-only estimator, the additional terms are a mean-only contribution
$\mathcal M_{{\rm cp},n}$ and two mean--noise contributions.  By
Cauchy--Schwarz and \eqref{eq:multiple-DB-mean-energy},
\begin{align}
 \norm{\mathcal M_{{\rm cp},n}}_{\HS}
 &\le C\sum_{|r|<\ell}\frac1{N_D-|r|}
 \sum_{i_{\rm D}=|r|+1}^{N_D}
 \norm{\bm m_{i_{\rm D}}^{\rm DB}}
 \norm{\bm m_{i_{\rm D}-|r|}^{\rm DB}}
 \notag\\
 &\le C\frac{\ell}{N_D}
 \sum_{i_{\rm D}=1}^{N_D}\norm{\bm m_{i_{\rm D}}^{\rm DB}}^2
 \notag\\
 &\le C\frac{\ell mh}{N_D}
 \sum_{b=1}^{J_{\rm cp}}\norm{d_{b,n}}^2.
 \label{eq:multiple-DB-mean-only-bound}
\end{align}

For one orientation of the mean--noise term, define, for $0\le r<\ell$,
\[
 \mathcal X_r
 =\frac1{N_D-r}\sum_{i_{\rm D}=r+1}^{N_D}
 \bm m_{i_{\rm D}}^{\rm DB}\otimes
 \beps_{i_{\rm D}-r}^{\rm DB}.
\]
Let $\mathfrak m_{i_{\rm D}}=\norm{\bm m_{i_{\rm D}}^{\rm DB}}$ and let
$\mathsf C_{\rm DB}
=(\mathfrak b_{i_{\rm D}-j_{\rm D}})_{i_{\rm D},j_{\rm D}=1}^{N_D}$,
where $\mathfrak b_u=\norm{\bGamma_{\rm DB}(u)}_1$.
The uniform summability in
\eqref{eq:filtered-cov-summability-audited} implies
$\norm{\mathsf C_{\rm DB}}_{\op}\le C$.  Hence, using
$|\tr\{\bGamma_{\rm DB}(u)\}|\le
\norm{\bGamma_{\rm DB}(u)}_1=\mathfrak b_u$,
\begin{align*}
 \E\norm{\mathcal X_r}_{\HS}^2
 &=\frac1{(N_D-r)^2}
 \sum_{i_{\rm D},j_{\rm D}=r+1}^{N_D}
 \ip{\bm m_{i_{\rm D}}^{\rm DB}}{\bm m_{j_{\rm D}}^{\rm DB}}
 \tr\{\bGamma_{\rm DB}(j_{\rm D}-i_{\rm D})\}\\
 &\le\frac1{(N_D-r)^2}
 \sum_{i_{\rm D},j_{\rm D}=r+1}^{N_D}
 \mathfrak m_{i_{\rm D}}\mathfrak m_{j_{\rm D}}
 \mathfrak b_{j_{\rm D}-i_{\rm D}}\\
 &\le\frac{\norm{\mathsf C_{\rm DB}}_{\op}}{(N_D-r)^2}
 \sum_{i_{\rm D}=1}^{N_D}\mathfrak m_{i_{\rm D}}^2\\
 &\le C\frac{mh}{N_D^2}
 \sum_{b=1}^{J_{\rm cp}}\norm{d_{b,n}}^2.
\end{align*}
The adjoint and shifted orientations obey the same bound.  Minkowski's
inequality over the $2\ell-1$ signed lags therefore gives
\begin{equation}
 \norm{\mathcal X_{{\rm cp},n}}_{L^2(\HS)}
 \le C\frac{\ell\sqrt{mh}}{N_D}
 \left\{\sum_{b=1}^{J_{\rm cp}}\norm{d_{b,n}}^2\right\}^{1/2}.
 \label{eq:multiple-DB-mean-noise-bound}
\end{equation}
Combining \eqref{eq:multiple-DB-mean-only-bound},
\eqref{eq:multiple-DB-mean-noise-bound}, and the noise-only rate
\eqref{eq:DB-HS-rate} proves \eqref{eq:multiple-DB-HS-rate}.  The
single-change expression is exactly $\mathfrak r_{\delta,n}$.  Finally,
\eqref{eq:PSD-projection-audited} and the rank-$p$ inequality yield
\eqref{eq:alternative-positive-trace-rate} and
\eqref{eq:multiple-positive-trace-rate}.

If $m$ is fixed, $h\asymp\ell$, and the jumps are uniformly bounded, both
contamination rates are $O(\ell^2/n+\ell^{3/2}/n)$ for fixed
$J_{\rm cp}$, and hence vanish when $\ell^2/n\to0$.

\section{Proofs for spectral weights and the null limit}

\subsection{Proof of Lemma~\ref{prop:weight-consistency}}

Let
\[
 \lambda_{1,p,N}\ge\cdots\ge\lambda_{p,p,N}\ge0
\]
be the eigenvalues of $\bOmega_{p,N}$, extended by zeros for $j>p$.
When $\Omega$ has finite rank, extend its positive-eigenvalue list by zeros
as well for the eigenvalue comparison in this proof.
The Hoffman--Wielandt inequality
\citep[Chapter~VI]{Bhatia1997} gives
\begin{align*}
 \sum_{j=1}^{p}
 (\widehat\lambda_{j,n}-\lambda_{j,p,N})^2
 &\le
 \norm{\widehat\bOmega_n^+-\bOmega_{p,N}}_{\HS}^2\\
 &=O_{\Pp}(\mathfrak r_n^2).
\end{align*}
Hence
\begin{equation}
 \sum_{j=1}^{p}
 |\widehat\lambda_{j,n}-\lambda_{j,p,N}|
 \le\sqrt p\,
 \norm{\widehat\bOmega_n^+-\bOmega_{p,N}}_{\HS}
 =O_{\Pp}(\sqrt p\,\mathfrak r_n).
 \label{eq:eigen-l1-audited}
\end{equation}
The Lidskii--Mirsky--Wielandt inequality for self-adjoint trace-class
operators \citep[Corollary~III.4.2]{Bhatia1997}
\citep[Chapter~3]{Simon2005} yields
\begin{equation}
 \sum_{j\ge1}|\lambda_{j,p,N}-\lambda_j|
 \le\norm{\Omega_{p,N}^{\cH}-\Omega}_1
 \le C\zeta_{p,N}.
 \label{eq:population-eigen-tail-audited}
\end{equation}
Therefore
\begin{equation*}
 \sum_{j\ge1}|\widehat\lambda_{j,n}-\lambda_j|
 =O_{\Pp}(\mathfrak e_n),
 \qquad
 \mathfrak e_n=\sqrt p\,\mathfrak r_n+\zeta_{p,N}.
\end{equation*}
Nonnegativity implies
\[
 |\widehat{\mathfrak t}_n-\mathfrak t_\Omega|
 =\left|\sum_{j\ge1}
   (\widehat\lambda_{j,n}-\lambda_j)\right|
 \le\sum_{j\ge1}|\widehat\lambda_{j,n}-\lambda_j|
 =O_{\Pp}(\mathfrak e_n).
\]

On
\[
 \mathcal E_{n,\mathrm{tr}}
 =\{|\widehat{\mathfrak t}_n-\mathfrak t_\Omega|\le\mathfrak t_\Omega/2\},
 \qquad
 \Pp(\mathcal E_{n,\mathrm{tr}})\to1,
\]
set $f_a(x,s)=x/(x+as)$.  Uniformly in $a\in\cG$,
\begin{align*}
 |f_a(\widehat\lambda,\widehat{\mathfrak t}_n)
  -f_a(\lambda,\widehat{\mathfrak t}_n)|
 &=
 \frac{a\widehat{\mathfrak t}_n|\widehat\lambda-\lambda|}
 {(\widehat\lambda+a\widehat{\mathfrak t}_n)
  (\lambda+a\widehat{\mathfrak t}_n)}
 \le\frac{2}{a_-\mathfrak t_\Omega}|\widehat\lambda-\lambda|,\\
 |f_a(\lambda,\widehat{\mathfrak t}_n)-f_a(\lambda,\mathfrak t_\Omega)|
 &=
 \frac{a\lambda|\widehat{\mathfrak t}_n-\mathfrak t_\Omega|}
 {(\lambda+a\widehat{\mathfrak t}_n)(\lambda+a\mathfrak t_\Omega)}
 \le C\lambda|\widehat{\mathfrak t}_n-\mathfrak t_\Omega|.
\end{align*}
Summation gives
\begin{align*}
 \sum_{j\ge1}
 |\widehat w_{j,n}(a)-w_j(a)|
 &\le
 C\sum_{j\ge1}
 |\widehat\lambda_{j,n}-\lambda_j|
 +C|\widehat{\mathfrak t}_n-\mathfrak t_\Omega|\sum_{j\ge1}\lambda_j\\
 &=O_{\Pp}(\mathfrak e_n),
\end{align*}
which proves \eqref{eq:weight-l1-rate}.  Since
$0\le \widehat w_{j,n}(a),w_j(a)\le1$,
\begin{align*}
 |\widehat A_{1,n}(a)-A_1(a)|
 &\le\sum_{j\ge1}|\widehat w_{j,n}(a)-w_j(a)|,\\
 |\widehat A_{2,n}(a)-A_2(a)|
 &\le\sum_{j\ge1}|\widehat w_{j,n}^2(a)-w_j^2(a)|\\
 &\le2\sum_{j\ge1}|\widehat w_{j,n}(a)-w_j(a)|.
\end{align*}
Taking the maximum over the finite ridge grid proves
\eqref{eq:Ar-consistency}.

Under one change, Lemma~\ref{prop:alternative-DB} replaces
\eqref{eq:eigen-l1-audited} by
\[
 \sum_{j=1}^{p}
 |\widehat\lambda_{j,n}-\lambda_{j,p,N}|
 =O_{\Pp}\!\left[\sqrt p\{\mathfrak r_n+\mathfrak r_{\delta,n}\}\right].
\]
Under model~\eqref{eq:multiple-change-model}, the same inequality holds
with $\mathfrak r_{\delta,n}$ replaced by
$\mathfrak r_{\delta,{\rm cp},n}$.  Combining the corresponding bound with
\eqref{eq:population-eigen-tail-audited} and repeating the preceding
Lipschitz calculation proves both alternative statements.

\subsection{Proof of Lemma~\ref{lem:bridge-series-regularity}}

Let
\[
 \mathcal Z_j(t)=\frac{B_j^2(t)}{t(1-t)}-1,
 \qquad
 t\in[\epsilon,1-\epsilon],
\]
and put
\[
 \mathfrak u_{\epsilon}^{\rm br}
 =\E\sup_{t\in[\epsilon,1-\epsilon]}|\mathcal Z_1(t)|.
\]
Since $B_j(t)=W_j(t)-tW_j(1)$ and
$t(1-t)\ge\epsilon(1-\epsilon)$,
\begin{align*}
 \mathfrak u_{\epsilon}^{\rm br}
 &\le
 1+\{\epsilon(1-\epsilon)\}^{-1}
 \E\sup_{0\le t\le1}B_j(t)^2\\
 &<\infty
\end{align*}
by Doob's $L^2$ inequality.  Moreover,
\[
 \sum_{j\ge1}w_j(a)
 =\sum_{j\ge1}\frac{\lambda_j}{\lambda_j+a\mathfrak t_\Omega}
 \le\frac1{a\mathfrak t_\Omega}\sum_{j\ge1}\lambda_j
 =\frac1a.
\]
Therefore
\begin{align*}
 \E\sum_{j\ge1}w_j(a)
 \sup_{t\in[\epsilon,1-\epsilon]}|\mathcal Z_j(t)|
 &=
 \sum_{j\ge1}w_j(a)
 \mathfrak u_{\epsilon}^{\rm br}\\
 &<\infty.
\end{align*}
Hence
\[
 \sum_{j\ge1}w_j(a)\sup_t|\mathcal Z_j(t)|<\infty
 \quad\text{a.s.}
\]
The Weierstrass criterion gives almost-sure uniform convergence, and
\[
 \E\sup_t
 \left|\sum_{j>J_{\rm br}}w_j(a)\mathcal Z_j(t)\right|
 \le \mathfrak u_{\epsilon}^{\rm br}\sum_{j>J_{\rm br}}w_j(a)\longrightarrow0
\]
gives uniform $L^1$ convergence.  The limit is continuous because every
partial sum is continuous.

\subsection{Proof of Theorem~\ref{thm:null-limit}}

Let $\mathcal J_n=\mathcal J_{p,N}$.  Since
\[
 \mathcal J_ne_i=\iota_p\beps_{i,p,N},
\]
the embedded score CUSUM equals
\begin{equation*}
 \iota_p\widehat\bS_n(t)
 =
 \frac1{\sqrt n}\left\{
 \sum_{i=1}^{\lfloor nt\rfloor}\mathcal J_ne_i
 -\frac{\lfloor nt\rfloor}{n}\sum_{i=1}^{n}\mathcal J_ne_i
 \right\}.
\end{equation*}
Lemmas~\ref{lem:projected-dependence} and
\ref{lem:ARS-main-FCLT} give
\begin{align*}
 &\sup_{0\le t\le1}
 \left\|
 n^{-1/2}\sum_{i=1}^{\lfloor nt\rfloor}(\mathcal J_ne_i-e_i)
 \right\|=o_{\Pp}(1),
 \notag\\
 &n^{-1/2}\left\{
 \sum_{i=1}^{\lfloor nt\rfloor}e_i
 -\frac{\lfloor nt\rfloor}{n}\sum_{i=1}^{n}e_i
 \right\}
 \dto\mathbb B_\Omega(t).
 \notag
\end{align*}
The uniform norm of the tied-down numerical-projection error is at most
twice the first supremum in the display.  Consequently,
\begin{equation}
 \iota_p\widehat\bS_n
 \dto\mathbb B_\Omega
 \quad\text{in }D([0,1],\cH).
 \label{eq:functional-bridge-FCLT-audited}
\end{equation}
Because the limit has continuous paths, this convergence also holds under
the uniform metric on every compact subinterval of $(0,1)$.  In particular,
\[
 \sup_{t\in[\epsilon,1-\epsilon]}
 \norm{\widehat\bS_n(t)}=O_{\Pp}(1).
\]

By Lemmas~\ref{lem:positive-part-rate} and
\ref{prop:quadrature},
\begin{align*}
 \norm{\widehat\Omega_{n,+}^{\cH}-\Omega}_1
 &\le
 \norm{\widehat\Omega_{n,+}^{\cH}-\Omega_{p,N}^{\cH}}_1
 +\norm{\Omega_{p,N}^{\cH}-\Omega}_1\\
 &=O_{\Pp}(\mathfrak e_n).
\end{align*}
Hence
\[
 |\widehat{\mathfrak t}_n-\mathfrak t_\Omega|
 \le\norm{\widehat\Omega_{n,+}^{\cH}-\Omega}_1
 =O_{\Pp}(\mathfrak e_n).
\]
Lemma~\ref{lem:resolvent-replacement} therefore yields
\begin{equation*}
 \max_{a\in\cG}
 \norm{\widehat{\mathcal R}_{n,a}-\mathcal R_a}_{\op}
 =O_{\Pp}(\mathfrak e_n),
 \qquad
 \mathcal R_a=(\Omega+a\mathfrak t_\Omega \Id_{\cH})^{-1}.
\end{equation*}
Uniformly in $a\in\cG$ and
$t\in[\epsilon,1-\epsilon]$,
\begin{align*}
 &\left|
 \ip{\widehat\bS_n(t)}
 {\widehat{\mathcal R}_{n,a}\widehat\bS_n(t)}
 -
 \ip{\widehat\bS_n(t)}
 {\mathcal R_a\widehat\bS_n(t)}
 \right|\\
 &\quad\le
 \norm{\widehat{\mathcal R}_{n,a}-\mathcal R_a}_{\op}
 \norm{\widehat\bS_n(t)}^2
 =O_{\Pp}(\mathfrak e_n).
\end{align*}

Put $t_n=\lfloor nt\rfloor/n$.  Since
\[
 \sup_{t\in[\epsilon,1-\epsilon]}
 |t_n(1-t_n)-t(1-t)|\le Cn^{-1},
\]
\eqref{eq:functional-bridge-FCLT-audited}, the continuous mapping theorem,
and the preceding resolvent bound imply
\begin{equation}
 \left\{V_{n,a}(\lfloor nt\rfloor):
 a\in\cG,\ t\in[\epsilon,1-\epsilon]\right\}
 \dto
 \left\{
 \frac{\ip{\mathbb B_\Omega(t)}
 {\mathcal R_a\mathbb B_\Omega(t)}}{t(1-t)}
 \right\}_{a,t}.
 \label{eq:oracle-quadratic-limit-audited}
\end{equation}

Let
\[
 \Omega=\sum_{j\ge1}\lambda_j\psi_j\otimes\psi_j.
\]
On a common probability space,
\[
 \mathbb B_\Omega(t)
 =\sum_{j\ge1}\sqrt{\lambda_j}B_j(t)\psi_j.
\]
Indeed,
\begin{align*}
 \E\sup_{t\in[\epsilon,1-\epsilon]}
 \left\|\sum_{j>J_{\rm eig}}\sqrt{\lambda_j}B_j(t)\psi_j\right\|^2
 &\le
 \E\sup_{0\le t\le1}B_1(t)^2
 \sum_{j>J_{\rm eig}}\lambda_j\\
 &\longrightarrow0.
\end{align*}
Since
\[
 \mathcal R_a\psi_j=(\lambda_j+a\mathfrak t_\Omega)^{-1}\psi_j,
\]
the limit in \eqref{eq:oracle-quadratic-limit-audited} equals
\begin{align*}
 \frac{\ip{\mathbb B_\Omega(t)}
 {\mathcal R_a\mathbb B_\Omega(t)}}{t(1-t)}
 &=
 \sum_{j\ge1}
 \frac{\lambda_j}{\lambda_j+a\mathfrak t_\Omega}
 \frac{B_j(t)^2}{t(1-t)}\\
 &=
 A_1(a)+\{2A_2(a)\}^{1/2}\mathcal D_a(t).
\end{align*}
Lemma~\ref{lem:bridge-series-regularity} justifies the uniform series
operations.  Lemma~\ref{prop:weight-consistency} gives
\[
 \max_{a\in\cG}
 \left\{
 |\widehat A_{1,n}(a)-A_1(a)|
 +|\widehat A_{2,n}(a)-A_2(a)|
 \right\}=o_{\Pp}(1).
\]
Slutsky's theorem proves \eqref{eq:joint-null-limit}.  Because the limit has
continuous paths, the sequence is asymptotically uniformly equicontinuous on
the trimmed interval.  The grid $\{k/n:k\in\cK_{n,\epsilon}\}$ has mesh at
most $1/n$ and its endpoints are within $1/n$ of
$[\epsilon,1-\epsilon]$.  Hence the difference between the supremum of the
step process in \eqref{eq:joint-null-limit} and the maximum over
$\cK_{n,\epsilon}$ is $o_{\Pp}(1)$.  The continuous mapping theorem therefore
proves \eqref{eq:joint-scan-limit}.

Finally, for centered jointly Gaussian $Z_1,Z_2$,
\[
 \Cov(Z_1^2,Z_2^2)=2\Cov(Z_1,Z_2)^2.
\]
Since
\[
 \Cov\{B_j(t),B_j(t')\}=(t\wedge t')-tt'
\]
and the bridges are independent across $j$,
\begin{align*}
 \Cov\{\mathcal D_a(t),\mathcal D_{a'}(t')\}
 &=
 \frac{\{(t\wedge t')-tt'\}^2}{t(1-t)t'(1-t')}
 \frac{\sum_{j\ge1}w_j(a)w_j(a')}
 {\{A_2(a)A_2(a')\}^{1/2}},
\end{align*}
which is \eqref{eq:cross-ridge-covariance}.

\subsection{Proof of Lemma~\ref{lem:finite-reference-regularity}}

Condition on the data and on $\mathcal E_n$.  Let
\[
 N_{n,\epsilon}^{\rm scan}=|\cK_{n,\epsilon}|.
\]
A nonsingular Gaussian factorization of the standardized bridge vector gives
\[
 \bm U=\mathsf L_n\bm Z^{\rm ref},
 \qquad
 \bm Z^{\rm ref}\sim N(\bm0,\bI_{pN_{n,\epsilon}^{\rm scan}}),
\]
where $\mathsf L_n$ is block diagonal with nonsingular lower triangular
blocks.  For $a\in\cG$ and $k\in\cK_{n,\epsilon}$, define
\[
 \mathcal Q_{n,a,k}^{\rm ref}(z)
 =\sum_{j=1}^{p}\widehat w_{j,n}(a)
 U_{j,k}(\mathsf L_nz)^2,
 \qquad
 \mathcal Q_{n,a}^{\rm ref}(z)=\max_{k}\mathcal Q_{n,a,k}^{\rm ref}(z).
\]
Then
\[
 \widehat{\mathcal T}_{n,a}^{*}
 =
 \frac{\mathcal Q_{n,a}^{\rm ref}(\bm Z^{\rm ref})-\widehat A_{1,n}(a)}
 {\{2\widehat A_{2,n}(a)\}^{1/2}}.
\]
On $\mathcal E_n$,
\[
 \widehat A_{2,n}(a)>0
\]
and every $\mathcal Q_{n,a,k}^{\rm ref}$ is a nonzero quadratic polynomial:
there is an index with positive ridge weight, and its Gaussian coordinate at
each scan point is nondegenerate.  For $x\in\mathbb R$, with
\[
 c_x^{\rm ref}=\widehat A_{1,n}(a)
 +x\{2\widehat A_{2,n}(a)\}^{1/2},
\]
\begin{align*}
 \{\mathcal Q_{n,a}^{\rm ref}(\bm Z^{\rm ref})=c_x^{\rm ref}\}
 &\subseteq
 \bigcup_{k\in\cK_{n,\epsilon}}
 \{\mathcal Q_{n,a,k}^{\rm ref}(\bm Z^{\rm ref})=c_x^{\rm ref}\}.
\end{align*}
Each set on the right is the zero set of a nonzero polynomial and therefore
has Lebesgue measure zero.  Since $\bm Z^{\rm ref}$ has a Lebesgue density,
\[
 \Pp_{\bm B}\{\widehat{\mathcal T}_{n,a}^{*}=x
 \mid\widehat{\mathcal W}_n\}=0.
\]
Thus $\widehat F_{n,a}$ is continuous.

The map $\mathcal Q_{n,a}^{\rm ref}$ is continuous, nonnegative, and homogeneous:
\[
 \mathcal Q_{n,a}^{\rm ref}(xz)=x^2\mathcal Q_{n,a}^{\rm ref}(z),
 \qquad x\ge0.
\]
Choose $z_0$ with $\mathcal Q_{n,a}^{\rm ref}(z_0)>0$.  For
$0\le y_1<y_2<\infty$, choose $x_0>0$ such that
\[
 y_1<\mathcal Q_{n,a}^{\rm ref}(x_0z_0)<y_2.
\]
By continuity, an open ball around $x_0z_0$ is contained in
$\{y_1<\mathcal Q_{n,a}^{\rm ref}<y_2\}$.  The driving vector $\bm Z^{\rm ref}$ is nondegenerate standard Gaussian
and therefore has a strictly positive Lebesgue density on its full Euclidean
state space; the nonsingularity of $\mathsf L_n$ ensures that no reference
coordinate has been collapsed.  Hence this ball has positive Gaussian
probability.  Hence $\widehat F_{n,a}$ is strictly increasing on its
support.  The conditional
probability integral transform gives
\[
 1-\widehat F_{n,a}(\widehat{\mathcal T}_{n,a}^{\circ})
 \mid\widehat{\mathcal W}_n
 \sim\operatorname{Uniform}(0,1).
\]

For the Cauchy statistic, write
\[
 \bm Z^{\rm ref}=R_{\rm rad}\Theta_{\rm dir},
\]
where $R_{\rm rad}>0$ has a continuous density and is independent of the uniform
direction $\Theta_{\rm dir}$.  On $\mathcal E_n$, let
\[
 \mathbb I_n^{\rm eig,+}=\{j:\widehat\lambda_{j,n}>0\}.
\]
This set is nonempty, and, for every $a\in\cG$,
\[
 \widehat w_{j,n}(a)>0
 \quad\Longleftrightarrow\quad j\in\mathbb I_n^{\rm eig,+}.
\]
Define the common nullspace
\[
 \mathcal N_n
 =\left\{z:
 U_{j,k}(\mathsf L_nz)=0
 \text{ for all }j\in\mathbb I_n^{\rm eig,+},
 k\in\cK_{n,\epsilon}\right\}.
\]
Because $\mathsf L_n$ is nonsingular and $\mathbb I_n^{\rm eig,+}\ne\varnothing$,
$\mathcal N_n$ is a proper linear subspace and
$\{\Theta_{\rm dir}\in\mathcal N_n\}$ has spherical measure zero.  If
$\Theta_{\rm dir}\notin\mathcal N_n$, then at least one active coordinate is
nonzero at at least one scan point.  Since every ridge value assigns a strictly
positive weight to every coordinate in $\mathbb I_n^{\rm eig,+}$,
\[
 \mathcal Q_{n,a}^{\rm ref}(\Theta_{\rm dir})>0,
 \qquad a\in\cG.
\]
Consequently, for every $a\in\cG$,
\begin{align*}
 \widehat{\mathcal T}_{n,a}^{*}(x\Theta_{\rm dir})
 &=
 \frac{x^2\mathcal Q_{n,a}^{\rm ref}(\Theta_{\rm dir})-\widehat A_{1,n}(a)}
 {\{2\widehat A_{2,n}(a)\}^{1/2}}
\end{align*}
is strictly increasing in $x>0$.  Since $\widehat F_{n,a}$ is continuous and
strictly increasing on its support,
\[
 x\longmapsto
 \mathfrak c_{\rm C}\!\left[
 1-\widehat F_{n,a}
 \{\widehat{\mathcal T}_{n,a}^{*}(x\Theta_{\rm dir})\}
 \right]
\]
is strictly increasing.  Positive ridge-combination weights preserve strict
monotonicity, so, conditionally on
$\Theta_{\rm dir}\notin\mathcal N_n$, the inverse image of a singleton under
$\mathfrak C_n^{\circ}$ contains at most one value of $R_{\rm rad}$.  Therefore
\[
 \Pp_{\bm B}(\mathfrak C_n^{\circ}=x
 \mid\widehat{\mathcal W}_n,\Theta_{\rm dir})=0,
 \qquad x\in\mathbb R,
\]
and integration over $\Theta_{\rm dir}$ proves continuity of the conditional law.

\subsection{Proofs of Theorem~\ref{thm:critical-values} and Lemma~\ref{lem:Monte-Carlo-reference}}

For $a\in\cG$, define
\[
 \mathbb X_a=C([\epsilon,1-\epsilon],\ell^2),
 \qquad
 \mathbb Z_a(t)=\{w_j(a)^{1/2}U_j(t):j\ge1\},
\]
with
\[
 \norm z_{\mathbb X_a}
 =\sup_{t\in[\epsilon,1-\epsilon]}\norm{z(t)}_{\ell^2}.
\]
For $J_{\rm eig}\ge1$, let
\[
 \mathbb Z_a^{(J_{\rm eig})}(t)
 =\{w_j(a)^{1/2}U_j(t)\ind(j\le J_{\rm eig}):j\ge1\}.
\]
Since $\sum_jw_j(a)<\infty$,
\begin{align*}
 \E\norm{\mathbb Z_a-\mathbb Z_a^{(J_{\rm eig})}}_{\mathbb X_a}^2
 &\le
 \sum_{j>J_{\rm eig}}w_j(a)\E\sup_tU_j(t)^2\\
 &=C_\epsilon\sum_{j>J_{\rm eig}}w_j(a)
 \longrightarrow0.
\end{align*}
Thus $\mathbb Z_a$ is a nonzero centered Gaussian element of the separable Banach
space $\mathbb X_a$, and
\begin{equation*}
 \mathcal T_a
 =\frac{\norm{\mathbb Z_a}_{\mathbb X_a}^2-A_1(a)}
 {\{2A_2(a)\}^{1/2}}.
\end{equation*}
For every $c_{\rm rad}>0$, the norm of a Gaussian measure has a density on
$[c_{\rm rad},\infty)$ \citep[pp.~277--278]{Rhee1984}.  Moreover, choosing
$j_0$ with $w_{j_0}(a)>0$ and
$t_0\in(\epsilon,1-\epsilon)$ gives
\[
 \{\norm{\mathbb Z_a}_{\mathbb X_a}=0\}
 \subseteq
 \{w_{j_0}(a)^{1/2}U_{j_0}(t_0)=0\},
\]
whose probability is zero.  The support of the Gaussian law is the closure
of its reproducing-kernel Hilbert space
\citep[Theorem~3.6.1]{Bogachev1998}
\citep[Section~5]{VanDerVaartVanZanten2008}, hence a nonzero closed linear
subspace.  For $0\le x_1<x_2$, choose a support point $z_{\rm sup}\ne0$ and
$c_{\rm sup}>0$ such that
\[
 \norm{c_{\rm sup}z_{\rm sup}}_{\mathbb X_a}=\frac{x_1+x_2}{2}.
\]
A sufficiently small open ball around $c_{\rm sup}z_{\rm sup}$ lies in
$\{z:x_1<\norm z_{\mathbb X_a}<x_2\}$ and has positive Gaussian probability.
Thus the distribution function of $\norm{\mathbb Z_a}_{\mathbb X_a}$ is
continuous and strictly increasing on $[0,\infty)$.  Since
\[
 \mathcal T_a
 =\frac{\norm{\mathbb Z_a}_{\mathbb X_a}^2-A_1(a)}
       {\{2A_2(a)\}^{1/2}},
\]
its distribution function $F_a$ is continuous and strictly increasing on its
support
\[
 \left[-\frac{A_1(a)}{\{2A_2(a)\}^{1/2}},\infty\right).
\]

Let
\[
 \mathcal Z_j(t)=\frac{B_j^2(t)}{t(1-t)}-1.
\]
For a nonnegative summable sequence $w=(w_j)$ with
$S_w=\sum_jw_j^2>0$, define
\begin{align*}
 \Phi(w,\bm B)
 &=\sup_{t\in[\epsilon,1-\epsilon]}
 \frac{\sum_jw_j\mathcal Z_j(t)}{\sqrt{2S_w}},\\
 \Phi_n(w,\bm B)
 &=\max_{k\in\cK_{n,\epsilon}}
 \frac{\sum_jw_j\mathcal Z_j(t_k)}{\sqrt{2S_w}}.
\end{align*}
Doob's inequality gives
\[
 \mathfrak u_{\epsilon}^{\rm br}
 :=\E\sup_{t\in[\epsilon,1-\epsilon]}|\mathcal Z_1(t)|
 <\infty.
\]
Consequently,
\begin{equation*}
 \E_{\bm B}\sup_t
 \left|\sum_j(w_j-\widetilde w_j)\mathcal Z_j(t)\right|
 \le \mathfrak u_{\epsilon}^{\rm br}\sum_j|w_j-\widetilde w_j|.
\end{equation*}
If
\[
 S_w\wedge S_{\widetilde w}\ge c_{\rm wt}>0,
 \qquad
 \norm{w}_1+\norm{\widetilde w}_1\le C,
\]
then
\begin{align*}
 &\sup_t
 \left|
 \frac{\sum_jw_j\mathcal Z_j(t)}{\sqrt{2S_w}}
 -\frac{\sum_j\widetilde w_j\mathcal Z_j(t)}
       {\sqrt{2S_{\widetilde w}}}
 \right|\\
 &\quad\le
 \frac{\sup_t|\sum_j(w_j-\widetilde w_j)\mathcal Z_j(t)|}{\sqrt{2c_{\rm wt}}}\\
 &\qquad+
 \frac{\sup_t|\sum_j\widetilde w_j\mathcal Z_j(t)|
 |S_w-S_{\widetilde w}|}
 {\sqrt2\,\sqrt{S_wS_{\widetilde w}}
  (\sqrt{S_w}+\sqrt{S_{\widetilde w}})}.
\end{align*}
Since $0\le w_j,\widetilde w_j\le1$,
\[
 |S_w-S_{\widetilde w}|
 \le2\sum_j|w_j-\widetilde w_j|.
\]
Taking bridge expectations in the preceding inequality yields
\begin{equation*}
 \E_{\bm B}\sup_t
 \left|
 \frac{\sum_jw_j\mathcal Z_j(t)}{\{2\sum_jw_j^2\}^{1/2}}
 -\frac{\sum_j\widetilde w_j\mathcal Z_j(t)}
       {\{2\sum_j\widetilde w_j^2\}^{1/2}}
 \right|
 \le C\sum_j|w_j-\widetilde w_j|.
\end{equation*}
Write $\bm w(a)=\{w_j(a):j\ge1\}$.  Lemma~\ref{prop:weight-consistency} therefore gives, uniformly over
$a\in\cG$,
\[
 \Phi\{\widehat{\bm w}_n(a),\bm B\}
 -\Phi\{\bm w(a),\bm B\}
 \longrightarrow0
\]
in conditional probability.

For the time-grid error, let
\[
 \operatorname{osc}_{\mathcal Z}(c_{\rm mod})
 =\sup_{\substack{s,t\in[\epsilon,1-\epsilon]\\|s-t|\le c_{\rm mod}}}
 |\mathcal Z_1(s)-\mathcal Z_1(t)|.
\]
Almost-sure continuity and
$\operatorname{osc}_{\mathcal Z}(c_{\rm mod})\le2\sup_t|\mathcal Z_1(t)|$ imply
\[
 \E\operatorname{osc}_{\mathcal Z}(c_{\rm mod})\longrightarrow0.
\]
On
\[
 \min_{a\in\cG}\widehat A_{2,n}(a)
 \ge\frac12\min_{a\in\cG}A_2(a),
\]
whose probability tends to one,
\begin{align*}
 &\max_{a\in\cG}
 \E_{\bm B}\left|
 \Phi_n\{\widehat{\bm w}_n(a),\bm B\}
 -\Phi\{\widehat{\bm w}_n(a),\bm B\}\right|\\
 &\quad\le
 C\max_{a\in\cG}
 \frac{\widehat A_{1,n}(a)}
      {\widehat A_{2,n}(a)^{1/2}}
 \E\operatorname{osc}_{\mathcal Z}(1/n)
 \longrightarrow0.
\end{align*}
Thus
\[
 d_{\mathrm{BL}}(\widehat F_n^{\rm joint},F^{\rm joint})
 \pto0.
\]
The marginal version and the continuity of $F_a$ imply, by P\'olya's
theorem \citep[Section~2]{Billingsley1999},
\begin{equation*}
 \max_{a\in\cG}\sup_x
 |\widehat F_{n,a}(x)-F_a(x)|
 \pto0.
\end{equation*}

Let $c_{a,1-\alpha}=F_a^{-1}(1-\alpha)$.  For every $c_{\rm fr}>0$,
strict increase gives
\[
 F_a(c_{a,1-\alpha}-c_{\rm fr})
 <1-\alpha
 <F_a(c_{a,1-\alpha}+c_{\rm fr}).
\]
On the event
\[
 \sup_x|\widehat F_{n,a}(x)-F_a(x)|
 <
 \frac12\min\left\{
 1-\alpha-F_a(c_{a,1-\alpha}-c_{\rm fr}),\
 F_a(c_{a,1-\alpha}+c_{\rm fr})-(1-\alpha)
 \right\},
\]
the generalized inverse satisfies
\[
 |\widehat c_{n,a,1-\alpha}-c_{a,1-\alpha}|<c_{\rm fr}.
\]
Hence
\[
 \widehat c_{n,a,1-\alpha}\pto c_{a,1-\alpha},
 \qquad
 \Pp\{T_{n,a}>\widehat c_{n,a,1-\alpha}\}\to\alpha,
\]
by Theorem~\ref{thm:null-limit}.

For the joint Cauchy statistic, put
\[
 \bm T_n=(T_{n,a_1},\ldots,T_{n,a_{L_{\cG}}}),
 \qquad
 \bm T_\infty=(\mathcal T_{a_1},\ldots,\mathcal T_{a_{L_{\cG}}}),
\]
and
\[
 \bm P_\infty
 =\{1-F_{a_v}(\mathcal T_{a_v}):1\le v\le L_{\cG}\}.
\]
Equation \eqref{eq:conditional-cdf-consistency} gives
\begin{equation*}
 \max_{1\le v\le L_{\cG}}
 \left|P_{n,v}-\{1-F_{a_v}(T_{n,a_v})\}\right|
 \le
 \max_{1\le v\le L_{\cG}}\sup_x
 |\widehat F_{n,a_v}(x)-F_{a_v}(x)|
 \pto0.
\end{equation*}
Theorem~\ref{thm:null-limit} and the continuous mapping theorem imply
\begin{equation}
 (P_{n,1},\ldots,P_{n,L_{\cG}})\dto\bm P_\infty.
 \label{eq:observed-pvalue-limit}
\end{equation}
Each marginal of $\bm P_\infty$ is uniform on $(0,1)$.

For a fresh common-bridge draw, the conditional joint-law convergence and
\eqref{eq:conditional-cdf-consistency} yield
\begin{equation}
 d_{\mathrm{BL}}\!\left[
 \cL\{(P_{n,1}^{\circ},\ldots,P_{n,L_{\cG}}^{\circ})
 \mid\widehat{\mathcal W}_n\},
 \cL(\bm P_\infty)\right]\pto0.
 \label{eq:reference-pvalue-limit}
\end{equation}
For $c_{\rm cut}\in(0,1/2)$, define the bounded continuous map
\[
 \mathfrak C_{c_{\rm cut}}^{\rm tr}(\bm p)
 =\sum_{v=1}^{L_{\cG}}\varpi_v
 \mathfrak c_{\rm C}\{\min(1-c_{\rm cut},\max(c_{\rm cut},p_v))\}.
\]
By \eqref{eq:observed-pvalue-limit} and
\eqref{eq:reference-pvalue-limit}, the observed and conditional reference
laws of $\mathfrak C_{c_{\rm cut}}^{\rm tr}$ converge to the law of $\mathfrak C_{c_{\rm cut}}^{\rm tr}(\bm P_\infty)$.
Furthermore,
\[
 \Pp\left(
 \min_vP_{\infty,v}\le c_{\rm cut}
 \ \text{or}\
 \max_vP_{\infty,v}\ge1-c_{\rm cut}
 \right)
 \le2L_{\cG}c_{\rm cut}.
\]
The same bound holds asymptotically for the observed vector and in
probability for the conditional reference vector.  Letting
$c_{\rm cut}\downarrow0$ gives
\begin{equation}
 \mathfrak C_n\dto \mathfrak C_\infty,
 \qquad
 d_{\mathrm{BL}}\!\left\{
 \cL(\mathfrak C_n^{\circ}\mid\widehat{\mathcal W}_n),
 \cL(\mathfrak C_\infty)\right\}\pto0.
 \label{eq:Cauchy-joint-law-limit}
\end{equation}

Assume $\mathsf{CR}_{\alpha}$ and fix $c_{\rm qnt}>0$.  A distribution
function has at most countably many discontinuities, so one may choose
continuity points $x_-$ and $x_+$ of $G_C$ such that
\begin{equation*}
 c_{C,1-\alpha}-c_{\rm qnt}<x_-<c_{C,1-\alpha}<x_+
 <c_{C,1-\alpha}+c_{\rm qnt}.
\end{equation*}
The strict crossing in $\mathsf{CR}_{\alpha}$ gives
\begin{equation*}
 G_C(x_-)<1-\alpha<G_C(x_+).
\end{equation*}
The conditional weak convergence in \eqref{eq:Cauchy-joint-law-limit},
evaluated at these two continuity points, transfers the inequalities to the
conditional law of $\mathfrak C_n^{\circ}$ with probability tending to one.
The definition of the generalized inverse then yields
\begin{equation*}
 x_-<\widehat c_{n,C,1-\alpha}\le x_+,
\end{equation*}
and hence
\begin{equation*}
 |\widehat c_{n,C,1-\alpha}-c_{C,1-\alpha}|<c_{\rm qnt}
\end{equation*}
with probability tending to one.  Since $G_C$ is continuous at
$c_{C,1-\alpha}$ and $G_C(c_{C,1-\alpha})=1-\alpha$,
\[
 \widehat c_{n,C,1-\alpha}\pto c_{C,1-\alpha},
 \qquad
 \Pp\{\mathfrak C_n>\widehat c_{n,C,1-\alpha}\}\to\alpha.
\]

It remains to control finite bridge simulation.  Conditional on the data,
the bridge replicates in a batch of size $N_{\rm sim}$ are independent and
identically distributed.
The sharp Dvoretzky--Kiefer--Wolfowitz inequality
\citep{Massart1990}, followed by a union bound over the fixed ridge grid,
gives \eqref{eq:DKW-marginal-bridge}.  Applied to an independent second
batch, the same inequality gives \eqref{eq:DKW-joint-bridge}.

For the practical first batch, define
\[
 \mathfrak E_{n,N_{{\rm sim},1}}^{\rm sim}
 =\max_{1\le v\le L_{\cG}}\sup_x
 \left|\widetilde P_{n,v,N_{{\rm sim},1}}(x)
 -\{1-\widehat F_{n,a_v}(x)\}\right|.
\]
Let $\widehat F_{n,a_v,N_{{\rm sim},1}}(x-)$ denote the left limit of the empirical
cdf.  The definition in \eqref{eq:midrank-bridge-pvalue} gives the exact
identity
\begin{equation*}
 \widetilde P_{n,v,N_{{\rm sim},1}}(x)
 =\frac{N_{{\rm sim},1}}{N_{{\rm sim},1}+1}
   \{1-\widehat F_{n,a_v,N_{{\rm sim},1}}(x-)\}
  +\frac{1}{2(N_{{\rm sim},1}+1)}.
\end{equation*}
Because $\widehat F_{n,a_v}$ is continuous on $\mathcal E_n$, for every
$x$ the left-hand difference below is the limit of the corresponding
right-continuous empirical-cdf differences along $y\uparrow x$.  Therefore
\begin{equation*}
 \sup_x|\widehat F_{n,a_v,N_{{\rm sim},1}}(x-)-\widehat F_{n,a_v}(x)|
 \le
 \sup_x|\widehat F_{n,a_v,N_{{\rm sim},1}}(x)-\widehat F_{n,a_v}(x)|.
\end{equation*}
It follows that
\begin{equation*}
 \mathfrak E_{n,N_{{\rm sim},1}}^{\rm sim}
 \le
 \max_{1\le v\le L_{\cG}}\sup_x
 |\widehat F_{n,a_v,N_{{\rm sim},1}}(x)-\widehat F_{n,a_v}(x)|
 +\frac{1}{2(N_{{\rm sim},1}+1)}.
\end{equation*}
The DKW inequality and a union bound over the fixed ridge grid therefore
imply, for every $y>0$,
\begin{equation}
 \Pp_{\bm B}\left\{
 \mathfrak E_{n,N_{{\rm sim},1}}^{\rm sim}>y+\frac{1}{2(N_{{\rm sim},1}+1)}
 \mid\widehat{\mathcal W}_n\right\}
 \le2L_{\cG}\exp(-2N_{{\rm sim},1}y^2).
 \label{eq:midrank-DKW-bound}
\end{equation}
Fix $0<c_{\rm cut}<1/4$.  On $[c_{\rm cut},1-c_{\rm cut}]$,
\[
 |\mathfrak c_{\rm C}'(z)|
 =\pi\csc^2(\pi z)
 \le\pi\csc^2(\pi c_{\rm cut}).
\]
By Lemma~\ref{lem:finite-reference-regularity} and a union bound,
\[
 \Pp_{\bm B}\left(
 \min_vP_{n,v}^{\circ}\le2c_{\rm cut}
 \ \text{or}\
 \max_vP_{n,v}^{\circ}\ge1-2c_{\rm cut}
 \mid\widehat{\mathcal W}_n\right)
 \le4L_{\cG}c_{\rm cut}.
\]
On the complement of this event and on
$\{\mathfrak E_{n,N_{{\rm sim},1}}^{\rm sim}\le c_{\rm cut}\}$, both the exact and first-batch
marginal tail values belong to $[c_{\rm cut},1-c_{\rm cut}]$.  Therefore
\[
 |\mathfrak C_{n,N_{{\rm sim},1}}^{\circ}-\mathfrak C_n^{\circ}|
 \le
 \pi\csc^2(\pi c_{\rm cut})\mathfrak E_{n,N_{{\rm sim},1}}^{\rm sim}.
\]
Write $C_{\rm Cau}(c_{\rm cut})=\pi\csc^2(\pi c_{\rm cut})$.  Conditional on the
data and the first batch, for every $\rho_{\rm MC}>0$,
\begin{align}
 &\Pp_{\bm B}\left\{
 |\mathfrak C_{n,N_{{\rm sim},1}}^{\circ}-\mathfrak C_n^{\circ}|>\rho_{\rm MC}
 \mid\widehat{\mathcal W}_n,\text{ first batch}\right\}
 \notag\\
 &\qquad\le4L_{\cG}c_{\rm cut}
 +\ind\!\left\{\mathfrak E_{n,N_{{\rm sim},1}}^{\rm sim}>
 \min(c_{\rm cut},\rho_{\rm MC}/C_{\rm Cau}(c_{\rm cut}))\right\}.
 \label{eq:MC-coupling-conditional-bound}
\end{align}
For any joint sequence $n\to\infty$ and $N_{{\rm sim},1}(n)\to\infty$,
\eqref{eq:midrank-DKW-bound} makes the indicator on the right converge to
zero in probability, uniformly in $n$.  Letting first $n,N_{{\rm sim},1}\to\infty$ and
then $c_{\rm cut}\downarrow0$ proves
\begin{equation}
 \mathfrak C_{n,N_{{\rm sim},1}}^{\circ}-\mathfrak C_n^{\circ}
 \longrightarrow0.
 \label{eq:MC-reference-comparison}
\end{equation}
Here the convergence is in conditional probability given the data and first
batch, in outer probability.  Moreover,
\eqref{eq:MC-coupling-conditional-bound} and
$\E\min(2,|X|)\le\rho_{\rm MC}+2\Pp(|X|>\rho_{\rm MC})$ give the corresponding conditional
$L^1$ convergence.
The same truncation argument, now using
\eqref{eq:observed-pvalue-limit}, gives
\[
 \mathfrak C_{n,N_{{\rm sim},1}}-\mathfrak C_n\pto0.
\]
Let $\widehat G_{n,N_{{\rm sim},1}}^{\circ}$ be the practical conditional reference cdf
defined in Lemma~\ref{lem:Monte-Carlo-reference}.  The coupling in
\eqref{eq:MC-reference-comparison} implies
\begin{align*}
 &d_{\mathrm{BL}}\{\widehat G_{n,N_{{\rm sim},1}}^{\circ},
 \cL(\mathfrak C_n^{\circ}\mid\widehat{\mathcal W}_n)\}\\
 &\quad\le
 \E_{\bm B}\!\left[
 \min\{2,|\mathfrak C_{n,N_{{\rm sim},1}}^{\circ}-\mathfrak C_n^{\circ}|\}
 \mid\widehat{\mathcal W}_n,\text{ first batch}\right]\pto0.
\end{align*}
Together with \eqref{eq:Cauchy-joint-law-limit}, this gives conditional weak
convergence of the practical first-batch reference law to $G_C$.  Conditional
DKW in \eqref{eq:DKW-joint-bridge} for the independent second batch and the
preceding quantile-separation argument therefore yield
\[
 \widetilde c_{n,C,1-\alpha}^{(N_{{\rm sim},1},N_{{\rm sim},2})}
 \pto c_{C,1-\alpha},
 \qquad
 \Pp\{\mathfrak C_{n,N_{{\rm sim},1}}>
 \widetilde c_{n,C,1-\alpha}^{(N_{{\rm sim},1},N_{{\rm sim},2})}\}\to\alpha.
\]
Because $\csc^2(\pi c_{\rm cut})\to\infty$ as $c_{\rm cut}\downarrow0$, this argument
establishes consistency but no universal $N_{{\rm sim},1}^{-1/2}$ rate after the
Cauchy transform.

\section{Proofs under alternatives}

\subsection{Proof of Theorem~\ref{thm:local-alternative}}

Recall that
\[
 t_k=\frac{k}{n},
 \qquad
 \vartheta_n=\frac{\tau_n}{n},
 \qquad
 d_n=\mathcal J_{p,N}\delta_n=\iota_p\bdelta_{p,N}.
\]
Define the pure-noise score CUSUM by
\begin{equation*}
 \widehat\bS_{n,0}(t)
 =\frac1{\sqrt n}\left\{
 \sum_{i=1}^{\lfloor nt\rfloor}\beps_{i,p,N}
 -\frac{\lfloor nt\rfloor}{n}\sum_{i=1}^{n}\beps_{i,p,N}
 \right\},
 \qquad 0\le t\le1.
\end{equation*}
Direct summation of the post-change mean gives
\begin{equation*}
 \E\widehat\bS_n(t_k)
 =-\sqrt n\,g(t_k,\vartheta_n)\bdelta_{p,N}.
\end{equation*}
For $t\in[0,1]$, put $t_n=\lfloor nt\rfloor/n$.  Since
$g(\cdot,\vartheta_n)$ is one-Lipschitz,
\begin{align*}
 &\sup_{0\le t\le1}
 \sqrt n\,|g(t_n,\vartheta_n)-g(t,\vartheta_n)|
 \norm{d_n}\\
 &\quad\le
 n^{-1}\norm{\sqrt n\,d_n}
 \longrightarrow0.
\end{align*}
Hence
\begin{equation*}
 \iota_p\widehat\bS_n(t)
 =
 \iota_p\widehat\bS_{n,0}(t)
 -\sqrt n\,g(t,\vartheta_n)d_n
 +\mathsf r_n^{\mu}(t),
 \qquad
 \sup_t\norm{\mathsf r_n^{\mu}(t)}\longrightarrow0.
\end{equation*}
By the null functional limit,
\[
 \iota_p\widehat\bS_{n,0}\dto\mathbb B_\Omega.
\]
Furthermore,
\begin{align*}
 &\sup_{t\in[\epsilon,1-\epsilon]}
 \norm{\sqrt n\,g(t,\vartheta_n)d_n-g(t,\vartheta)d}\\
 &\quad\le
 \sup_t|g(t,\vartheta_n)|
 \norm{\sqrt n\,d_n-d}
 +\sup_t|g(t,\vartheta_n)-g(t,\vartheta)|
 \norm d\\
 &\longrightarrow0.
\end{align*}
Since the null limit is continuous, the Skorokhod convergence and
the preceding uniform deterministic approximation imply
\begin{equation}
 \iota_p\widehat\bS_n
 \dto
 \mathbb B_\Omega-g(\,\cdot\,,\vartheta)d
 \quad\text{in }
 \ell^\infty([\epsilon,1-\epsilon],\cH).
 \label{eq:local-CUSUM-limit}
\end{equation}

Set
\[
 \mathfrak e_{{\rm alt},n}
 =\sqrt p\{\mathfrak r_n+\mathfrak r_{\delta,n}\}+\zeta_{p,N}.
\]
Lemmas~\ref{prop:alternative-DB},
\ref{prop:quadrature}, and
\ref{prop:weight-consistency} imply
\begin{align*}
 \norm{\widehat\Omega_{n,+}^{\cH}-\Omega}_1
 +|\widehat{\mathfrak t}_n-\mathfrak t_\Omega|
 &=O_{\Pp}(\mathfrak e_{{\rm alt},n}),\\
 \max_{a\in\cG}
 \norm{\widehat{\mathcal R}_{n,a}-\mathcal R_a}_{\op}
 &=O_{\Pp}(\mathfrak e_{{\rm alt},n}),\\
 \max_{a\in\cG}
 \left\{
 |\widehat A_{1,n}(a)-A_1(a)|
 +|\widehat A_{2,n}(a)-A_2(a)|
 \right\}
 &=O_{\Pp}(\mathfrak e_{{\rm alt},n}).
\end{align*}
Because $\mathfrak e_{{\rm alt},n}\to0$ and
\[
 \sup_t\norm{\widehat\bS_n(t)}=O_{\Pp}(1),
\]
uniformly over $a\in\cG$,
\begin{align*}
 &\sup_t
 \left|
 \ip{\widehat\bS_n(t)}
 {\widehat{\mathcal R}_{n,a}\widehat\bS_n(t)}
 -
 \ip{\widehat\bS_n(t)}
 {\mathcal R_a\widehat\bS_n(t)}
 \right|\\
 &\quad\le
 \norm{\widehat{\mathcal R}_{n,a}-\mathcal R_a}_{\op}
 \sup_t\norm{\widehat\bS_n(t)}^2
 =O_{\Pp}(\mathfrak e_{{\rm alt},n}).
\end{align*}
The continuous mapping theorem applied to
\eqref{eq:local-CUSUM-limit} gives
\begin{align*}
 V_{n,a}(\lfloor nt\rfloor)
 &\dto
 \frac1{t(1-t)}
 \ip{\mathbb B_\Omega(t)-g(t,\vartheta)d}
 {\mathcal R_a\{\mathbb B_\Omega(t)-g(t,\vartheta)d\}}
\end{align*}
jointly over $a\in\cG$.

Write
\[
 d=d_0+\sum_{j\ge1}d_j\psi_j,
 \qquad d_0\in\ker(\Omega),
\]
and
\[
 \mathbb B_\Omega(t)
 =\sum_{j\ge1}\sqrt{\lambda_j}B_j(t)\psi_j.
\]
Since
\[
 \mathcal R_a d_0=(a\mathfrak t_\Omega)^{-1}d_0,
 \qquad
 \mathcal R_a\psi_j=(\lambda_j+a\mathfrak t_\Omega)^{-1}\psi_j,
\]
the limiting quadratic form equals
\begin{align*}
 &\sum_{j\ge1}w_j(a)\frac{B_j(t)^2}{t(1-t)}
 -\frac{2g(t,\vartheta)}{t(1-t)}
 \sum_{j\ge1}
 \frac{\sqrt{\lambda_j}d_j}{\lambda_j+a\mathfrak t_\Omega}B_j(t)\\
 &\quad+
 \frac{g(t,\vartheta)^2}{t(1-t)}
 \left\{
 \sum_{j\ge1}\frac{d_j^2}{\lambda_j+a\mathfrak t_\Omega}
 +\frac{\norm{d_0}^2}{a\mathfrak t_\Omega}
 \right\}.
\end{align*}
For
\[
 \beta_j^{\rm nc}(a)=\frac{\sqrt{\lambda_j}d_j}
              {\lambda_j+a\mathfrak t_\Omega},
\]
we have, uniformly over the finite ridge grid,
\[
 \sum_{j\ge1}\beta_j^{\rm nc}(a)^2
 \le
 \sup_{x\ge0}\frac{x}{(x+a_-\mathfrak t_\Omega)^2}\norm{d}^2
 =\frac{\norm{d}^2}{4a_-\mathfrak t_\Omega}.
\]
For $J_2^{\rm ser}>J_1^{\rm ser}$, the finite tail
$\sum_{j=J_1^{\rm ser}+1}^{J_2^{\rm ser}}\beta_j^{\rm nc}(a)B_j(t)$ is a scalar Brownian bridge with variance
factor $\sum_{j=J_1^{\rm ser}+1}^{J_2^{\rm ser}}\beta_j^{\rm nc}(a)^2$.  Hence Doob's inequality gives
\begin{align*}
 &\E\sup_{t\in[\epsilon,1-\epsilon]}
 \left|
 \sum_{j=J_1^{\rm ser}+1}^{J_2^{\rm ser}}\beta_j^{\rm nc}(a)U_j(t)
 \right|^2\\
 &\quad\le
 \{\epsilon(1-\epsilon)\}^{-1}
 \E\sup_{0\le t\le1}B_1(t)^2
 \sum_{j=J_1^{\rm ser}+1}^{J_2^{\rm ser}}\beta_j^{\rm nc}(a)^2
 \longrightarrow0.
\end{align*}
The convergence holds as $J_1^{\rm ser}\to\infty$, uniformly over
$J_2^{\rm ser}>J_1^{\rm ser}$.  Thus the random linear series in
\eqref{eq:noncentral-weighted-bridge} converges uniformly in $L^2$ and has a
continuous version.  The deterministic series satisfies
\[
 \sum_{j\ge1}\frac{d_j^2}{\lambda_j+a\mathfrak t_\Omega}
 +\frac{\norm{d_0}^2}{a\mathfrak t_\Omega}
 \le\frac{\norm{d}^2}{a_-\mathfrak t_\Omega}.
\]
After centering by $A_1(a)$ and scaling by
$\{2A_2(a)\}^{1/2}$, the limit is exactly
\eqref{eq:noncentral-weighted-bridge}.  Slutsky's theorem proves
\eqref{eq:local-process-limit}.

\subsection{Proof of Theorem~\ref{thm:local-power} and the power calculations}

The one-change version of Lemma~\ref{prop:weight-consistency} and the local
alternative rate in Theorem~\ref{thm:local-alternative} give, uniformly over
the finite ridge grid,
\begin{equation*}
 \sum_{j\ge1}|\widehat w_{j,n}(a)-w_j(a)|
 +|\widehat A_{1,n}(a)-A_1(a)|
 +|\widehat A_{2,n}(a)-A_2(a)|=o_{\Pp}(1).
\end{equation*}
The proof of Theorem~\ref{thm:critical-values}, which uses only this spectral
weight convergence for the reference process, therefore remains valid under
the local alternative.  In particular,
\begin{equation}
 \max_{a\in\cG}\sup_x|\widehat F_{n,a}(x)-F_a(x)|\pto0,
 \qquad
 \max_{a\in\cG}|\widehat c_{n,a,1-\alpha}-c_{a,1-\alpha}|\pto0,
 \label{eq:alternative-reference-convergence}
\end{equation}
and the joint Cauchy critical value converges to $c_{C,1-\alpha}$ under the
condition $\mathsf{CR}_{\alpha}$.

Theorem~\ref{thm:local-alternative} and the continuous mapping theorem yield
\begin{equation}
 (T_{n,a}:a\in\cG)\dto(\mathcal T_{a,d}:a\in\cG).
 \label{eq:local-scan-vector-limit}
\end{equation}
We next verify that every $\mathcal T_{a,d}$ has a continuous distribution.
For fixed $a$, define
\begin{align*}
 Q_{a,d}
 =\sup_{\epsilon\le t\le1-\epsilon}
 \Bigg[&\sum_{j\ge1}
 \left\{
 \sqrt{w_j(a)}U_j(t)
 -\frac{g(t,\vartheta)d_j}
 {\sqrt{t(1-t)}\sqrt{\lambda_j+a\mathfrak t_\Omega}}
 \right\}^2\\
 &+\frac{g(t,\vartheta)^2\norm{d_0}^2}
 {t(1-t)a\mathfrak t_\Omega}
 \Bigg].
\end{align*}
Then
\begin{equation*}
 \mathcal T_{a,d}
 =\frac{Q_{a,d}-A_1(a)}{\{2A_2(a)\}^{1/2}}.
\end{equation*}
Choose an index $j_0$ with $\lambda_{j_0}>0$.  The Karhunen--Lo\`eve expansion of
$B_{j_0}$ permits the decomposition
\begin{equation*}
 B_{j_0}(t)=Z\varphi_1(t)+\widetilde B_{j_0}(t),
 \qquad
 \varphi_1(t)=\frac{\sqrt2\sin(\pi t)}{\pi},
\end{equation*}
where $Z\sim N(0,1)$ is independent of $\widetilde B_{j_0}$ and of all other
bridges.  Conditional on these remaining Gaussian variables, $Q_{a,d}$ is a
function of $Z$ of the form
\begin{equation*}
 q_a(z)=\sup_{\epsilon\le t\le1-\epsilon}
 \{[r_a(t)z+s_a(t)]^2+R_a(t)\},
\end{equation*}
where $R_a(t)\ge0$ and
\begin{equation*}
 r_a(t)=\frac{\sqrt{w_{j_0}(a)}\varphi_1(t)}{\sqrt{t(1-t)}}.
\end{equation*}
Because the scan interval is trimmed,
\begin{equation*}
 m_{a,\epsilon}:=\inf_{\epsilon\le t\le1-\epsilon}r_a(t)^2>0.
\end{equation*}
For every $t$, the map $z\mapsto[r_a(t)z+s_a(t)]^2+R_a(t)$ is
$m_{a,\epsilon}$-strongly convex.  The pointwise supremum preserves strong
convexity, so $q_a$ is coercive, has a unique minimizer, and every level set
$\{z:q_a(z)=x\}$ contains at most two points.  Since the conditional law of
$Z$ has a density,
\begin{equation*}
 \Pp(Q_{a,d}=x\mid\widetilde B_{j_0},\{B_j:j\ne j_0\})=0,
 \qquad x\in\R.
\end{equation*}
Integration proves continuity of the distribution of $\mathcal T_{a,d}$.
Equations~\eqref{eq:alternative-reference-convergence} and
\eqref{eq:local-scan-vector-limit} now imply
\eqref{eq:fixed-ridge-local-power}.

The same two equations and continuity of every $F_a$ give
\begin{equation}
 (P_{n,1},\ldots,P_{n,L_{\cG}})
 \dto(P_{a_1,d},\ldots,P_{a_{L_{\cG}},d}).
 \label{eq:alternative-pvalue-vector-limit}
\end{equation}
Every limiting tail value is strictly inside the unit interval.  To see this,
write
\begin{equation*}
 \ell_a=-\frac{A_1(a)}{\{2A_2(a)\}^{1/2}},
\end{equation*}
which is the lower endpoint of the support of $F_a$.  Since $Q_{a,d}\ge0$
and its distribution is continuous by the preceding argument,
$\Pp(Q_{a,d}=0)=0$.  Hence $\mathcal T_{a,d}>\ell_a$ almost surely, and strict
increase of $F_a$ on its support gives $F_a(\mathcal T_{a,d})>0$.  Moreover,
$\mathcal T_{a,d}$ is finite almost surely and the null reference law has
unbounded upper support, so $F_a(\mathcal T_{a,d})<1$.  Therefore
\begin{equation}
 0<P_{a,d}<1\qquad\text{almost surely for every }a\in\cG.
 \label{eq:alternative-pvalues-interior}
\end{equation}

For $0<c_{\rm cut}<1/4$, let
$\mathfrak C_{c_{\rm cut}}^{\rm tr}$ be the bounded truncated transform from
the proof of Theorem~\ref{thm:critical-values}.  Equations
\eqref{eq:alternative-pvalue-vector-limit} and
\eqref{eq:alternative-pvalues-interior} imply
\begin{align*}
 &\lim_{c_{\rm cut}\downarrow0}\limsup_{n\to\infty}
 \Pp\left\{
 \min_vP_{n,v}\le2c_{\rm cut}
 \ \text{or}\
 \max_vP_{n,v}\ge1-2c_{\rm cut}
 \right\}=0,\\
 &\lim_{c_{\rm cut}\downarrow0}
 \Pp\left\{
 \min_vP_{a_v,d}\le2c_{\rm cut}
 \ \text{or}\
 \max_vP_{a_v,d}\ge1-2c_{\rm cut}
 \right\}=0.
\end{align*}
For each fixed $c_{\rm cut}$, the truncated map is bounded and continuous.
The converging-together argument obtained by first letting $n\to\infty$ and
then $c_{\rm cut}\downarrow0$ therefore yields
\begin{equation}
 \mathfrak C_n\dto\mathfrak C_d.
 \label{eq:alternative-Cauchy-limit}
\end{equation}

Assume now $\mathsf{CR}_{\alpha}$.  The joint critical value converges to
$c_{C,1-\alpha}$ by \eqref{eq:alternative-reference-convergence} and the
reference-law part of Theorem~\ref{thm:critical-values}.  From
\eqref{eq:alternative-Cauchy-limit}, convergence of the critical value, and
the Portmanteau theorem,
\begin{align*}
 \Pp\{\mathfrak C_d>c_{C,1-\alpha}\}
 &\le\liminf_{n\to\infty}
 \Pp\{\mathfrak C_n>\widehat c_{n,C,1-\alpha}\}\\
 &\le\limsup_{n\to\infty}
 \Pp\{\mathfrak C_n>\widehat c_{n,C,1-\alpha}\}
 \le\Pp\{\mathfrak C_d\ge c_{C,1-\alpha}\}.
\end{align*}
This proves \eqref{eq:Cauchy-local-power-bounds}, and it proves
\eqref{eq:Cauchy-local-power} when the boundary atom is absent.

It remains to justify the two-batch statement under the local alternative.
For the first simulation batch, define
\begin{equation*}
 \mathfrak E_{n,N_{{\rm sim},1}}^{\rm sim}
 =\max_{1\le v\le L_{\cG}}\sup_x
 \left|\widetilde P_{n,v,N_{{\rm sim},1}}(x)
 -\{1-\widehat F_{n,a_v}(x)\}\right|.
\end{equation*}
The conditional DKW bound \eqref{eq:midrank-DKW-bound} is distribution-free
and remains valid.  On the event
\begin{equation*}
 \min_vP_{n,v}\ge2c_{\rm cut},\qquad
 \max_vP_{n,v}\le1-2c_{\rm cut},\qquad
 \mathfrak E_{n,N_{{\rm sim},1}}^{\rm sim}\le c_{\rm cut},
\end{equation*}
both the exact and mid-rank tail values lie in
$[c_{\rm cut},1-c_{\rm cut}]$.  The mean-value theorem then gives
\begin{equation*}
 |\mathfrak C_{n,N_{{\rm sim},1}}-\mathfrak C_n|
 \le\pi\csc^2(\pi c_{\rm cut})
 \mathfrak E_{n,N_{{\rm sim},1}}^{\rm sim}.
\end{equation*}
Using the endpoint bound above, then
\eqref{eq:midrank-DKW-bound}, and finally letting
$c_{\rm cut}\downarrow0$, we obtain, for every joint sequence
$n,N_{{\rm sim},1}\to\infty$,
\begin{equation}
 \mathfrak C_{n,N_{{\rm sim},1}}-\mathfrak C_n\pto0.
 \label{eq:alternative-MC-observed-coupling}
\end{equation}
The simulated reference statistics depend on the observations only through
the estimated ridge weights.  Therefore the weight convergence in
\eqref{eq:alternative-reference-convergence}, the reference coupling
\eqref{eq:MC-reference-comparison}, and the second-batch DKW bound
\eqref{eq:DKW-joint-bridge} give, under $\mathsf{CR}_{\alpha}$,
\begin{equation*}
 \widetilde c_{n,C,1-\alpha}^{(N_{{\rm sim},1},N_{{\rm sim},2})}
 \pto c_{C,1-\alpha}
\end{equation*}
for every joint sequence $n,N_{{\rm sim},1},N_{{\rm sim},2}\to\infty$.
Combining this convergence with
\eqref{eq:alternative-Cauchy-limit} and
\eqref{eq:alternative-MC-observed-coupling} proves the same Portmanteau bounds,
and hence the same equality under boundary atomlessness, for the two-batch
implementation.

It remains to verify the displayed power calculations.  At $t=\vartheta$,
\begin{equation*}
 g(\vartheta,\vartheta)=\vartheta(1-\vartheta),
 \qquad
 U_j(\vartheta)\stackrel{\mathrm{iid}}{\sim}N(0,1).
\end{equation*}
For $\eta_j=\{\vartheta(1-\vartheta)\}^{1/2}d_j/\sqrt{\lambda_j}$,
\begin{align*}
 &w_j(a)\{Z_j^2-1\}
 -2\{\vartheta(1-\vartheta)\}^{1/2}
 \frac{\sqrt{\lambda_j}d_j}{\lambda_j+a\mathfrak t_\Omega}Z_j
 +\vartheta(1-\vartheta)
 \frac{d_j^2}{\lambda_j+a\mathfrak t_\Omega}\\
 &\qquad=w_j(a)\{(Z_j-\eta_j)^2-1\}.
\end{align*}
Moreover,
\begin{align*}
 \sum_{j\ge1}\E\left[w_j(a)(Z_j-\eta_j)^2\right]
 &=A_1(a)+\vartheta(1-\vartheta)
 \sum_{j\ge1}\frac{d_j^2}{\lambda_j+a\mathfrak t_\Omega}\\
 &\le A_1(a)+
 \frac{\vartheta(1-\vartheta)\norm{d}^2}{a\mathfrak t_\Omega}<\infty.
\end{align*}
The nonnegative partial sums therefore converge almost surely and in $L^1$.
Applying the preceding identity to finite partial sums and passing to the
$L^1$ limit gives \eqref{eq:true-change-noncentral-form}.  More generally, the
centered quadratic and linear Gaussian terms in
\eqref{eq:noncentral-weighted-bridge} have mean zero, so
\begin{equation*}
 \E\{\mathcal D_{a,d}(t)\}
 =\frac{\mathfrak h_{\rm cp}(t,\vartheta)}{\{2A_2(a)\}^{1/2}}
 \left\{
 \sum_{j\ge1}\frac{d_j^2}{\lambda_j+a\mathfrak t_\Omega}
 +\frac{\norm{d_0}^2}{a\mathfrak t_\Omega}
 \right\},
\end{equation*}
which is \eqref{eq:local-power-mean-drift}.  At the true fraction,
\begin{equation*}
 \sum_{j\ge1}w_j(a)^2\eta_j^2
 =\vartheta(1-\vartheta)
 \sum_{j\ge1}\frac{\lambda_jd_j^2}
 {(\lambda_j+a\mathfrak t_\Omega)^2}
 \le\frac{\vartheta(1-\vartheta)\norm{d}^2}
 {4a\mathfrak t_\Omega}<\infty,
\end{equation*}
where $\sup_{x\ge0}x/(x+c)^2=1/(4c)$.  Thus the centered series is square
integrable.  Independence and
\begin{equation*}
 \Var(Z_j^2-1)=2,
 \qquad
 \Cov(Z_j^2-1,Z_j)=0
\end{equation*}
then give
\begin{align*}
 \Var\{\mathcal D_{a,d}(\vartheta)\}
 &=\frac{1}{2A_2(a)}
 \sum_{j\ge1}w_j(a)^2\{2+4\eta_j^2\}\\
 &=1+\frac{2\vartheta(1-\vartheta)}{A_2(a)}
 \sum_{j\ge1}\frac{\lambda_jd_j^2}
 {(\lambda_j+a\mathfrak t_\Omega)^2},
\end{align*}
proving \eqref{eq:local-power-pointwise-variance}.  For $t\le\vartheta$ and
$t\ge\vartheta$, respectively,
\begin{equation*}
 \mathfrak h_{\rm cp}(t,\vartheta)
 =\frac{t(1-\vartheta)^2}{1-t},
 \qquad
 \mathfrak h_{\rm cp}(t,\vartheta)
 =\frac{\vartheta^2(1-t)}{t}.
\end{equation*}
The first expression is strictly increasing and the second strictly decreasing,
so the unique maximum is
$\mathfrak h_{\rm cp}(\vartheta,\vartheta)=\vartheta(1-\vartheta)$.
This proves \eqref{eq:ridge-standardized-signal}.  Since
$\mathcal T_{a,d}\ge\mathcal D_{a,d}(\vartheta)$, rewriting
\eqref{eq:true-change-noncentral-form} in terms of independent
$\chi_1^2(\eta_j^2)$ variables gives
\eqref{eq:pointwise-power-lower-bound}.

For completeness, suppose $d\ne0$ and define
\begin{equation*}
 \mathsf S_a(d)=\frac{\mathfrak q_a(d)}{\{A_2(a)\}^{1/2}},
 \qquad
 B_3(a)=\sum_{j\ge1}
 \frac{\lambda_j^2}{(\lambda_j+a\mathfrak t_\Omega)^3}.
\end{equation*}
Resolvent differentiation and termwise differentiation of $A_2(a)$ give
\begin{equation*}
 \frac{d}{da}\mathfrak q_a(d)
 =-\mathfrak t_\Omega\ip{d}{\mathcal R_a^2d},
 \qquad
 A_2'(a)=-2\mathfrak t_\Omega B_3(a).
\end{equation*}
Consequently,
\begin{equation}
 \frac{d}{da}\log\mathsf S_a(d)
 =\mathfrak t_\Omega
 \left\{
 \frac{B_3(a)}{A_2(a)}
 -\frac{\ip{d}{\mathcal R_a^2d}}{\mathfrak q_a(d)}
 \right\}.
 \label{eq:ridge-oracle-first-order}
\end{equation}
Any interior maximizer over a continuous ridge interval therefore solves the
zero-derivative equation in \eqref{eq:ridge-oracle-first-order}, whose second
term depends explicitly on the unknown alternative direction $d$.  Finally,
if $\Omega$ has finite rank $r$ and $d_0=0$, direct convergence of the finite
sums gives the first limit in \eqref{eq:ridge-signal-limits}.  If $d_0\ne0$,
$\mathfrak q_a(d)\ge\norm{d_0}^2/(a\mathfrak t_\Omega)$ while
$A_2(a)\to r$, so the limit is infinite.  As $a\to\infty$,
\begin{equation*}
 a\mathfrak t_\Omega\mathfrak q_a(d)\longrightarrow\norm{d}^2,
 \qquad
 a^2\mathfrak t_\Omega^2A_2(a)
 \longrightarrow\sum_{j\ge1}\lambda_j^2=\norm{\Omega}_{\rm HS}^2,
\end{equation*}
which proves the second limit in \eqref{eq:ridge-signal-limits}.

\subsection{Proofs of Theorems~\ref{thm:consistency} and~\ref{thm:consistency-location}}

Let
\[
 \widehat{\mathfrak q}_{n,a}(d_n)
 =\ip{d_n}{\widehat{\mathcal R}_{n,a}d_n}.
\]
The signal condition and ridge positivity imply that $d_n\ne0$ and
$\mathfrak q_a(d_n)>0$ for every $a\in\cG$ and all sufficiently large $n$.
By Lemma~\ref{prop:alternative-DB},
Lemma~\ref{prop:quadrature}, and
Lemma~\ref{lem:resolvent-replacement},
\begin{equation*}
 \max_{a\in\cG}
 \norm{\widehat{\mathcal R}_{n,a}-\mathcal R_a}_{\op}
 =O_{\Pp}\!\left[
 \sqrt p\{\mathfrak r_n+\mathfrak r_{\delta,n}\}+\zeta_{p,N}
 \right].
\end{equation*}
Since $0\preceq\Omega\preceq\mathfrak t_\Omega\Id_{\cH}$,
\begin{align*}
 |\widehat{\mathfrak q}_{n,a}(d_n)-\mathfrak q_a(d_n)|
 &\le
 \norm{\widehat{\mathcal R}_{n,a}-\mathcal R_a}_{\op}\norm{d_n}^2,
 \notag\\
 \frac{\norm{d_n}^2}{(1+a_+)\mathfrak t_\Omega}
 &\le \mathfrak q_a(d_n)
 \le\frac{\norm{d_n}^2}{a_-\mathfrak t_\Omega}.
\end{align*}
Thus
\begin{equation}
 \max_{a\in\cG}
 \left|\frac{\widehat{\mathfrak q}_{n,a}(d_n)}{\mathfrak q_a(d_n)}-1\right|
 =
 O_{\Pp}\!\left[
 \sqrt p\{\mathfrak r_n+\mathfrak r_{\delta,n}\}+\zeta_{p,N}
 \right].
 \label{eq:qhat-relative}
\end{equation}

On an event with probability tending to one,
\[
 \max_{a\in\cG}\norm{\widehat{\mathcal R}_{n,a}}_{\op}
 \le\frac{2}{a_-\mathfrak t_\Omega}.
\]
The centered CUSUM is tight, so
\begin{equation}
 \max_{a\in\cG}
 \max_{k\in\cK_{n,\epsilon}}
 \frac{
 \ip{\widehat\bS_{n,0}(t_k)}
 {\widehat{\mathcal R}_{n,a}\widehat\bS_{n,0}(t_k)}
 }{t_k(1-t_k)}
 =O_{\Pp}(1).
 \label{eq:null-quadratic-tight-audited}
\end{equation}
Cauchy--Schwarz in the inner product
$\ip{x}{y}_{\widehat{\mathcal R}_{n,a}}=\ip{x}{\widehat{\mathcal R}_{n,a}y}$ gives
\begin{align*}
 |\ip{d_n}{\widehat{\mathcal R}_{n,a}\widehat\bS_{n,0}(t_k)}|
 &\le
 \widehat{\mathfrak q}_{n,a}(d_n)^{1/2}
 \ip{\widehat\bS_{n,0}(t_k)}
 {\widehat{\mathcal R}_{n,a}\widehat\bS_{n,0}(t_k)}^{1/2}.
\end{align*}
Because
\[
 \sup_{t\in[\epsilon,1-\epsilon]}
 \frac{|g(t,\vartheta_n)|}{\sqrt{t(1-t)}}\le C_\epsilon,
\]
\eqref{eq:null-quadratic-tight-audited} and
\eqref{eq:qhat-relative} imply
\begin{equation}
 \max_{a\in\cG}
 \max_{k\in\cK_{n,\epsilon}}
 \frac{1}{\sqrt{n\mathfrak q_a(d_n)}}
 \left|
 \frac{2\sqrt n\,g(t_k,\vartheta_n)}{t_k(1-t_k)}
 \ip{d_n}{\widehat{\mathcal R}_{n,a}\widehat\bS_{n,0}(t_k)}
 \right|
 =O_{\Pp}(1).
 \label{eq:cross-term-order-audited}
\end{equation}

The exact CUSUM decomposition gives
\begin{align*}
 V_{n,a}(k)
 &=
 \frac{
 \ip{\widehat\bS_{n,0}(t_k)}
 {\widehat{\mathcal R}_{n,a}\widehat\bS_{n,0}(t_k)}
 }{t_k(1-t_k)}
 \notag\\
 &\quad-
 \frac{2\sqrt n\,g(t_k,\vartheta_n)}{t_k(1-t_k)}
 \ip{d_n}{\widehat{\mathcal R}_{n,a}\widehat\bS_{n,0}(t_k)}
 \notag\\
 &\quad+
 n\frac{g(t_k,\vartheta_n)^2}{t_k(1-t_k)}
 \widehat{\mathfrak q}_{n,a}(d_n).
\end{align*}
Divide by $n\mathfrak q_a(d_n)$.  Equations
\eqref{eq:null-quadratic-tight-audited},
\eqref{eq:cross-term-order-audited}, and
\eqref{eq:qhat-relative} yield, uniformly over the finite ridge grid,
\begin{align}
 &\max_{a\in\cG}
 \max_{k\in\cK_{n,\epsilon}}
 \left|
 \frac{V_{n,a}(k)}{n\mathfrak q_a(d_n)}
 -\frac{g(t_k,\vartheta_n)^2}{t_k(1-t_k)}
 \right|
 \notag\\
 &\quad=
 O_{\Pp}\!\left[
 \max_{a\in\cG}(n\mathfrak q_a(d_n))^{-1}
 +\max_{a\in\cG}(n\mathfrak q_a(d_n))^{-1/2}
 +\sqrt p\{\mathfrak r_n+\mathfrak r_{\delta,n}\}
 +\zeta_{p,N}
 \right].
 \label{eq:uniform-drift-normalization-audited}
\end{align}
For $a,a'\in\cG$ and $x\in[0,\mathfrak t_\Omega]$,
\[
 \frac1{x+a'\mathfrak t_\Omega}
 \ge
 \frac{a_-}{a_+}\frac1{x+a\mathfrak t_\Omega}.
\]
Hence, by the spectral theorem,
\[
 \mathfrak q_{a'}(d_n)\ge c_\cG \mathfrak q_a(d_n),
 \qquad
 c_\cG=\frac{a_-}{a_+}.
\]
Thus $n\mathfrak q_a(d_n)\to\infty$ for one ridge implies
\[
 \min_{a'\in\cG}n\mathfrak q_{a'}(d_n)\to\infty.
\]
Under \eqref{eq:alt-estimation-condition}, the right-hand side of
\eqref{eq:uniform-drift-normalization-audited} converges to zero, proving
\eqref{eq:uniform-drift-main}.

Because $\vartheta_n\to\vartheta\in(\epsilon,1-\epsilon)$,
$\tau_n\in\cK_{n,\epsilon}$ for all sufficiently large $n$.  At
$k=\tau_n$,
\[
 \frac{g(\vartheta_n,\vartheta_n)^2}
 {\vartheta_n(1-\vartheta_n)}
 =\vartheta_n(1-\vartheta_n)
 \ge c_\epsilon>0.
\]
Therefore
\[
 \min_{a'\in\cG}V_{n,a'}(\tau_n)\pto\infty.
\]
Lemma~\ref{prop:weight-consistency} gives
\[
 \max_{a'\in\cG}\widehat A_{1,n}(a')=O_{\Pp}(1),
 \qquad
 \min_{a'\in\cG}\widehat A_{2,n}(a')
 \pto\min_{a'\in\cG}A_2(a')>0,
\]
and consequently
\[
 \min_{a'\in\cG}T_{n,a'}\pto\infty.
\]

Fix $v\in\{1,\ldots,L_{\cG}\}$. Conditional on the estimated weights,
\begin{align*}
 \E_{\bm B}
 \sup_{t\in[\epsilon,1-\epsilon]}
 |\widehat{\mathcal D}_{n,a_v}^{*}(t)|
 &\le
 \frac{\mathfrak u_{\epsilon}^{\rm br}\widehat A_{1,n}(a_v)}
 {\{2\widehat A_{2,n}(a_v)\}^{1/2}}
 =O_{\Pp}(1).
\end{align*}
Since
$|\widehat{\mathcal T}_{n,a_v}^{*}|\le
\sup_t|\widehat{\mathcal D}_{n,a_v}^{*}(t)|$, Markov's inequality gives,
for every fixed $\alpha_0\in(0,1)$,
\[
 \widehat c_{n,a_v,1-\alpha_0}
 \le\alpha_0^{-1}\E_{\bm B}
 |\widehat{\mathcal T}_{n,a_v}^{*}|
 =O_{\Pp}(1).
\]
Thus every fixed-ridge bridge-calibrated critical value is stochastically
bounded.  On $\{T_{n,a_v}>0\}$, the same inequality gives
\[
 P_{n,v}
 \le
 \frac{
 \E_{\bm B}|\widehat{\mathcal T}_{n,a_v}^{*}|
 }{T_{n,a_v}}
 \pto0.
\]
Hence
\[
 \min_{1\le v\le L_{\cG}}\mathfrak c_{\rm C}(P_{n,v})\pto\infty,
 \qquad
 \mathfrak C_n\pto\infty,
\]
which proves consistency of the analytic Cauchy statistic.

For the joint calibration, each conditional reference
$p$-value is uniform by
Lemma~\ref{lem:finite-reference-regularity}.  Therefore
$\mathfrak c_{\rm C}(P_{n,v}^{\circ})$ is standard Cauchy and, for $x_{\rm C}>0$,
\begin{align*}
 \Pp_{\bm B}(\mathfrak C_n^{\circ}>x_{\rm C}\mid\widehat{\mathcal W}_n)
 &\le
 \sum_{v=1}^{L_{\cG}}
 \Pp_{\bm B}\{\mathfrak c_{\rm C}(P_{n,v}^{\circ})>x_{\rm C}
 \mid\widehat{\mathcal W}_n\}\\
 &=
 L_{\cG}\left\{\frac12-\frac1\pi\arctan(x_{\rm C})\right\}.
\end{align*}
With
\[
 c_{\alpha,L_{\cG}}^{\rm C}=\cot\left(\frac{\pi\alpha}{2L_{\cG}}\right),
\]
the right-hand side is at most $\alpha/2$.  Hence the conditional
$(1-\alpha)$ joint critical value is bounded by $c_{\alpha,L_{\cG}}^{\rm C}$, and the joint
test is consistent.

It remains to verify the simulated two-batch implementation under the
alternative.  Let $N_{{\rm sim},1}=N_{{\rm sim},1}(n)\to\infty$ and $N_{{\rm sim},2}=N_{{\rm sim},2}(n)\to\infty$, and let
$\mathfrak E_{n,N_{{\rm sim},1}}^{\rm sim}$ be defined as in the proof of
Lemma~\ref{lem:Monte-Carlo-reference}.  Its conditional DKW bound is uniform
in $n$, so
\begin{equation*}
 \mathfrak E_{n,N_{{\rm sim},1}}^{\rm sim}=o_{\Pp}(1).
\end{equation*}
Since $P_{n,v}\pto0$ for every $v$, uniformly over the finite ridge grid,
\begin{equation*}
 \widetilde P_{n,v,N_{{\rm sim},1}}(T_{n,a_v})\pto0,
 \qquad
 \mathfrak C_{n,N_{{\rm sim},1}}\pto\infty.
\end{equation*}
For a fresh reference draw, the comparison argument in
\eqref{eq:MC-reference-comparison} uses only the conditional probability
integral transform and therefore remains valid under the alternative.  In
particular, under the coupling using the same fresh bridge draw,
$\mathfrak C_{n,N_{{\rm sim},1}}^{\circ}-\mathfrak C_n^{\circ}$ converges to zero in
conditional probability.  Moreover, uniformly in the estimated weights,
\begin{equation*}
 \Pp_{\bm B}(\mathfrak C_n^{\circ}>x\mid\widehat{\mathcal W}_n)
 \le L_{\cG}\left\{\frac12-\frac1\pi\arctan(x)\right\},
 \qquad x>0.
\end{equation*}
Choose a finite $x_\alpha^{\rm ref}>0$ such that the right-hand side at $x_\alpha^{\rm ref}$ is
smaller than $\alpha/4$.  Under the coupling using the same fresh draw,
\begin{align*}
 &\Pp_{\bm B}\{\mathfrak C_{n,N_{{\rm sim},1}}^{\circ}>x_\alpha^{\rm ref}+1
 \mid\widehat{\mathcal W}_n,\text{ first batch}\}\\
 &\quad\le
 \Pp_{\bm B}(\mathfrak C_n^{\circ}>x_\alpha^{\rm ref}
 \mid\widehat{\mathcal W}_n)
 +\Pp_{\bm B}\{|\mathfrak C_{n,N_{{\rm sim},1}}^{\circ}-\mathfrak C_n^{\circ}|>1
 \mid\widehat{\mathcal W}_n,\text{ first batch}\}.
\end{align*}
The first term is below $\alpha/4$, and the second converges to zero in
probability.  Hence, with probability tending to one over the data and first
batch, the left-hand side is below $\alpha/2$.
On the event in \eqref{eq:DKW-joint-bridge} with $y=\alpha/4$, the empirical
second-batch cdf at $x_\alpha^{\rm ref}+1$ is then at least $1-3\alpha/4>1-\alpha$.
Since the probability of the complementary DKW event tends to zero as
$N_{{\rm sim},2}\to\infty$, the definition in \eqref{eq:two-batch-critical-value} gives
\begin{equation*}
 \Pp\{\widetilde c_{n,C,1-\alpha}^{(N_{{\rm sim},1},N_{{\rm sim},2})}\le x_\alpha^{\rm ref}+1\}
 \longrightarrow1.
\end{equation*}
Thus
$\widetilde c_{n,C,1-\alpha}^{(N_{{\rm sim},1},N_{{\rm sim},2})}=O_{\Pp}(1)$.  Consequently,
\begin{equation*}
 \Pp\{\mathfrak C_{n,N_{{\rm sim},1}}>
 \widetilde c_{n,C,1-\alpha}^{(N_{{\rm sim},1},N_{{\rm sim},2})}\}\longrightarrow1,
\end{equation*}
which proves consistency of direct Monte Carlo bridge calibration without
invoking a null-only level statement.

For localization, recall the common shape $\mathfrak h_{\rm cp}$ in
\eqref{eq:common-change-shape}.  For $t\le\vartheta$ and
$t\ge\vartheta$, respectively,
\[
 \mathfrak h_{\rm cp}(t,\vartheta)=\frac{t(1-\vartheta)^2}{1-t},
 \qquad
 \mathfrak h_{\rm cp}(t,\vartheta)=\frac{\vartheta^2(1-t)}{t}.
\]
Hence
\[
 \partial_t\mathfrak h_{\rm cp}(t,\vartheta)
 =
 \begin{cases}
 (1-\vartheta)^2(1-t)^{-2}>0,&t<\vartheta,\\
 -\vartheta^2t^{-2}<0,&t>\vartheta.
 \end{cases}
\]
Thus $\mathfrak h_{\rm cp}(\cdot,\vartheta)$ has the unique maximum at $t=\vartheta$.  For every
$\rho_{\rm pos}>0$, compactness and $\vartheta_n\to\vartheta$ imply that there are
$c_{\rm pos}>0$ and $n_{\rm pos}$ such that
\[
 \mathfrak h_{\rm cp}(\vartheta_n,\vartheta_n)
 -
 \sup_{\substack{t\in[\epsilon,1-\epsilon]\\
 |t-\vartheta_n|\ge\rho_{\rm pos}/2}}
 \mathfrak h_{\rm cp}(t,\vartheta_n)
 \ge c_{\rm pos},
 \qquad n\ge n_{\rm pos}.
\]
By \eqref{eq:uniform-drift-main},
\[
 \Pp\left[
 \max_{a\in\cG}\max_{k\in\cK_{n,\epsilon}}
 \left|
 \frac{V_{n,a}(k)}{n\mathfrak q_a(d_n)}
 -\mathfrak h_{\rm cp}(t_k,\vartheta_n)
 \right|<\frac{c_{\rm pos}}{3}
 \right]\longrightarrow1.
\]
On this event,
\[
 \max_{a\in\cG}
 \left|\frac{\widehat\tau_{n,a}}{n}-\vartheta_n\right|
 <\frac{\rho_{\rm pos}}{2}.
\]
Since $|\vartheta_n-\vartheta|<\rho_{\rm pos}/2$ eventually,
\[
 \Pp\left
 \{\max_{a\in\cG}
 \left|\widehat\tau_{n,a}/n-\vartheta\right|
 \ge\rho_{\rm pos}\right\}\longrightarrow0,
\]
which proves \eqref{eq:location-consistency}.

\section{Proofs for multiple-change estimation}
\label{app:multiple-proofs}

\subsection{Proof of Lemma~\ref{lem:wbs-coverage}}

The number of admissible endpoint pairs is
\begin{equation*}
 |\mathbb I_n^{\rm WBS}|
 =\sum_{s=0}^{n-2}(n-s-1)=\frac{n(n-1)}2.
\end{equation*}
For $D\ge\Delta_n$,
\begin{equation*}
 \left\lfloor\frac{2D}{3}\right\rfloor
 -\left\lceil\frac{D}{3}\right\rceil+1
 \ge\frac D3-1\ge\frac D4
\end{equation*}
for all sufficiently large $n$.  Hence, uniformly in $b$,
\begin{equation*}
 |\mathcal L_{b,n}^{\rm WBS}|\ge\frac{\Delta_n}{4},
 \qquad
 |\mathcal R_{b,n}^{\rm WBS}|\ge\frac{\Delta_n}{4}.
\end{equation*}
For one random interval,
\begin{align*}
 p_{b,n}^{\rm iso}
 &=\Pp\left\{
 S_{r_{\rm W}}^{\rm WBS}\in\mathcal L_{b,n}^{\rm WBS},\ 
 E_{r_{\rm W}}^{\rm WBS}\in\mathcal R_{b,n}^{\rm WBS}
 \right\}\\
 &=\frac{|\mathcal L_{b,n}^{\rm WBS}|
 |\mathcal R_{b,n}^{\rm WBS}|}{|\mathbb I_n^{\rm WBS}|}
 \ge\frac18\left(\frac{\Delta_n}{n}\right)^2.
\end{align*}
Independence of the $M_n^{\rm WBS}$ draws and a union bound give
\begin{align*}
 \Pp\{(\mathcal C_n^{\rm WBS})^c\}
 &\le\sum_{b=1}^{J_{\rm cp}}(1-p_{b,n}^{\rm iso})^{M_n^{\rm WBS}}\\
 &\le J_{\rm cp}\left\{1-\frac18
 \left(\frac{\Delta_n}{n}\right)^2\right\}^{M_n^{\rm WBS}},
\end{align*}
which is \eqref{eq:wbs-coverage-probability} with $c_{\rm W}=1/8$.

On $\mathcal C_n^{\rm WBS}$, select
\begin{equation*}
 s_{b,n}^{\rm iso}\in\mathcal L_{b,n}^{\rm WBS},
 \qquad
 e_{b,n}^{\rm iso}\in\mathcal R_{b,n}^{\rm WBS}.
\end{equation*}
Then
\begin{align*}
 s_{b,n}^{\rm iso}-\tau_{b-1,n}
 &\ge\left\lceil\frac{\Delta_{b,n}^{-}}3\right\rceil
 \ge\frac{\Delta_n}{3},\\
 \tau_{b,n}-s_{b,n}^{\rm iso}
 &\ge\Delta_{b,n}^{-}
 -\left\lfloor\frac{2\Delta_{b,n}^{-}}3\right\rfloor
 \ge\frac{\Delta_n}{3}-1,\\
 e_{b,n}^{\rm iso}-\tau_{b,n}
 &\ge\left\lceil\frac{\Delta_{b,n}^{+}}3\right\rceil
 \ge\frac{\Delta_n}{3},\\
 \tau_{b+1,n}-e_{b,n}^{\rm iso}
 &\ge\Delta_{b,n}^{+}
 -\left\lfloor\frac{2\Delta_{b,n}^{+}}3\right\rfloor
 \ge\frac{\Delta_n}{3}-1.
\end{align*}
Thus $I_{b,n}^{\rm iso}=(s_{b,n}^{\rm iso},e_{b,n}^{\rm iso}]$
contains only $\tau_{b,n}$ and satisfies
\eqref{eq:wbs-isolation-margins}.

\subsection{Proof of Lemma~\ref{lem:local-score-envelope}}

Write $\mathcal J_n=\mathcal J_{p,N}$.  From
\eqref{eq:numerical-projection-operator}, Sobolev embedding, and
$\chi_{p,N}\to0$,
\begin{equation*}
 \norm{\mathcal J_nf}
 \le\norm{\Pi_pf}+\norm{(\mathcal J_n-\Pi_p)f}
 \le C\norm f_{\Hs},
 \qquad f\in\Hs.
\end{equation*}
Consequently,
\begin{equation*}
 \norm{\mathcal J_ne_0}_{L^\nu(\cH)}\le C,
 \qquad
 \sum_{M=1}^\infty
 \norm{\mathcal J_n(e_0-e_{0,M})}_{L^\nu(\cH)}
 \le C\sum_{M=1}^\infty\mathfrak d_\nu(M)<\infty.
\end{equation*}
Lemma~\ref{lem:Bernoulli-maximal}, with $\nu_0=\nu$, therefore yields
\begin{equation}
 \sup_{s_0\in\mathbb Z}
 \left\|\max_{1\le u\le r}
 \norm{\sum_{i=s_0+1}^{s_0+u}\beps_{i,p,N}}\right\|_{L^\nu}
 \le C\sqrt r,
 \qquad r\ge1.
 \label{eq:projected-partial-max-nu}
\end{equation}
For $0\le s<n$ and
$0\le j_{\rm dyad}\le\lceil\log_2n\rceil$, define
\begin{equation*}
 \mathfrak P_{s,j_{\rm dyad}}
 =2^{-j_{\rm dyad}/2}
 \max_{1\le u\le2^{j_{\rm dyad}}\wedge(n-s)}
 \norm{\sum_{i=s+1}^{s+u}\beps_{i,p,N}}.
\end{equation*}
Equation~\eqref{eq:projected-partial-max-nu} gives
\begin{equation*}
 \sup_{s,j_{\rm dyad}}\E\mathfrak P_{s,j_{\rm dyad}}^\nu\le C.
\end{equation*}
For $1\le r\le n-s$, choose $j_{\rm dyad}$ so that
$2^{j_{\rm dyad}-1}<r\le2^{j_{\rm dyad}}$.  Then
\begin{equation*}
 \frac1{\sqrt r}
 \norm{\sum_{i=s+1}^{s+r}\beps_{i,p,N}}
 \le\sqrt2\,\mathfrak P_{s,j_{\rm dyad}}.
\end{equation*}
Since the number of pairs $(s,j_{\rm dyad})$ is at most
$n\{2+\log_2n\}$, Markov's inequality and a union bound imply
\begin{equation}
 \Pp\left\{
 \max_{0\le s<e\le n}
 \frac{\norm{\sum_{i=s+1}^{e}\beps_{i,p,N}}}{\sqrt{e-s}}>x
 \right\}
 \le Cn(1+\log n)x^{-\nu}.
 \label{eq:all-normalized-blocks}
\end{equation}
Moreover,
\begin{align*}
 \norm{\mathsf Z_I(k)}
 &\le
 \sqrt{\frac{e-k}{n_I}}
 \frac{\norm{\sum_{i=s+1}^{k}\beps_{i,p,N}}}{\sqrt{k-s}}
 +\sqrt{\frac{k-s}{n_I}}
 \frac{\norm{\sum_{i=k+1}^{e}\beps_{i,p,N}}}{\sqrt{e-k}}\\
 &\le
 \frac{\norm{\sum_{i=s+1}^{k}\beps_{i,p,N}}}{\sqrt{k-s}}
 +\frac{\norm{\sum_{i=k+1}^{e}\beps_{i,p,N}}}{\sqrt{e-k}}.
\end{align*}
Combining this display with \eqref{eq:all-normalized-blocks} proves
\eqref{eq:local-score-tail}.  Taking
$x=y(n\log n)^{1/\nu}$ gives
\begin{equation*}
 \limsup_{n\to\infty}
 \Pp\left\{
 \frac{\mathfrak Z_n^{\rm WBS}}{(n\log n)^{1/\nu}}>2y
 \right\}
 \le Cy^{-\nu},
\end{equation*}
which proves \eqref{eq:local-score-order}.

\subsection{Proof of Lemma~\ref{lem:local-ridge-decomposition}}

Lemmas~\ref{prop:alternative-DB} and~\ref{prop:weight-consistency}, together
with $\mathfrak e_{{\rm cp},n}\to0$, imply
\begin{equation*}
 \Pp(\mathcal E_n^{\rm loc})\longrightarrow1,
\end{equation*}
where, for a sufficiently large constant $C_A$,
\begin{equation*}
 \mathcal E_n^{\rm loc}
 =\left\{
 \begin{array}{c}
 \mathfrak t_\Omega/2\le\widehat{\mathfrak t}_n\le2\mathfrak t_\Omega,\\[2pt]
 \min_{a\in\cG}\widehat A_{2,n}(a)
 \ge\frac12\min_{a\in\cG}A_2(a),\\[2pt]
 \max_{a\in\cG}\{\widehat A_{1,n}(a)+\widehat A_{2,n}(a)\}
 \le C_A
 \end{array}
 \right\}.
\end{equation*}
On $\mathcal E_n^{\rm loc}$,
\begin{equation}
 \frac{1}{2(1+a_+)\mathfrak t_\Omega}\norm v^2
 \le\ip v{\widehat{\mathcal R}_{n,a}v}
 \le\frac{2}{a_-\mathfrak t_\Omega}\norm v^2,
 \qquad a\in\cG.
 \label{eq:multiple-ridge-quadratic-bounds}
\end{equation}
The bounds on $\widehat A_{2,n}(a)$ and
\eqref{eq:multiple-ridge-quadratic-bounds} yield constants
$0<c_{\rm R}<C_{\rm R}<\infty$ for which
\eqref{eq:local-signal-norm-equivalence} holds.

The local contrast decomposes exactly as
\begin{equation}
 \iota_p\{\sqrt{N_I(k)}\widehat\bDelta_I(k)\}
 =\mathsf M_I(k)+\iota_p\mathsf Z_I(k).
 \label{eq:wbs-exact-local-contrast}
\end{equation}
Because $J_{\rm cp}$ is fixed,
\begin{align}
 \norm{\mathsf M_I(k)}
 &\le\sqrt{N_I(k)}
 \sum_{\tau_{b,n}\in\mathcal T_I^{\rm cp}}\norm{d_{b,n}}
 \le C\sqrt{n_I}\,d_I^{\max}.
 \label{eq:wbs-mean-contrast-bound}
\end{align}
For $a\in\cG$, write
\begin{equation*}
 \mathcal Q_{I,a}(k)
 =\frac{\ip{\mathsf M_I(k)}
 {\widehat{\mathcal R}_{n,a}\mathsf M_I(k)}}
 {\{2\widehat A_{2,n}(a)\}^{1/2}}.
\end{equation*}
Equations \eqref{eq:wbs-exact-local-contrast} and
\eqref{eq:multiple-ridge-quadratic-bounds} give
\begin{align*}
 D_{I,a}(k)-\mathcal Q_{I,a}(k)
 &=\frac{2\ip{\mathsf M_I(k)}
 {\widehat{\mathcal R}_{n,a}\mathsf Z_I(k)}
 +\ip{\mathsf Z_I(k)}
 {\widehat{\mathcal R}_{n,a}\mathsf Z_I(k)}
 -\widehat A_{1,n}(a)}
 {\{2\widehat A_{2,n}(a)\}^{1/2}},\\
 |D_{I,a}(k)-\mathcal Q_{I,a}(k)|
 &\le C\left\{1+\norm{\mathsf Z_I(k)}^2
 +\norm{\mathsf M_I(k)}\norm{\mathsf Z_I(k)}\right\}.
\end{align*}
Using
\begin{equation*}
 \left|\max_{a\in\cG}x_a-\max_{a\in\cG}y_a\right|
 \le\max_{a\in\cG}|x_a-y_a|
\end{equation*}
and then \eqref{eq:wbs-mean-contrast-bound} proves
\eqref{eq:wbs-uniform-score-exact} and
\eqref{eq:wbs-uniform-score-approximation}.

Suppose that $I$ contains only $\tau_I$, with retained jump $d_I$.  Direct
summation gives
\begin{equation*}
 \mathsf M_I(k)
 =\sqrt{n_I}\,
 \frac{g(x_{I,k},\theta_I)}
 {\{x_{I,k}(1-x_{I,k})\}^{1/2}}d_I.
\end{equation*}
Therefore
\begin{align*}
 \mathcal Q_{I,a}(k)
 &=n_I\mathfrak h_{\rm cp}(x_{I,k},\theta_I)
 \frac{\widehat{\mathfrak q}_{n,a}(d_I)}
 {\{2\widehat A_{2,n}(a)\}^{1/2}},\\
 |\mathrm{Rem}_{I,a}(k)|
 &\le C\left\{1+\norm{\mathsf Z_I(k)}^2
 +\sqrt{n_I}\norm{d_I}\norm{\mathsf Z_I(k)}\right\}.
\end{align*}
This proves \eqref{eq:local-ridge-decomposition}--
\eqref{eq:local-ridge-remainder}.  Equation
\eqref{eq:local-ridge-signal-equivalence} follows from
\eqref{eq:local-signal-norm-equivalence} with $v=d_I$.

Finally,
\begin{equation*}
 \mathfrak h_{\rm cp}(\theta,\theta)
 -\mathfrak h_{\rm cp}(x,\theta)
 =\begin{cases}
 (1-\theta)(\theta-x)/(1-x),&x\le\theta,\\
 \theta(x-\theta)/x,&x\ge\theta.
 \end{cases}
\end{equation*}
Hence, for
$\theta\in[\varepsilon_{\rm gap},1-\varepsilon_{\rm gap}]$,
\begin{equation*}
 \mathfrak h_{\rm cp}(\theta,\theta)
 -\mathfrak h_{\rm cp}(x,\theta)
 \ge\varepsilon_{\rm gap}|x-\theta|,
\end{equation*}
which is \eqref{eq:local-shape-linear-gap}.

\subsection{Proof of Lemma~\ref{lem:wbs-deterministic-geometry}}

For $a\in\cG$, define
\begin{equation*}
 \norm{v}_{n,a}^2
 =\frac{\ip v{\widehat{\mathcal R}_{n,a}v}}
 {\{2\widehat A_{2,n}(a)\}^{1/2}}.
\end{equation*}
Fix $I=(s,e]$, write $L=n_I$ and $t=k-s$, and set
\begin{equation*}
 \overline\mu_I^{\rm ret}
 =L^{-1}\sum_{i=s+1}^{e}\mu_{i,n}^{\rm ret},
 \qquad
 \mathsf C_I(t)=\sum_{i=s+1}^{s+t}
 (\mu_{i,n}^{\rm ret}-\overline\mu_I^{\rm ret}).
\end{equation*}
Then
\begin{equation}
 \mathsf M_I(s+t)
 =-\left\{\frac{L}{t(L-t)}\right\}^{1/2}\mathsf C_I(t).
 \label{eq:wbs-cumulative-signal}
\end{equation}
Between two successive changes, $\mathsf C_I(t)=u+tv$.  Hence
\begin{align}
 Q_{I,a}(t)
 &:=\norm{\mathsf M_I(s+t)}_{n,a}^2\notag\\
 &=\frac{L\norm{u+tv}_{n,a}^2}{t(L-t)}\notag\\
 &=\frac{\norm u_{n,a}^2}{t}
 +\frac{\norm{u+Lv}_{n,a}^2}{L-t}
 -L\norm v_{n,a}^2,\notag\\
 Q_{I,a}''(t)
 &=\frac{2\norm u_{n,a}^2}{t^3}
 +\frac{2\norm{u+Lv}_{n,a}^2}{(L-t)^3}\ge0.
 \label{eq:wbs-segment-second-derivative}
\end{align}
Therefore, by \eqref{eq:wbs-cumulative-signal},
\begin{equation*}
 \cS_I^{\rm sig}(s+t)=\max_{a\in\cG}Q_{I,a}(t)
\end{equation*}
is convex on every change-free segment.  Its continuous extension satisfies
\begin{equation*}
 \cS_I^{\rm sig}(s)=\cS_I^{\rm sig}(e)=0.
\end{equation*}
We next rule out a positive flat segment before locating the maximizers.  For a
fixed $a$, if $Q_{I,a}$ is constant on a nonempty open subinterval of a
change-free segment, then \eqref{eq:wbs-segment-second-derivative} gives
\begin{equation*}
 \norm u_{n,a}=\norm{u+Lv}_{n,a}=0.
\end{equation*}
On $\mathcal E_n^{\rm loc}$, $\widehat{\mathcal R}_{n,a}$ is strictly positive
definite and $\widehat A_{2,n}(a)>0$, so $u=0$, $u+Lv=0$, and hence $v=0$.
Consequently, that constant must be zero.  If, on an open subinterval $J$ of
the same change-free segment, $\max_{a\in\cG}Q_{I,a}$ were equal to a
positive constant $M$, then the finitely many relatively closed sets
\begin{equation*}
 \{t\in J:Q_{I,a}(t)=M\},\qquad a\in\cG,
\end{equation*}
would cover $J$.  By the Baire category theorem (equivalently, the finite
closed-cover property for a real interval), at least one of them contains a
nonempty relatively open subinterval, forcing the corresponding $Q_{I,a}$ to
be the positive constant
$M$, a contradiction.  Thus $\cS_I^{\rm sig}$ has no positive flat interval.

A convex function on a closed change-free segment attains its maximum at an
endpoint; if an interior point also attained a positive global maximum, the
convex slope inequalities would force a positive flat interval between that
point and an endpoint.  Since interval endpoints are either true changes or
$s,e$, and the latter have zero signal, it follows that
\begin{equation}
 \cS_I^{\rm sig,max}>0
 \quad\Longrightarrow\quad
 \argmax_{k\in\cK_I}\cS_I^{\rm sig}(k)
 \subseteq\mathcal T_I^{\rm cp}.
 \label{eq:wbs-signal-max-at-change}
\end{equation}
Moreover,
\begin{align*}
 \norm{\mathsf M_I(k)}
 &\le \sqrt{N_I(k)}
 \sum_{\tau_{b,n}\in\mathcal T_I^{\rm cp}}\norm{d_{b,n}}\\
 &\le C\sqrt{n_I}\,d_I^{\max}.
\end{align*}
Combining this bound with \eqref{eq:local-signal-norm-equivalence} proves
\eqref{eq:wbs-deterministic-upper}.

It remains to prove the uniform gap.  Fix $c_0>0$ and consider intervals for
which
\begin{equation}
 \cS_I^{\rm sig,max}\ge c_0\Delta_n(d_I^{\max})^2.
 \label{eq:wbs-gap-signal-scale}
\end{equation}
By \eqref{eq:wbs-deterministic-upper} and $\Delta_n\ge c_\Delta n$,
\begin{equation*}
 0<c_L:=c_0/C
 \le \mathfrak l_I:=\frac{n_I}{\Delta_n}
 \le c_\Delta^{-1}=:C_L.
\end{equation*}
Enumerate the changes in $I$ as
\begin{equation*}
 \mathcal T_I^{\rm cp}
 =\{t_{1,I}<\cdots<t_{J_I^{\rm loc},I}\},
 \qquad
 t_{j,I}=\tau_{b_j,n},
 \qquad
 d_{j,I}=d_{b_j,n}.
\end{equation*}
Since $J_I^{\rm loc}\le J_{\rm cp}$, a contradiction sequence has a
subsequence on which $J_I^{\rm loc}=J_{\rm loc}$ is fixed.  Define
\begin{equation*}
 \theta_{j,I}=\frac{t_{j,I}-s}{\Delta_n},
 \qquad
 r_{j,I}=\frac{d_{j,I}}{d_I^{\max}}.
\end{equation*}
Assumption~\ref{ass:multiple-configuration} and
\eqref{eq:uniform-functional-resolution} imply
\begin{equation}
 \begin{gathered}
 c_L\le\mathfrak l_I\le C_L,
 \qquad
 0\le\theta_{1,I}<\cdots<\theta_{J_{\rm loc},I}\le\mathfrak l_I,\\
 \theta_{j+1,I}-\theta_{j,I}\ge1,
 \qquad
 0<c_d\le\norm{r_{j,I}}\le1.
 \end{gathered}
 \label{eq:wbs-normalized-parameter-bounds}
\end{equation}
For $a\in\cG$, let
\begin{equation*}
 \bm H_{I,a}^{\rm gram}
 =\left[
 \frac{\ip{r_{j,I}}{\widehat{\mathcal R}_{n,a}r_{l,I}}}
 {\{2\widehat A_{2,n}(a)\}^{1/2}}
 \right]_{j,l=1}^{J_{\rm loc}}.
\end{equation*}
On $\mathcal E_n^{\rm loc}$,
\begin{equation}
 \max_{a\in\cG}\norm{\bm H_{I,a}^{\rm gram}}_{\op}\le C_G,
 \qquad
 \min_{a\in\cG}\min_{1\le j\le J_{\rm loc}}
 (\bm H_{I,a}^{\rm gram})_{jj}\ge c_G>0.
 \label{eq:wbs-normalized-gram-bounds}
\end{equation}
The normalized directions $r_{j,I}$ live in growing subspaces and need not
possess strongly convergent subsequences.  Such convergence is unnecessary:
the deterministic criteria depend on these directions only through the finite
Gram matrices above.  To make this explicit, define
\begin{equation*}
 \gamma_{j,I}(x)
 =(x-\theta_{j,I})_+-
 \frac{x}{\mathfrak l_I}(\mathfrak l_I-\theta_{j,I}),
 \qquad
 \bm\gamma_I(x)=\{\gamma_{j,I}(x):1\le j\le J_{\rm loc}\}^{\top}.
\end{equation*}
At the normalized grid points $x=t/\Delta_n$, direct summation of the retained
step mean gives
\begin{equation*}
 \mathsf C_I(t)
 =\Delta_n d_I^{\max}
 \sum_{j=1}^{J_{\rm loc}}\gamma_{j,I}(x)r_{j,I}.
\end{equation*}
Consequently, for $0<x<\mathfrak l_I$,
\begin{equation}
 \mathfrak F_{I,a}(x)
 :=\frac{Q_{I,a}(\Delta_nx)}
 {\Delta_n(d_I^{\max})^2}
 =\frac{\mathfrak l_I}{x(\mathfrak l_I-x)}
 \bm\gamma_I(x)^{\top}\bm H_{I,a}^{\rm gram}\bm\gamma_I(x),
 \label{eq:wbs-finite-gram-representation}
\end{equation}
where the right-hand side defines the continuous extension between grid points
and is set equal to zero at $x=0,\mathfrak l_I$.  Put
\begin{equation*}
 \mathfrak F_I(x)=\max_{a\in\cG}\mathfrak F_{I,a}(x).
\end{equation*}

The finite-dimensional tuple
\begin{equation*}
 \left(\mathfrak l_I,\theta_{1,I},\ldots,\theta_{J_{\rm loc},I},
 \{\bm H_{I,a}^{\rm gram}:a\in\cG\}\right)
\end{equation*}
lies in a closed and bounded subset of a Euclidean space: the location and
length constraints are given by
\eqref{eq:wbs-normalized-parameter-bounds}, while every Gram matrix is
symmetric positive semidefinite and satisfies
\eqref{eq:wbs-normalized-gram-bounds}.  Hence every contradiction sequence has
a subsequence on which this tuple converges.  Each limiting Gram matrix is
positive semidefinite and induces the quadratic seminorm
$\norm{z}_a^2=z^{\top}\bm H_a^{\rm gram}z$ on $\mathbb R^{J_{\rm loc}}$.
By \eqref{eq:wbs-signal-max-at-change},
\begin{equation*}
 \mathfrak F_I^{\max}:=\max_{0\le x\le\mathfrak l_I}\mathfrak F_I(x)
 =\frac{\cS_I^{\rm sig,max}}
 {\Delta_n(d_I^{\max})^2}\ge c_0.
\end{equation*}
Formula \eqref{eq:wbs-finite-gram-representation} shows that, on every
compact subinterval of a knot-free segment, the ridge-specific criteria and
their derivatives depend continuously on this finite-dimensional tuple.  Their
finite maximum $\mathfrak F_I$ is therefore jointly continuous and uniformly
Lipschitz there.  No continuity of its derivative at a ridge-switching point is
used.

Set $\bar c=1/4$.  If the desired gap failed, there would be admissible
intervals $I$ and points $x_I$ satisfying
\begin{equation*}
 x_I\notin\{\theta_{1,I},\ldots,\theta_{J_{\rm loc},I}\}
\end{equation*}
such that
\begin{equation}
 \frac{\mathfrak F_I^{\max}-\mathfrak F_I(x_I)}
 {\operatorname{dist}(x_I,
 \{\theta_{1,I},\ldots,\theta_{J_{\rm loc},I}\})\wedge\bar c}
 \longrightarrow0.
 \label{eq:wbs-gap-contradiction-ratio}
\end{equation}
Pass to a subsequence on which the finite-dimensional tuples above and $x_I$
converge.  Denote the limiting ridge-specific criteria by $\mathfrak F_a$ and
their maximum by $\mathfrak F$.

Suppose first that the limit $x_0$ is not a knot.  Then
\begin{equation*}
 \mathfrak F(x_0)=\mathfrak F^{\max}\ge c_0,
 \qquad
 \mathfrak F(0)=\mathfrak F(\mathfrak l)=0.
\end{equation*}
Thus $x_0$ lies in the interior of a knot-free segment $(\alpha,\beta)$.
For $\alpha<x<x_0<y<\beta$, convexity gives
\begin{equation*}
 0\le
 \frac{\mathfrak F(x_0)-\mathfrak F(x)}{x_0-x}
 \le
 \frac{\mathfrak F(y)-\mathfrak F(x_0)}{y-x_0}
 \le0.
\end{equation*}
Hence $\mathfrak F\equiv\mathfrak F^{\max}$ on a nonempty open interval.

Suppose next that $x_I$ approaches a knot $\theta_{j,I}$.  Along one side of
that knot, put
\begin{equation*}
 h_I=|x_I-\theta_{j,I}|<\bar c.
\end{equation*}
Equations \eqref{eq:wbs-gap-signal-scale} and
\eqref{eq:wbs-gap-contradiction-ratio} give
$\mathfrak F_I(x_I)\ge c_0/2$ eventually.  Since
\begin{equation*}
 N_I(s+\Delta_nx)
 =\Delta_n\frac{x(\mathfrak l_I-x)}{\mathfrak l_I}
 \le\Delta_n\{x\wedge(\mathfrak l_I-x)\},
\end{equation*}
\eqref{eq:local-signal-norm-equivalence} and the bounded jump range yield
\begin{equation}
 \mathfrak F_I(x)\le C\{x\wedge(\mathfrak l_I-x)\}.
 \label{eq:wbs-normalized-boundary-envelope}
\end{equation}
Equation~\eqref{eq:wbs-normalized-boundary-envelope} implies
\begin{equation*}
 \min\{x_I,\mathfrak l_I-x_I,\theta_{j,I},
 \mathfrak l_I-\theta_{j,I}\}\ge c_{\rm bd}>0.
\end{equation*}
The one-sided derivatives are therefore uniformly bounded near $x_I$ and
$\theta_{j,I}$.  It follows that
\begin{align*}
 0\le \mathfrak F_I^{\max}-\mathfrak F_I(\theta_{j,I})
 &\le \mathfrak F_I^{\max}-\mathfrak F_I(x_I)+Ch_I\longrightarrow0,\\
 \frac{\mathfrak F_I(x_I)-\mathfrak F_I(\theta_{j,I})}{h_I}
 &\ge-\frac{\mathfrak F_I^{\max}-\mathfrak F_I(x_I)}{h_I}
 \longrightarrow0.
\end{align*}
The limiting knot is a global maximizer.  If $\sigma\in\{-1,1\}$ points
into the adjacent segment, then
\begin{equation*}
 0\ge
 \lim_{h\downarrow0}
 \frac{\mathfrak F(\theta+\sigma h)-\mathfrak F(\theta)}{h}
 \ge0.
\end{equation*}
For every $y$ in that segment, the supporting-line inequality gives
\begin{equation*}
 \mathfrak F(y)
 \ge \mathfrak F(\theta)
 +0\,|y-\theta|
 =\mathfrak F^{\max}.
\end{equation*}
Thus $\mathfrak F\equiv\mathfrak F^{\max}$ on a nonempty open interval in
this case as well.

Let $\mathcal O$ be an open interval on which
$\mathfrak F\equiv\mathfrak F^{\max}$.  Since
\begin{equation*}
 \mathcal O=\bigcup_{a\in\cG}
 \{x\in\mathcal O:\mathfrak F_a(x)=\mathfrak F^{\max}\},
\end{equation*}
at least one of these finitely many closed sets has nonempty relative
interior.  Shrinking that interior to a knot-free interval if necessary, there
is an $a_0\in\cG$ for which $\mathfrak F_{a_0}$ is constant there.  On that
segment, $\bm\gamma(x)=\bm u+x\bm v$ for fixed vectors
$\bm u,\bm v\in\mathbb R^{J_{\rm loc}}$.  The limiting form of
\eqref{eq:wbs-segment-second-derivative}, equivalently the second derivative of
\eqref{eq:wbs-finite-gram-representation}, gives
\begin{equation*}
 0=\mathfrak F_{a_0}''(x)
 =\frac{2\norm{\bm u}_{a_0}^2}{x^3}
 +\frac{2\norm{\bm u+\mathfrak l\bm v}_{a_0}^2}
 {(\mathfrak l-x)^3}.
\end{equation*}
Therefore
\begin{equation*}
 \norm{\bm u}_{a_0}
 =\norm{\bm u+\mathfrak l\bm v}_{a_0}
 =\norm{\bm v}_{a_0}=0,
 \qquad
 \mathfrak F_{a_0}\equiv0,
\end{equation*}
contradicting $\mathfrak F^{\max}\ge c_0$.  Hence there is a constant
$c_1>0$ such that
\begin{equation*}
 \mathfrak F_I^{\max}-\mathfrak F_I(x)
 \ge c_1\left\{
 \operatorname{dist}(x,
 \{\theta_{1,I},\ldots,\theta_{J_{\rm loc},I}\})\wedge\bar c
 \right\}.
\end{equation*}
Taking $x=(k-s)/\Delta_n$ gives
\begin{align*}
 \cS_I^{\rm sig,max}-\cS_I^{\rm sig}(k)
 &\ge c_1(d_I^{\max})^2
 \left\{\operatorname{dist}(k,\mathcal T_I^{\rm cp})
 \wedge\bar c\Delta_n\right\},
\end{align*}
which proves \eqref{eq:wbs-deterministic-gap} with $c_2=\bar c$.

Finally, if $k-s\le r_{\rm bd}$ or $e-k\le r_{\rm bd}$, then
\begin{align*}
 N_I(k)&\le r_{\rm bd},\\
 \left\|
 \frac1{e-k}\sum_{i=k+1}^{e}\mu_{i,n}^{\rm ret}
 -\frac1{k-s}\sum_{i=s+1}^{k}\mu_{i,n}^{\rm ret}
 \right\|
 &\le\sum_{b=1}^{J_{\rm cp}}\norm{d_{b,n}}\le C.
\end{align*}
Consequently,
\begin{equation*}
 \norm{\mathsf M_I(k)}^2\le Cr_{\rm bd},
 \qquad
 \cS_I^{\rm sig}(k)\le C\norm{\mathsf M_I(k)}^2\le Cr_{\rm bd},
\end{equation*}
which is \eqref{eq:wbs-boundary-score}.

\subsection{Proof of Theorem~\ref{thm:multiple-consistency}}

By \eqref{eq:uniform-functional-resolution},
\begin{equation*}
 \max_{1\le b\le J_{\rm cp}}
 \norm{d_{b,n}-\delta_{b,n}}
 \le\mathfrak a_{p,N}(B_\delta)=o(\delta_{\min,n}).
\end{equation*}
Assumption~\ref{ass:multiple-configuration} therefore yields constants
$0<d_-<d_+<\infty$ such that, for all sufficiently large $n$,
\begin{equation}
 d_-\le\min_{1\le b\le J_{\rm cp}}\norm{d_{b,n}}
 \le\max_{1\le b\le J_{\rm cp}}\norm{d_{b,n}}\le d_+,
 \qquad
 \min_b\norm{d_{b,n}}\ge\frac{\delta_{\min,n}}2.
 \label{eq:retained-jump-lower-bound}
\end{equation}

Choose $C_{\rm loc}>0$ sufficiently
large in
\begin{equation*}
 \mathfrak R_n^{\rm WBS}
 =C_{\rm loc}\left\{
 \frac{1+(\mathfrak Z_n^{\rm WBS})^2}{\delta_{\min,n}^2}
 +\frac{\mathfrak Z_n^{\rm WBS}\sqrt n}{\delta_{\min,n}}\right\}.
\end{equation*}
Since $\mathfrak Z_n^{\rm WBS}=O_{\Pp}(\mathfrak v_n^{\rm WBS})$ and
$\mathfrak v_n^{\rm WBS}=o(\sqrt n)$,
\begin{align}
 (\mathfrak Z_n^{\rm WBS})^2&=o_{\Pp}(\Lambda_n),
 &\mathfrak Z_n^{\rm WBS}\sqrt n&=o_{\Pp}(\Lambda_n),
 \label{eq:wbs-random-noise-orders}\\
 \mathfrak R_n^{\rm WBS}
 &=O_{\Pp}\left\{
 \frac{(\mathfrak v_n^{\rm WBS})^2}{\delta_{\min,n}^2}
 +\frac{\mathfrak v_n^{\rm WBS}\sqrt n}{\delta_{\min,n}}
 \right\}
 =o_{\Pp}(\Lambda_n\wedge\Delta_n),
 \label{eq:wbs-random-radius-orders}\\
 \mathfrak Z_n^{\rm WBS}\sqrt{\mathfrak R_n^{\rm WBS}}
 &\le\frac12\{(\mathfrak Z_n^{\rm WBS})^2+\mathfrak R_n^{\rm WBS}\}
 =o_{\Pp}(\Lambda_n).\notag
\end{align}
Here
\begin{equation*}
 \frac{(\mathfrak v_n^{\rm WBS})^2}
 {\mathfrak v_n^{\rm WBS}\sqrt n}
 =\frac{\mathfrak v_n^{\rm WBS}}{\sqrt n}\longrightarrow0,
 \qquad
 \Lambda_n=o(\Delta_n\delta_{\min,n}^2)=o(\Delta_n),
\end{equation*}
and $0<c\le\delta_{\min,n}\le C$ by
Assumption~\ref{ass:multiple-configuration}.  In particular,
$\Lambda_n\to\infty$.

Choose a deterministic sequence $\eta_n\downarrow0$ such that
\begin{equation*}
 \Pp\left\{
 \begin{array}{c}
 1+(\mathfrak Z_n^{\rm WBS})^2+\mathfrak Z_n^{\rm WBS}\sqrt n\le\eta_n\Lambda_n,\\
 \mathfrak R_n^{\rm WBS}
 \le\eta_n(\Lambda_n\wedge\Delta_n),\\
 1+(\mathfrak Z_n^{\rm WBS})^2+\mathfrak Z_n^{\rm WBS}\sqrt{\mathfrak R_n^{\rm WBS}}
 \le\eta_n\Lambda_n
 \end{array}
 \right\}\longrightarrow1.
\end{equation*}
Define
\begin{equation*}
 \mathcal G_n^{\rm WBS}
 =\mathcal C_n^{\rm WBS}\cap\mathcal E_n^{\rm loc}
 \cap\left\{
 \begin{array}{c}
 1+(\mathfrak Z_n^{\rm WBS})^2+\mathfrak Z_n^{\rm WBS}\sqrt n\le\eta_n\Lambda_n,\\
 \mathfrak R_n^{\rm WBS}
 \le\eta_n(\Lambda_n\wedge\Delta_n),\\
 1+(\mathfrak Z_n^{\rm WBS})^2+\mathfrak Z_n^{\rm WBS}\sqrt{\mathfrak R_n^{\rm WBS}}
 \le\eta_n\Lambda_n
 \end{array}
 \right\}.
\end{equation*}
The coverage condition, Lemma~\ref{lem:local-ridge-decomposition}, and the
preceding displays give
\begin{equation}
 \Pp(\mathcal G_n^{\rm WBS})\longrightarrow1.
 \label{eq:wbs-main-high-probability-event}
\end{equation}
We henceforth work on $\mathcal G_n^{\rm WBS}$.

For the isolating interval
$I_{b,n}^{\rm iso}=(s_{b,n}^{\rm iso},e_{b,n}^{\rm iso}]$, write
\begin{equation*}
 L_{b,n}=|I_{b,n}^{\rm iso}|,
 \qquad
 \theta_{b,n}
 =\frac{\tau_{b,n}-s_{b,n}^{\rm iso}}{L_{b,n}}.
\end{equation*}
By \eqref{eq:wbs-isolation-margins},
\begin{align*}
 L_{b,n}\mathfrak h_{\rm cp}(\theta_{b,n},\theta_{b,n})
 &=\frac{(\tau_{b,n}-s_{b,n}^{\rm iso})
 (e_{b,n}^{\rm iso}-\tau_{b,n})}{L_{b,n}}\\
 &\ge\frac12\min\{\tau_{b,n}-s_{b,n}^{\rm iso},
 e_{b,n}^{\rm iso}-\tau_{b,n}\}\\
 &\ge c\Delta_n.
\end{align*}
Equations \eqref{eq:local-ridge-decomposition},
\eqref{eq:local-ridge-signal-equivalence}, and
\eqref{eq:retained-jump-lower-bound} imply, uniformly in $b$,
\begin{align}
 \cS_{I_{b,n}^{\rm iso}}^{\max}
 &\ge\cS_{I_{b,n}^{\rm iso}}(\tau_{b,n})\notag\\
 &\ge c\Delta_n\norm{d_{b,n}}^2
 -C\{1+(\mathfrak Z_n^{\rm WBS})^2+\sqrt n\norm{d_{b,n}}\mathfrak Z_n^{\rm WBS}\}\notag\\
 &\ge c_0\Delta_n\delta_{\min,n}^2>\Lambda_n.
 \label{eq:wbs-isolating-signal}
\end{align}
If $\mathcal T_I^{\rm cp}=\varnothing$, then $\mathsf M_I(k)=0$ and
\begin{equation}
 \cS_I^{\max}\le C(1+(\mathfrak Z_n^{\rm WBS})^2)<\Lambda_n.
 \label{eq:wbs-no-change-below-threshold}
\end{equation}

Consider a recursive call on $(s,e]$ and define the unselected changes in
that node by
\begin{equation*}
 \mathcal U_n(s,e)
 =\{b:s<\tau_{b,n}<e,\ \tau_{b,n}
 \text{ has not been selected}\}.
\end{equation*}
We prove recursively that
\begin{equation}
 \begin{gathered}
 I_{b,n}^{\rm iso}\subseteq(s,e],
 \qquad b\in\mathcal U_n(s,e),\\
 \operatorname{dist}(\tau_{j,n},\{s,e\})
 \le\mathfrak R_n^{\rm WBS}
 \quad\text{for every selected $\tau_{j,n}$ in $(s,e)$.}
 \end{gathered}
 \label{eq:wbs-recursion-invariant}
\end{equation}
Both statements hold at $(s,e]=(0,n]$.

Suppose first that $\mathcal U_n(s,e)=\varnothing$.  Every change in
$(s,e)$ is then previously selected, and the second invariant gives
\begin{equation*}
 \operatorname{dist}(\tau_{b,n},\{s,e\})
 \le\mathfrak R_n^{\rm WBS},
 \qquad \tau_{b,n}\in(s,e).
\end{equation*}
For a sampled interval $I=(S,E]\subseteq(s,e]$ with
$\cS_I^{\rm sig,max}>0$, choose
$\tau_{b,n}\in\mathcal T_I^{\rm cp}$ at which the signal criterion is
maximized.  Then
\begin{equation*}
 \min\{\tau_{b,n}-S,E-\tau_{b,n}\}
 \le\min\{\tau_{b,n}-s,e-\tau_{b,n}\}
 \le\mathfrak R_n^{\rm WBS}.
\end{equation*}
Equations~\eqref{eq:wbs-boundary-score} and
\eqref{eq:local-signal-norm-equivalence} yield
\begin{equation*}
 \cS_I^{\rm sig,max}\le C\mathfrak R_n^{\rm WBS},
 \qquad
 \sup_{k\in\cK_I}\norm{\mathsf M_I(k)}^2
 \le C\cS_I^{\rm sig,max}
 \le C\mathfrak R_n^{\rm WBS}.
\end{equation*}
The same bounds hold trivially when $\cS_I^{\rm sig,max}=0$.  Hence,
by \eqref{eq:wbs-uniform-score-exact},
\begin{align}
 \cS_I^{\max}
 &\le C\mathfrak R_n^{\rm WBS}
 +C\left\{1+(\mathfrak Z_n^{\rm WBS})^2
 +\mathfrak Z_n^{\rm WBS}\sqrt{\mathfrak R_n^{\rm WBS}}\right\}\\
 &\le C\eta_n\Lambda_n<\Lambda_n.
 \label{eq:wbs-empty-node-stops}
\end{align}
Thus a recursive node containing no unselected change stops.

Suppose now that $\mathcal U_n(s,e)\ne\varnothing$.  The first invariant and
\eqref{eq:wbs-isolating-signal} yield
\begin{equation}
 \cS_{I^\star}(k^\star)
 =\max_{r_{\rm W}\in\mathcal M_n^{\rm WBS}(s,e)}
 \cS_{I_{r_{\rm W}}^{\rm WBS}}^{\max}
 \ge c_0\Delta_n\delta_{\min,n}^2.
 \label{eq:wbs-selected-score-lower}
\end{equation}
Thus $I^\star$ contains at least one change by
\eqref{eq:wbs-no-change-below-threshold}.  Put
\begin{equation*}
 L_\star=|I^\star|,
 \qquad
 d_\star=d_{I^\star}^{\max}.
\end{equation*}
The uniform approximation gives
\begin{align*}
 \cS_{I^\star}^{\rm sig,max}
 &\ge c_0\Delta_n\delta_{\min,n}^2
 -C\{1+(\mathfrak Z_n^{\rm WBS})^2+\sqrt{L_\star}d_\star \mathfrak Z_n^{\rm WBS}\}\\
 &\ge c_1\Delta_n\delta_{\min,n}^2.
\end{align*}
Since \eqref{eq:wbs-deterministic-upper} gives
\begin{equation*}
 Cn d_\star^2
 \ge c_1\Delta_n\delta_{\min,n}^2,
\end{equation*}
we have $d_\star\ge c\delta_{\min,n}$.  Moreover,
\begin{equation*}
 \frac{\delta_{\min,n}^2}{d_\star^2}
 \ge\frac{\inf_n\delta_{\min,n}^2}{d_+^2}>0,
\end{equation*}
so
\begin{equation}
 \cS_{I^\star}^{\rm sig,max}
 \ge c_2\Delta_n d_\star^2.
 \label{eq:wbs-selected-signal-scale}
\end{equation}

Let $k_{\rm sig}^\star$ maximize $\cS_{I^\star}^{\rm sig}$.  Since $k^\star$
maximizes $\cS_{I^\star}$,
\begin{align*}
 \cS_{I^\star}^{\rm sig,max}
 -\cS_{I^\star}^{\rm sig}(k^\star)
 &\le
 |\cS_{I^\star}(k_{\rm sig}^\star)
 -\cS_{I^\star}^{\rm sig}(k_{\rm sig}^\star)|\\
 &\quad+
 |\cS_{I^\star}(k^\star)
 -\cS_{I^\star}^{\rm sig}(k^\star)|\\
 &\le C\{1+(\mathfrak Z_n^{\rm WBS})^2+\sqrt{L_\star}d_\star \mathfrak Z_n^{\rm WBS}\}.
\end{align*}
The right-hand side is $o(\Delta_nd_\star^2)$ by
\eqref{eq:wbs-random-noise-orders} and
\eqref{eq:wbs-selected-signal-scale}.  Therefore the truncated term in
\eqref{eq:wbs-deterministic-gap} cannot be active, and
\begin{align}
 \operatorname{dist}(k^\star,\mathcal T_{I^\star}^{\rm cp})
 &\le C\left\{
 \frac{1+(\mathfrak Z_n^{\rm WBS})^2}{d_\star^2}
 +\frac{\mathfrak Z_n^{\rm WBS}\sqrt{L_\star}}{d_\star}\right\}
 \le\mathfrak R_n^{\rm WBS}.
 \label{eq:wbs-selected-localization}
\end{align}
Choose $\tau_{b^\star,n}\in\mathcal T_{I^\star}^{\rm cp}$ attaining this
distance.

The change $\tau_{b^\star,n}$ has not been selected previously.  Otherwise,
the second invariant gives
\begin{equation*}
 \operatorname{dist}(\tau_{b^\star,n},\{s,e\})
 \le\mathfrak R_n^{\rm WBS}.
\end{equation*}
For $I^\star=(S^\star,E^\star]\subseteq(s,e]$,
\begin{align*}
 \min\{\tau_{b^\star,n}-S^\star,
 E^\star-\tau_{b^\star,n}\}
 &\le\mathfrak R_n^{\rm WBS},\\
 \min\{k^\star-S^\star,E^\star-k^\star\}
 &\le2\mathfrak R_n^{\rm WBS}.
\end{align*}
Hence \eqref{eq:wbs-boundary-score} and
\eqref{eq:wbs-uniform-score-exact} give
\begin{align*}
 \cS_{I^\star}(k^\star)
 &\le C\mathfrak R_n^{\rm WBS}
 +C\{1+(\mathfrak Z_n^{\rm WBS})^2+\mathfrak Z_n^{\rm WBS}\sqrt{\mathfrak R_n^{\rm WBS}}\}\\
 &\le C\eta_n\Lambda_n<\Lambda_n,
\end{align*}
contradicting \eqref{eq:wbs-selected-score-lower}.  Thus every successful
split is matched to a previously unselected change.

We next verify that \eqref{eq:wbs-recursion-invariant} is preserved.  For all
large $n$,
\begin{equation*}
 |k^\star-\tau_{b^\star,n}|
 \le\mathfrak R_n^{\rm WBS}<\Delta_n/12.
\end{equation*}
If $b<b^\star$, then \eqref{eq:wbs-isolation-margins} gives
\begin{equation*}
 e_{b,n}^{\rm iso}
 \le\tau_{b+1,n}-\Delta_n/3+1
 \le\tau_{b^\star,n}-\Delta_n/3+1<k^\star.
\end{equation*}
If $b>b^\star$, similarly,
\begin{equation*}
 s_{b,n}^{\rm iso}
 \ge\tau_{b-1,n}+\Delta_n/3-1
 \ge\tau_{b^\star,n}+\Delta_n/3-1>k^\star.
\end{equation*}
Thus each unselected isolating interval remains in the correct child.

The new change $\tau_{b^\star,n}$ lies within
$\mathfrak R_n^{\rm WBS}$ of the new boundary $k^\star$.  Let
$\tau_{j,n}$ be a previously selected change that remains in one of the two
children.  If
$\tau_{j,n}-s\le\mathfrak R_n^{\rm WBS}$ and $k^\star<\tau_{j,n}$, then
\begin{equation*}
 |k^\star-\tau_{b^\star,n}|<\Delta_n/2
 \quad\Longrightarrow\quad
 \tau_{b^\star,n}<\tau_{j,n}.
\end{equation*}
The spacing and \eqref{eq:wbs-selected-localization} then imply
\begin{equation*}
 k^\star
 \le\tau_{b^\star,n}+\mathfrak R_n^{\rm WBS}
 \le\tau_{j,n}-\Delta_n+\mathfrak R_n^{\rm WBS}
 <\tau_{j,n}-\mathfrak R_n^{\rm WBS}\le s,
\end{equation*}
a contradiction.  Hence
\begin{equation*}
 \tau_{j,n}-s\le\mathfrak R_n^{\rm WBS}
 \quad\Longrightarrow\quad k^\star>\tau_{j,n},
\end{equation*}
and the child containing $\tau_{j,n}$ retains the boundary $s$.  The symmetric
argument gives
\begin{equation*}
 e-\tau_{j,n}\le\mathfrak R_n^{\rm WBS}
 \quad\Longrightarrow\quad k^\star<\tau_{j,n},
\end{equation*}
so that child retains $e$.  Both invariants are preserved.

The recursion now has the deterministic implication
\begin{equation*}
 \mathcal U_n(s,e)\ne\varnothing
 \quad\Longrightarrow\quad
 \max_{r_{\rm W}\in\mathcal M_n^{\rm WBS}(s,e)}
 \cS_{I_{r_{\rm W}}^{\rm WBS}}^{\max}>\Lambda_n.
\end{equation*}
Every successful split selects exactly one previously unselected change, and
\begin{equation*}
 |\mathcal U_n(s,k^\star)|+|\mathcal U_n(k^\star,e)|
 =|\mathcal U_n(s,e)|-1.
\end{equation*}
The first invariant preserves the isolating interval of every change counted
on the left.  Recursion from $(0,n]$ therefore produces exactly
$J_{\rm cp}$ successful splits.

After the $J_{\rm cp}$th successful split, every recursive node has
$\mathcal U_n(s,e)=\varnothing$.  Equation~\eqref{eq:wbs-empty-node-stops}
therefore excludes any additional split.  Hence, on
$\mathcal G_n^{\rm WBS}$,
\begin{equation}
 \widehat J_{\rm cp}=\widetilde J_{\rm cp}=J_{\rm cp},
 \qquad
 \max_{1\le b\le J_{\rm cp}}
 |\widetilde\tau_{b,n}-\tau_{b,n}|
 \le\mathfrak R_n^{\rm WBS}.
 \label{eq:wbs-preliminary-bound-on-good-event}
\end{equation}
Furthermore,
\begin{equation*}
 \widetilde\tau_{b+1,n}-\widetilde\tau_{b,n}
 \ge\Delta_n-2\mathfrak R_n^{\rm WBS}>0,
\end{equation*}
so the ordered preliminary estimates are matched to the ordered true changes.
Combining \eqref{eq:wbs-main-high-probability-event},
\eqref{eq:wbs-random-radius-orders}, and
\eqref{eq:wbs-preliminary-bound-on-good-event} proves
\eqref{eq:preliminary-exact-recovery}.

Set $\widetilde\tau_{0,n}=\tau_{0,n}=0$ and
$\widetilde\tau_{J_{\rm cp}+1,n}=\tau_{J_{\rm cp}+1,n}=n$.  On the event in
\eqref{eq:wbs-preliminary-bound-on-good-event},
\begin{equation*}
 \max_{0\le b\le J_{\rm cp}+1}
 |\widetilde\tau_{b,n}-\tau_{b,n}|
 \le\mathfrak R_n^{\rm WBS}<\Delta_n/8.
\end{equation*}
Using $\lfloor x\rfloor\in[x-1,x]$ in the midpoint definitions yields
\begin{align*}
 s_{b,n}^{\rm ref}-\tau_{b-1,n}
 &\ge\frac{\tau_{b,n}-\tau_{b-1,n}}2
 -\mathfrak R_n^{\rm WBS}-1,\\
 \tau_{b,n}-s_{b,n}^{\rm ref}
 &\ge\frac{\tau_{b,n}-\tau_{b-1,n}}2
 -\mathfrak R_n^{\rm WBS}-1,\\
 e_{b,n}^{\rm ref}-\tau_{b,n}
 &\ge\frac{\tau_{b+1,n}-\tau_{b,n}}2
 -\mathfrak R_n^{\rm WBS}-1,\\
 \tau_{b+1,n}-e_{b,n}^{\rm ref}
 &\ge\frac{\tau_{b+1,n}-\tau_{b,n}}2
 -\mathfrak R_n^{\rm WBS}-1.
\end{align*}
All four quantities exceed $\Delta_n/3$ for large $n$.  Therefore
$I_{b,n}^{\rm ref}$ contains exactly $\tau_{b,n}$ and, for constants
$c,C,\varepsilon_{\rm ref}>0$,
\begin{equation}
 c\Delta_n\le|I_{b,n}^{\rm ref}|\le C\Delta_n,
 \qquad
 \theta_{b,n}^{\rm ref}
 =\frac{\tau_{b,n}-s_{b,n}^{\rm ref}}
 {|I_{b,n}^{\rm ref}|}
 \in[\varepsilon_{\rm ref},1-\varepsilon_{\rm ref}].
 \label{eq:wbs-refinement-interior-fraction}
\end{equation}

For $I=I_{b,n}^{\rm ref}$, the one-change formula gives
\begin{equation*}
 \cS_I^{\rm sig}(k)
 =|I|\mathfrak h_{\rm cp}(x_{I,k},\theta_{b,n}^{\rm ref})
 \widehat{\mathfrak q}_{n}^{\max}(d_{b,n}).
\end{equation*}
Since $\widehat\tau_{b,n}$ maximizes $\cS_I$,
\begin{align*}
 &|I|\left\{
 \mathfrak h_{\rm cp}(\theta_{b,n}^{\rm ref},
 \theta_{b,n}^{\rm ref})
 -\mathfrak h_{\rm cp}(x_{I,\widehat\tau_{b,n}},
 \theta_{b,n}^{\rm ref})\right\}
 \widehat{\mathfrak q}_{n}^{\max}(d_{b,n})\\
 &\qquad\le
 C\{1+(\mathfrak Z_n^{\rm WBS})^2+\sqrt{|I|}\norm{d_{b,n}}\mathfrak Z_n^{\rm WBS}\}.
\end{align*}
Using \eqref{eq:local-ridge-signal-equivalence},
\eqref{eq:local-shape-linear-gap}, and
\eqref{eq:wbs-refinement-interior-fraction},
\begin{equation*}
 \norm{d_{b,n}}^2
 |\widehat\tau_{b,n}-\tau_{b,n}|
 \le C\{1+(\mathfrak Z_n^{\rm WBS})^2
 +\sqrt{\Delta_n}\norm{d_{b,n}}\mathfrak Z_n^{\rm WBS}\}.
\end{equation*}
The number of changes is fixed.  Hence
\eqref{eq:local-score-order} and
\eqref{eq:retained-jump-lower-bound} imply
\eqref{eq:multiple-localization-rate}.  Finally,
\begin{equation*}
 \frac{(\mathfrak v_n^{\rm WBS})^2}{\Delta_n}
 +\frac{\mathfrak v_n^{\rm WBS}}{\sqrt{\Delta_n}}
 \le C\left\{
 n^{2/\nu-1}(\log n)^{2/\nu}
 +n^{1/\nu-1/2}(\log n)^{1/\nu}\right\}
 \longrightarrow0,
\end{equation*}
so the refined error is $o_{\Pp}(\Delta_n)$.  Therefore
\begin{equation*}
 \widehat\tau_{b+1,n}-\widehat\tau_{b,n}
 \ge\Delta_n-2\max_j
 |\widehat\tau_{j,n}-\tau_{j,n}|>0
\end{equation*}
with probability tending to one, which proves the final assertion.

\subsection{Proof of Corollary~\ref{cor:multiple-fixed-jumps}}

Lemma~\ref{prop:quadrature} gives
$\mathfrak a_{p,N}(B_\delta)\to0$.  Since
$\Delta_n/n\ge c_\Delta$,
\begin{equation*}
 \pi_n^{\rm WBS}
 \le J_{\rm cp}\exp\{-c_{\rm W}c_\Delta^2M_n^{\rm WBS}\}
 \longrightarrow0
\end{equation*}
when $M_n^{\rm WBS}=\lceil C_M\log n\rceil$.  Moreover,
\begin{equation*}
 \mathfrak v_n^{\rm WBS}\sqrt n
 =n^{1/2+1/\nu}(\log n)^{1/\nu}
 =o(n^{\gamma_{\rm W}}),
 \qquad
 n^{\gamma_{\rm W}}
 =o(\Delta_n\delta_{\min,n}^2)
\end{equation*}
for $\gamma_{\rm W}\in(1/2+1/\nu,1)$.  The conditions of
Theorem~\ref{thm:multiple-consistency} follow.

\bibliographystyle{chicago}
{\footnotesize\bibliography{references}}

\end{document}